%% file: 4-conn.tex
\input{preamble}

\newcommand{\sidenote}[1]{\textbf{(*)}\marginpar {\tiny \raggedright{(*) #1}}}
\newcommand{\ignore}[1]{}
\renewcommand{\sidenote}[1]{}

\newtheorem{remark}[theorem]{Remark}
\usepackage{enumitem}

\usepackage{color,soul}
\definecolor{cblue}{rgb}{0.36, 0.54, 0.66}

\begin{document}


\title{Computing the $(k+2)$-Edge-Connected Components in $k$-Edge-Connected Digraphs in Subquadratic Time}

\maketitle

\begin{abstract}
Computing edge-connected components in directed and undirected graphs is a fundamental and well-studied problem in graph algorithms.
In a very recent breakthrough, Korhonen [STOC 2025] showed that for any fixed $k$, the $k$-edge connected components of an undirected graph can be computed in linear time. 
In contrast, the directed case remains significantly more challenging: linear-time algorithms are only known for $k \le 3$, and for any fixed $k > 3$, the best known bound for sparse or moderately dense graphs
is still the $O(mn)$-time algorithm of Nagamochi and Watanabe (1993).

In this paper, we break the $O(mn)$ barrier for all $k = o(n^{1/4}/\sqrt{\log{n}})$. 
We present a randomized algorithm that computes the $(k+2)$-edge-connected components of a $k$-edge-connected directed graph in $O(k^2 m \sqrt{n} \log n)$ time, for any~$k$. 
This constitutes the first improvement over the classic Nagamochi--Watanabe bound for any constant $k > 3$. 
Our approach introduces new structural insights into directed edge-cuts and combines these with both new and existing techniques. 
A central contribution of our work is a substantial simplification and generalization of the framework introduced in~\cite{GKPP:3ECC}, which achieved an $\widetilde{O}(m\sqrt{m})$ bound for computing the $3$-edge-connected components of a digraph. 
In addition, we develop a variant of our algorithm that achieves the same $O(m \sqrt{n} \log n)$ running time for computing the $4$-edge-connected components of a \emph{general} directed graph.
\end{abstract}

\thispagestyle{empty}
\clearpage

\tableofcontents

\thispagestyle{empty}
\clearpage

\setcounter{page}{1}

\input{introduction}

\input{technical}

\input{text}

\input{decompositions}

\input{localSearchDeterministic}

\input{extending_k_plus_3}

\bibliographystyle{acm}
\bibliography{references}




\end{document}

%% file: preamble.tex
\documentclass[a4paper,11pt]{article}
\usepackage[utf8]{inputenc}
\usepackage{amssymb,amsmath,amsthm,amsfonts}
\usepackage{wrapfig}
\usepackage{graphicx}
\usepackage[lined, algoruled, linesnumbered]{algorithm2e}

\usepackage{fullpage}
\usepackage{epsfig}
\usepackage{xspace}
\usepackage[disable]{todonotes}
\usepackage[nolabel]{showlabels}
\usepackage{thmtools}
\usepackage{thm-restate}
\usepackage[hidelinks,pagebackref=true]{hyperref}
\usepackage[nameinlink]{cleveref}

\definecolor{winered}{rgb}{0.5,0,0}

\hypersetup{
    colorlinks = true,
    linkcolor = blue,
    citecolor = winered 
}

\newtheorem{theorem}{Theorem}
\newtheorem{definition}[theorem]{Definition}
\newtheorem{lemma}[theorem]{Lemma}
\newtheorem{corollary}[theorem]{Corollary}
\newtheorem{claim}[theorem]{Claim}
\newtheorem{property}[theorem]{Property}
\newtheorem{observation}[theorem]{Observation}

\newtheorem{proposition}[theorem]{Proposition}

\newcommand{\cut}[1]{$#1$-cut}
\newcommand{\conn}[1]{$#1$-edge-connected}

\newcommand{\incomp}[1]{$#1$-in set}

\newcommand{\outcomp}[1]{$#1$-out set}
\newcommand{\set}[1]{\{#1\}}

\newcommand{\vin}[2][]{vol_{in#1}(#2)}
\newcommand{\vout}[2][]{vol_{out#1}(#2)}

\newcommand{\kec}{\ensuremath{\leftrightarrow_{\mathit{k}}}}

\def \Ot {\ensuremath{\widetilde{O}}}

\newcommand{\outgadget}[1]{$\mathcal{G}_{out}(#1)$}
\newcommand{\ingadget}[1]{$\mathcal{G}_{in}(#1)$}

\newcommand{\intermediate}{Partition-Witness}

\newcommand{\lightgraph}[1]{$#1$CLG}

\newcommand{\gadgetsube}[1]{$G'_{<#1>}$}
\newcommand{\gadgetsubeM}[1]{G'_{<#1>}}

\newcommand{\contract}[1]{\ensuremath{G_{\langle #1 \rangle}}}
\newcommand{\splitop}[2]{\ensuremath{\mathit{split}(#1,#2)}}
\newcommand{\contractp}[1]{\ensuremath{G'_{\langle #1 \rangle}}}

\newcommand{\nnote}[1]{\todo[color=orange!25!white]{Nikos: #1}\xspace}
\newcommand{\cnote}[1]{\todo[color=orange!05!white]{Charis: #1}\xspace}
\newcommand{\lnote}[1]{\todo[color=orange!05!white]{Loukas: #1}\xspace}

\newcommand{\ecc}[1]{$#1$-ECC}

\newcommand{\mout}[1]{$\mu(#1)$}
\newcommand{\moutM}[1]{\mu(#1)}

\newcommand{\moutR}[1]{$\mu_R(#1)$}
\newcommand{\moutRM}[1]{\mu_R(#1)}

\newcommand{\mpout}[1]{$\mu'(#1)$}
\newcommand{\mpoutM}[1]{\mu'(#1)}

\author{
  Loukas Georgiadis\thanks{Department of Computer Science \& Engineering, University of Ioannina, Greece.  E-mail:\texttt{loukas@uoi.gr}}
  \and
  Evangelos Kipouridis\thanks{Max Planck Institute for Informatics, Saarland Informatics Campus, Germany. Email: \texttt{kipouridis@mpi-inf.mpg.de}}
  \and
  Evangelos Kosinas\thanks{Department of Computer Science \& Engineering, University of Ioannina, Greece. Email: \texttt{ekosinas@cs.uoi.gr}}
  \and
  Charis Papadopoulos\thanks{Department of Mathematics, University of Ioannina, Greece.  E-mail:\texttt{charis@uoi.gr}} 
  \and
  Nikos Parotsidis\thanks{Google Research, Switzerland. Email: \texttt{nikosp@google.com}}
}

\date{}

%% file: introduction.tex
\section{Introduction}

Computing edge-connected components in directed and undirected graphs is a fundamental problem in graph algorithms, motivated by several applications (see, e.g.,~\cite{connectivity:nagamochi-ibaraki}), and has been extensively studied. Before discussing related work and our contributions, we begin with the necessary definitions.

Let $G = (V, E)$ be a strongly connected directed graph (digraph) with $n$ vertices and $m$ edges. 
A set of edges $C \subseteq E$ is called a \emph{cut} if the graph $G \setminus C$ is not strongly connected. 
If $|C| = k$, we refer to $C$ as a \cut{k} of $G$. 
For any pair of vertices $v, w \in V$, we say that a cut $C$ \emph{separates $v$ from $w$} if there is no $v$–$w$ path in $G \setminus C$; in this case, we call $C$ a \emph{$(v,w)$-cut}. 
We denote by $\lambda(v, w)$ the size of a minimum $(v,w)$-cut. 
A digraph $G$ is said to be \emph{$k$-edge-connected} if it has no cuts of size at most $k-1$, that is, if $\lambda(v, w) \ge k$ for all pairs of vertices $v, w \in V$.

\begin{wrapfigure}{r}{0.4\textwidth}  
  \centering
  \vspace{-10pt} 
  \includegraphics[width=.5\textwidth,
                   trim={10cm 0cm 8cm 10cm},clip]{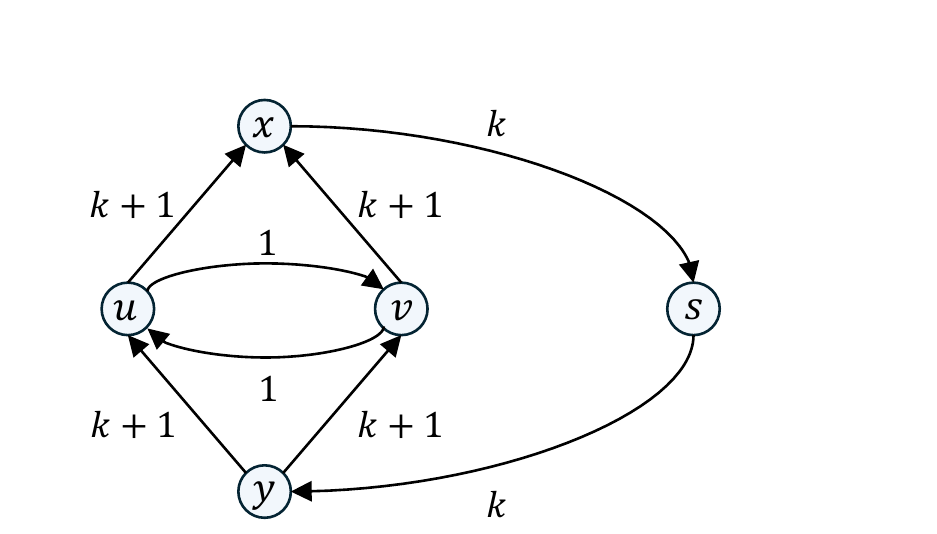}
  \vspace{-10pt}
  \caption{A $k$-edge-connected digraph. 
  The weight on each edge corresponds to its multiplicity. 
  Vertices $u$ and $v$ are $(k+1)$-edge-connected but not $(k+2)$-edge-connected 
  (e.g., $(v,u)$ and the $k$ edges $(x,s)$ form a $(k+1)$-cut separating $v$ from $u$).}
  \label{figure:small-example}
\end{wrapfigure}

We say that two vertices $v$ and $w$ are \emph{$k$-edge-connected}, denoted by $v \kec w$, if there exist $k$ edge-disjoint directed paths from $v$ to $w$ and $k$ edge-disjoint directed paths from $w$ to $v$. See \Cref{figure:small-example}.
(Note that a path from $v$ to $w$ and a path from $w$ to $v$ need not be edge-disjoint.)
By Menger's theorem~\cite{menger}, $v \kec w$ if and only if the removal of any set of at most $k-1$ edges leaves $v$ and $w$ in the same strongly connected component. 
A \emph{$k$-edge-connected component} of a digraph $G = (V, E)$ is a maximal subset $U \subseteq V$ such that $u \kec v$ for all $u, v \in U$. 
The $k$-edge-connected components of $G$ form a partition of $V$, since the relation $\kec$ is an equivalence relation~\cite{georgiadis20162edge}. 
The same definition applies analogously to undirected graphs.

For undirected graphs, the problem of computing the $k$-edge-connected components is very well understood. Until recently, linear-time algorithms were known only for $k \leq 5$~\cite{pathdfs:g00,GI:ECtoVC,Linear4ECC:GIK,3-connectivity:ht,Kosinas2024,Linear4ECC:Nadara,NagamochiIbaraki:3CC,dfs:t,Tsin:3CC}. In a significant recent breakthrough, Korhonen~\cite{LinearkECC:STOC25} introduced an algorithm with running time $O(m+k^{O(k^2)}n)$ for computing the $k$-edge-connected components of an undirected graph, which is linear in the number of edges for any fixed $k$.

The directed case remains more challenging. Until recently, linear-time algorithms were known only for $k \le 2$~\cite{dfs:t,GIP20:SICOMP}.
GKPP~\cite{GKPP:3ECC} presented a randomized (Monte-Carlo) algorithm that computes the $3$-edge-connected components of a digraph with $m$ edges in $\widetilde{O}(m^{3/2})$ time.\footnote{The $\widetilde{O}(\cdot)$ notation hides polylogarithmic factors.}
Their algorithm applies an involved extension of the framework of~\cite{forster2020computing,nanongkai2019breaking} for deciding whether a digraph is $(k+1)$-edge-connected.
It is based on a local search procedure~\cite{chechik2017faster,forster2020computing}, initiated from sampled vertices, which identifies \outcomp{2}s or \incomp{2}s,
that is, vertex sets $S$ that have 
$2$ outgoing edges to $V \setminus S$ or 
$2$ entering edges from $V \setminus S$.
After finding such a set $S$, \cite{GKPP:3ECC} applies an efficient graph operation for replacing $S$ with a gadget of small size that preserves the pairwise connectivity among the vertices of $V \setminus S$.
Georgiadis, Italiano, and Kosinas~\cite{Linear3ECC} later presented a deterministic linear-time algorithm for computing the $3$-edge-connected components of a digraph. 
This result is based on an adaptation of the minset-poset technique of Gabow~\cite{Gabow:Poset:TALG} for identifying minimal \incomp{k}s in $(k-1)$-edge-connected digraphs, combined with a characterization of \incomp{2}s using the concept of strongly divergent spanning trees~\cite{DomCert:TALG}.
%

For any fixed $k > 3$, the best known bound for computing the $k$-edge-connected components 
of a sparse or moderately dense digraph
remains $O(mn)$, as given by the algorithm of Nagamochi and Watanabe~\cite{nagamochi1993computing}.
%
Unfortunately, both approaches proposed in \cite{GKPP:3ECC} and \cite{Linear3ECC} appear difficult to extend to the case $k>3$. Specifically, \cite{GKPP:3ECC} applies a graph transformation that replaces a \incomp{2} or a \outcomp{2} $S$ with a small gadget that preserves the pairwise connectivity of the vertices in $V \setminus S$. The fast construction of this gadget relies crucially on the simple structure of \cut{1}s. 
In contrast, \cite{Linear3ECC} computes minimal \incomp{2} sets by exploiting properties of two strongly divergent spanning trees rooted at a designated vertex $s$. 
Extending this technique to $k>3$ would require, at minimum, an efficient method for computing $k-1$ spanning trees with similar properties. However, Huck~\cite{Huck:Disproof} showed that such spanning trees do not always exist: there are $k$-vertex-connected digraphs where, for a fixed start vertex $s$, do not have such collection of spanning trees.

\subsection{Our results}

In this paper, we present the first improvement over the classical bound of Nagamochi and Watanabe~\cite{nagamochi1993computing} for any constant $k>3$ in $k$-edge-connected directed graphs. 
Specifically, we develop a novel randomized algorithm that computes the $(k+2)$-edge-connected components of a $k$-edge-connected digraph in $O(k^2 m \sqrt{n} \log n)$ time, for any $k$. Hence, we break the $O(mn)$ barrier for $k = o(n^{1/4}/\sqrt{\log{n}})$. 
Our first main result is summarized in the following theorem. 

\begin{theorem}
\label{theorem:keccs}
Let $G$ be a $k$-edge-connected digraph, 
and let $\delta$ be a parameter with $0 < \delta < 1$. 
There exists a randomized algorithm that runs in $O(k^2 m \sqrt{n} \log(n/\delta))$ time and outputs a partition $\mathcal{P}$ of $V(G)$ satisfying the following guarantees:
\begin{itemize}
    \item Every two $(k+2)$-edge-connected vertices of $G$ belong to the same set of $\mathcal{P}$.
    \item With probability at least $1 - \delta$, every two vertices of $G$ that are not $(k+2)$-edge-connected are separated by $\mathcal{P}$.
\end{itemize}
\end{theorem}

Hence, we achieve a significant improvement over the $O(mn)$ bound of Nagamochi and Watanabe for all graph densities.
Our approach introduces new structural insights into directed edge-cuts and combines these with both new and existing techniques. 
A key contribution of our work is a substantial simplification and generalization of the framework proposed in \cite{GKPP:3ECC}, which achieved an $\widetilde{O}(m\sqrt{m}\,)$ bound for computing the $3$-edge-connected components of a digraph.
We also give an extension of our algorithm that computes the $(k+3)$-edge-connected components of a given $(k+2)$-edge-connected component $S$ of $G$ in $O(k^3 m \sqrt{|S|} \log(|S|/\delta))$ time.\lnote{Added this sentence.}

Furthermore, we develop a variant of our algorithm that achieves 
$O(m \sqrt{n} \log n)$ running time for computing the $4$-edge-connected components of a \emph{general} directed graph.

\begin{theorem}
\label{theorem:4eccs}
Let $G$ be a digraph, and let $\delta$ be a parameter with $0 < \delta < 1$. 
There exists a randomized algorithm that runs in $O(m \sqrt{n} \log(n/\delta))$ time and outputs a partition $\mathcal{P}$ of $V(G)$ satisfying the following guarantees:
\begin{itemize}
    \item Every two $4$-edge-connected vertices of $G$ belong to the same set of $\mathcal{P}$.
    \item With probability at least $1 - \delta$, every two vertices of $G$ that are not $4$-edge-connected are separated by $\mathcal{P}$.
\end{itemize}
\end{theorem}

Note that \Cref{theorem:4eccs} is not subsumed by \Cref{theorem:keccs}, as we do not require the input digraph $G$ to be $2$-edge-connected. We also remark that all our results apply to multigraphs, that is, the input graph may contain parallel directed edges.

\ignore{
Unlike the algorithm of \cite{GKPP:3ECC} for computing the $3$-edge-connected components of a digraph, 
our approach achieves a significant improvement over the $O(mn)$ bound of Nagamochi and Watanabe~\cite{nagamochi1993computing} for all graph densities.
our approach achieves a significant improvement over the $O(mn)$ bound of Nagamochi and Watanabe~\cite{nagamochi1993computing} for all graph densities.
}

Our algorithm loosely follows the framework of GKPP~\cite{GKPP:3ECC} in that it repeatedly performs local searches from sampled vertices. However, it replaces certain bottleneck computations of~\cite{GKPP:3ECC} 
by exploiting new structural insights into directed edge-cuts and combines both new and existing techniques, as detailed below.

\ignore{
A fundamental difficulty in computing $k$-edge-connected components for $k > 3$, which does not arise in the case of $3$-edge-connected components, is that vertices lying in different $(k-1)$-edge-connected components can be separated by cuts of size at most $k-2$, which may exhibit complex structure. In contrast, for $k = 3$, vertices that are not $2$-edge-connected are separated by $1$-cuts, which admit a simple, tree-like representation~\cite{GIP20:SICOMP}. This crucial property is exploited in~\cite{GKPP:3ECC} to enable an efficient gadget substitution operation.\lnote{I think we should remove this paragraph now, or rephrase it somehow.}
\cnote{yes, I think we should remove it.}
}

\subsection{Related work}

Connectivity-related problems are known to be  
much more difficult in directed graphs than in undirected graphs (see, e.g., \cite{Gabow:Poset:TALG,LocalFlow:EC,KT:EC}). 
Cuts in undirected graphs can be represented compactly by the \emph{Gomory-Hu tree} (or \emph{cut tree}), which can be used to identify 
the $k$-edge-connected components of undirected graphs, for any $k$. Furthermore, many efficient algorithms for computing Gomory-Hu trees are available (see e.g., \cite{abboud2021apmf,abboud2021subcubic,APMF2023,CKLPGS:AlmostLinearMaxFlow,hariharan2007efficient,li2022nearly}). 
An $\widetilde{O}(m+nk^3)$-time algorithm for computing the restricted version of the Gomory-Hu tree, that captures edge-cuts of size at most $k$, was presented by Hariharan, Kavitha, and Panigrahi~\cite{hariharan2007efficient}.
In a breakthrough result, Korhonen~\cite{LinearkECC:STOC25} showed how to compute the $k$-edge-connected components of an undirected graph in $O(m+k^{O(k^2)}n)$ time.

On the contrary, in directed graphs, edge cuts have a more complicated structure, and it was proved by Bencz\'ur~\cite{Benczur:CutTrees} that in this case, cut trees do not even exist.
Linear-time algorithms for finding the $2$-edge-connected components are given in \cite{georgiadis20162edge,GIP20:SICOMP}, and for finding the $3$-edge-connected components in \cite{Linear3ECC}.
Schnorr~\cite{schnorr1979bottlenecks} showed how to compute the minimum of $\mathit{flow}(u,v)$ and $\mathit{flow}(v,u)$, for all pairs of vertices $u$ and $v$ of a weighted digraph, using $O(n \log n)$ max-flow computations.
By computing \emph{$k$-bounded max-flow} (i.e., the value $\min\{k, \mathit{flow}(u,v)\}$) using $k$ iterations of Ford-Fulkerson, this approach gives an $O(mnk\log n)$-time algorithm for computing the $k$-edge-connected components of a digraph. 
If we use the $O(m^{1+o(1)})$ max-flow algorithm of CKLPGS~\cite{CKLPGS:AlmostLinearMaxFlow}, we obtain an $O(m^{1+o(1)}n)$-time algorithm, for any $k$.
Nagamochi and Watanabe~\cite{nagamochi1993computing} presented an algorithm that computes the
$k$-edge-connected components of a digraph in $O(mn\min\{k,n,\sqrt{m}\})$  time. This remains the best bound for fixed $k>4$ and for sparse or moderately dense graphs.
Cheung, Lau, and Leung~\cite{boundedEC:SICOMP} presented an algorithm that computes the edge connectivity between all pairs of vertices (i.e., $\lambda(u,v)$ for all vertex pairs $u,v$) in $O(m^{\omega})$ time, where $\omega < 2.371339$ is the matrix multiplication exponent~\cite{FastMatrixMultiplication:SODA25}.  
Subsequently, Akmal and Jin~\cite{akmal_et_al:LIPIcs.ICALP.2023.11} developed an $\widetilde{O}((kn)^{\omega})$-time algorithm for computing the \emph{$k$-bounded edge connectivities} between all vertex pairs, that is, the values $\min\{k,\lambda(u,v)\}$ for all $u,v \in V$ (see also~\cite{boundedEC:SOSA24}).  
Note that knowing all $k$-bounded edge connectivities is sufficient to determine the $i$-edge-connected components of the graph for every $i \le k$.\lnote{Added references for all-pair connectivities.} 

Very recently, Hoppenworth, Saranurak, and Wang~\cite{FTSCC:FOCS25} established near-optimal bounds for a \emph{$k$-connectivity preserver} of a directed graph $G$, that is, a subgraph $H$ of $G$ such that, for all $i \le k$, the $i$-edge-connected components of $G$ and $H$ coincide. 
They showed that any digraph with $n$ vertices admits a $k$-connectivity preserver with $O(k 4^{k} n \log n)$ edges. 
Hence, a fast algorithm for constructing such a preserver would directly improve the running time bounds for computing $k$-edge-connected components. 
Still, the fast construction of these preservers remains an important open problem.

\subsection{Overview of our techniques}

A \outcomp{k} $S$ is a set of vertices of $G$ such that there are $k$ edges from $S$ to $V \setminus S$. \sidenote{Loukas: Defined $k$-out set.}
Following \cite{GKPP:3ECC}, we base our algorithm on the notion of \emph{minimal \outcomp{(k+1)}s}. 
Let $G$ be the input digraph, and let $U \subseteq V(G)$ be a set of $(k+1)$-edge-connected vertices. 
We denote by $\lambda(x,y)$ the size of the minimum $(x,y)$-cut in $G$.
Let $s \in U$ be a fixed vertex of $G$. Then, for any $v \in U$, we have $\lambda(v,s) \ge k+1$.
If there is a \outcomp{(k+1)} $S$ with $v\in S$ and $s\notin S$, then we let $M(v)$ denote the (inclusion-wise) minimum such $(k+1)$-out set. Otherwise, we let $M(v)=\bot$. 
Then, for any vertex $v \in U$ that is not $(k+2)$-edge-connected with $s$, the $(k+2)$-edge-connected component containing $v$ is a subset of $M(v)$. 
We use ``$M_R$'' to denote the same concept as ``$M$'' in the reverse graph $G^R$.
We show that these $M$-sets are sufficient to determine the $(k+2)$-edge-connectivity relation for the vertices in $U$. 
Moreover, if the $M$-set of a vertex exists, we show how to compute it in time proportional to its volume using a variant of the local search procedure of \cite{chechik2017faster}.

We emphasize that the condition $\lambda(v,s) \ge k+1$ for every $v \in U$ is essential in order to correctly partition $U$ into its $(k+2)$-edge-connected components. 
Consider, for example, \Cref{figure:small-example}. 
Suppose we wish to determine whether vertices $u$ and $v$ are $(k+2)$-edge-connected by examining their minimum cuts to a reference vertex $s$. 
In this case, the minimum $(u,s)$-cut and the minimum $(v,s)$-cut coincide, which implies that considering these cuts alone is insufficient to distinguish the connectivity between $u$ and $v$. 
Therefore, when refining a $(k+1)$-edge-connected component $S$ into its constituent $(k+2)$-edge-connected components, we must carefully select a reference vertex $s \in S$. 
We address this issue by introducing the following decomposition.\lnote{Added this paragraph.}  

\paragraph{$(k+1)$-edge-connected-component decomposition.} We begin by computing, a decomposition of the input digraph $G$, into a collection $H_1,\dots,H_t$ of directed graphs that preserve the $(k+2)$-edge-connected components of $G$. 
Each graph $H_i$ contains two types of vertices, ordinary and auxiliary, where the ordinary vertices in $H_i$ form a $(k+1)$-edge-connected component of $G$. 
Moreover, any two vertices $u$ and $v$ of $G$ are $(k+2)$-edge-connected if and only if $u$ and $v$ are $(k+2)$-edge-connected ordinary vertices of a graph $H_i$.
%

Here, we present a randomized $\Ot(k^2 m)$-time algorithm that computes the $(k+1)$-edge-connected-component (\ecc{(k+1)}) decomposition of a $k$-edge-connected digraph, for any value of $k$.  
In addition, we provide a deterministic near-linear-time algorithm that computes the $3$-edge-connected-component (\ecc{3}) decomposition of a general digraph.  
In both decompositions, the total size of the resulting graphs $H_1, \dots, H_t$ is $O(m)$.
To achieve this, we exploit the computation of the nodes corresponding to a poset representation of the minimal \outcomp{k}s~\cite{Gabow:Poset:TALG,Linear3ECC} combined with a contraction operation for the \ecc{(k+1)} decomposition and the gadget substitution operation of \cite{GKPP:3ECC} for the \ecc{3} decomposition. Now, our goal is to partition the ordinary vertices in each graph $H_i$ into $(k+2)$-edge-connected components.

%
%

\paragraph{Local search.} A \outcomp{(\le k)} $S$ is a set of vertices of $G$ such that there are at most $k$ edges from $S$ to $V \setminus S$.
The \emph{volume} of a vertex set $S$, denoted by $\mathit{vol}(S)$, is the number of edges $(v,w)$ whose tail $v$ is in $S$.
CHILP~\cite{chechik2017faster} considered the problem of computing a \outcomp{(\le k)} with the following restrictions:
Given a graph $G$, a vertex $v$, and two integers $\Delta$ and $k$ such that $\Delta \geq k\geq 1$, the goal is to identify a \outcomp{(\le k)} $S$ such that $v\in S$, all vertices in $S \setminus v$ are reachable from $v$ in $G[S]$, and $vol(S) \leq \Delta$, or to conclude that no such set $S$ exists.
CHILP~\cite{chechik2017faster} described a deterministic algorithm, $\mathtt{LocalSearch}(G,v,k,\Delta)$, with $O(\Delta \cdot k^{2k})$ running time, for the following relaxation of the problem:
(i) if there exists a \outcomp{(\le k)} $S$ such that $v\in S$, all vertices in $S\setminus v$ are reachable from $v$ in $G[S]$, and $vol(S) \leq \Delta$, then $\mathtt{LocalSearch}(G,v,k,\Delta)$ returns a \outcomp{(\le k)} $S'$ such that $v\in S'$, all vertices in $S' \setminus v$ are reachable from $v$ in $G[S']$, and $vol(S')\leq (2k)^{k+2}\Delta$, and (ii) otherwise it concludes that no such set $S$ exists.
Note that if a set $S$ satisfying the constraints exists, then in this relaxed version, procedure $\mathtt{LocalSearch}(G,v,k,\Delta)$ returns a set $S'$ which satisfies all constraints except the bound on the volume of $S'$, which can be larger by a factor of $(2k)^{k+2}$.
Later on, FNYSY~\cite{forster2020computing} presented a randomized algorithm for the relaxed version of the problem with $\widetilde{O}(\Delta \cdot k^{2})$ running time. 

Notice that the algorithms of \cite{chechik2017faster,forster2020computing} are useful for $\Delta < m/k$, as otherwise the whole vertex set $V(G)$ may be returned.
Here, we introduce a slight variant of the local search procedure of CHILP~\cite{chechik2017faster}, that we refer to as 
$\mathtt{LocalSearchForMSet}(G,v,s,k,\Delta)$,
which guarantees that the returned set $S'$ does not contain $s$. 
We make use of the following two simple observations: (i) the additional condition, $s \not\in S'$, implies that we can use the local search procedure without any restriction on the volume $\Delta$, i.e., it is not necessary to have $\Delta < m/k$, and (ii) if $\lambda(v,s) = k$ and $\mathit{vol}(M(v)) \le \Delta$, then $\mathtt{LocalSearchForMSet}(G,v,s,k,\Delta)$ is guaranteed to return $S'=M(v)$.

\paragraph{Sampling for large $M$-sets.}

We say that an $M$-set $M(v)$ is ``small'' if its volume is at most $m/\sqrt{n}$, and ``large'' otherwise. Using our modified local search from each vertex $v \not= s$, we can afford to compute all small $M$-sets. In order to account for the partition of the ordinary vertices due to large $M$-sets, we sample the edges uniformly at random. 
Specifically, by sampling $\Ot(\sqrt{n})$ edges, we can guarantee that for every $M$-set $U$ with volume more than $m/\sqrt{n}$, we have sampled at least one edge $(v,w)$ whose tail $v$ is in $U$ with high probability. 

\paragraph{Good partitions.} Since we cannot afford to compute all large $M$-sets explicitly, we must instead find a way to exploit the set $M(v)$ associated with each sampled vertex $v$, i.e., the tail of a sampled edge $(v,w)$.\sidenote{Loukas: Explained what we mean by ``sampled vertex''.} 
To this end, we define a partition $\mathcal{P}$ of $V$ to be \emph{good for $v$} if: 
\begin{itemize}
\item{$\mathcal{P}$ maintains the $(k+2)$-edge-connected components, and}
\item{for every ordinary vertex $u$ such that $v\in M(u)$, the partition $\mathcal{P}$ separates any ordinary vertex $w \notin M(u)$ from $u$.}
\end{itemize}
We show that it suffices to compute a partition that is good for each sampled vertex $v$.
To construct such partitions efficiently, we distinguish two cases depending on whether $\lambda(v, s) > k$ or $\lambda(v, s) \le k$. 
In the former case, we exploit the directed acyclic graph (DAG) representation of all minimum cuts introduced by Picard and Queyranne~\cite{PicardQ82}; in the latter, we leverage structural properties of latest minimum cuts. (A \emph{latest minimum $(x,y)$-cut} is defined as an inclusion-wise maximum \outcomp{\lambda(x,y)} $S$ such that $x \in S$ and $y \notin S$.)\sidenote{Loukas: Defined latest minimum cut.}

\ignore{
\paragraph{Triggered search.}

To compute the partition of the ordinary vertices induced by the large $M$-sets, we iteratively perform local searches initiated from the sampled vertices $v$, which we refer to as \emph{triggered searches}. 
Each triggered search returns a sequence of \outcomp{3}s, which we use to compute the partition of the ordinary vertices. We note that GKPP~\cite{GKPP:3ECC} also employs a triggered search procedure; however, their method for processing the computed out-sets fundamentally differs from ours. 
Moreover, unlike~\cite{GKPP:3ECC}, we cannot guarantee that the vertex $v$ from which the local search is initiated is ordinary; in fact, an identified out-set may not contain any ordinary vertices at all.\lnote{Added this note. Please check!}\cnote{Looks fine!}
In our approach, we introduce a novel technique based on a new concept of \emph{\intermediate{} graphs}. GKPP~\cite{GKPP:3ECC}, on the other hand, use the gadget substitution operation, combined with recursion and backward searches. (See next subsection.)

\paragraph{\intermediate{} Graphs.}

One of the main novelties of our approach is the notion of \emph{\intermediate{} graphs} that we introduce, which completely eliminates the need for recursive calls.
An informal description of a \intermediate{} graph is the following: suppose $X,Y$ are both $3$-out sets with $X\subseteq Y$.
To build a \intermediate{} graph we start with the induced subgraph $G[Y]$ and contract $X$ in a single vertex $z$. 
Then we add the $3$ outgoing edges from $Y$ (which may add at most $3$ vertices, the exit points of $Y$), and add $3$ parallel edges from each exit point of $Y$ to each vertex in $(Y\setminus X) \cup \set{z}$.
We process the \intermediate{} graphs and exploit their properties to create a coarse partitioning of the vertices in $Y\setminus X$, in a way that allows us to identify the final $4$-edge-connected components.
\nnote{Added brief description of why we need them. Please check}
\todo[inline]{Kip: It's good. }

In more detail regarding how the \intermediate{} graphs are employed, we define a sequence of \intermediate{} graphs that corresponds to the sequence of \outcomp{3}s computed by a triggered search. The useful property of these \intermediate{} graphs is that we can compute the partition that is induced from the sequence of \outcomp{3}s, by computing the SCCs and $3$-edge-connected components after removing a few specific edges. Note that since we sample $\Ot(\sqrt{n})$ vertices, we can afford to spend $O(m)$ time to process the \intermediate{} graphs induced from each triggered search. 
}

\subsection{Comparison to previous work}

The linear-time algorithm of Georgiadis, Italiano, and Kosinas~\cite{Linear3ECC} relies on the concept of strongly divergent spanning trees \cite{DomCert:TALG}, which is defined only for pairs of trees. Extending this approach to compute \conn{k} components for $k > 3$ would require $k - 1$ spanning trees with analogous edge- and vertex-disjointness properties. However, the result of Huck~\cite{Huck:Disproof} implies that it may not be possible to find more than two spanning trees with these properties.

At a high level, our algorithm follows a framework similar to that of~\cite{GKPP:3ECC}, but with several fundamental differences, which we highlight below.

\paragraph{Improved running time.} The framework in \cite{GKPP:3ECC} required time $\widetilde{O}(m\sqrt{m})$; the $\sqrt{m}$ factor appears by balancing the need to hit each component of interest (via random sampling) multiple times, and the time spent on the subsequent backward searches (that depended on the number of times each component was hit). In our framework, we overcome the need for backward searches and achieve a significant improvement in the running time, namely $\widetilde{O}(m\sqrt{n})$ for any constant $k$. 

\paragraph{Lack of gadgets.} 
The algorithm of GKPP~\cite{GKPP:3ECC} repeatedly detects \outcomp{2}s and \incomp{2}s, which are then replaced by connectivity-preserving gadgets of sufficiently small size. These gadgets possess the crucial property of reducing the (out- or in-)volume of the replaced component by a constant factor. This property enables GKPP to perform $\Omega(n)$ local searches,
each with a worst-case running time of $\Theta(m)$, without incurring a total cost of $\Omega(mn)$. Whenever a sequence of local searches becomes expensive, the corresponding gadget replacement significantly reduces the size of the graph, thereby ensuring a smaller amortized cost. 
For $k \ge 4$, however, no analogous gadget construction is known that would provide similar guarantees, making it unclear how to afford the time required for $\Omega(n)$ local searches.

\ignore{
The algorithm of GKPP~\cite{GKPP:3ECC} repeatedly detects \outcomp{2}s and \incomp{2}s, which are replaced with connectivity-preserving gadgets of sufficiently small size. These gadgets have the crucial property of reducing the (out- or in-) volume of the replaced component by a constant factor. This allows GKPP to perform $\Omega(n)$ 
local searches\lnote{Replaced ``triggered searches'' with ``local searches'' because now we don't define the former yet.}, each with a worst-case running time of $\Theta(m)$, without incurring a total cost of $\Omega(mn)$. Whenever a sequence of local searches 
takes a long time, the corresponding gadget replacement significantly reduces the graph size, leading to a smaller amortized cost. For $k \ge 4$, however, we are not aware of any gadget construction with similar guarantees that would make it feasible to afford the time required for $\Omega(n)$ local searches.
}

\paragraph{No recursions.} The algorithm of GKPP~\cite{GKPP:3ECC} refines the current partition of the ordinary vertices after detecting an \outcomp{2} or an \incomp{2} $S$, by simply separating the vertices inside $S$ from those outside $S$. 
However, vertices inside $S$ that are not $3$-edge-connected can only be further separated through a recursive call that replaces $V \setminus S$ with a gadget. 
A key technical difficulty of this approach is that the same edge may participate in multiple recursive calls, making the analysis tedious and heavily dependent on properties specific to the case $k=3$, which do not extend to $k \ge 4$. 
In contrast, our new framework leverages the structural properties of minimal \outcomp{k}s, eliminating the need for recursion entirely.

\paragraph{No triggered or backward searches.} 
Another major complication in~\cite{GKPP:3ECC} arises from handling intersecting \outcomp{2}s. 
Consider a local search initiated from a sampled vertex $v$ that identifies a large \outcomp{2} $S$. 
There may exist several \outcomp{2}s $S_1, \ldots, S_k$ that intersect $S$, such that the ordinary vertices in different sets $S_i$ are not $3$-edge-connected. 
To separate these vertices, it is necessary to identify all such sets $S_1, \ldots, S_k$; however, this cannot be achieved through sampling alone, since $S_i \setminus S$ may be very small. 
To address this issue, GKPP~\cite{GKPP:3ECC} employ a sequence of \emph{triggered} and \emph{backward} searches, that is, local searches initiated from specific vertices of previously identified \outcomp{2} sets. 
Through a complex analysis, they show that it suffices to search for sets of volume $O(\sqrt{m})$, resulting in an overall running time of $\widetilde{O}(m \sqrt{m})$. 
In contrast, our framework completely avoids this complication by exploiting the structural properties of minimal \outcomp{k}s.

\paragraph{Simplicity.} Perhaps the two most technically involved contributions of GKPP~\cite{GKPP:3ECC} for the $k=3$ case were the gadget-substitution procedure and the backward searches. Our framework eliminates the need for both, giving an overall significantly simpler algorithm.

\paragraph{Extension for general $k$.} 
Although achieving an $o(mn)$-time algorithm for computing the $k$-edge-connected components of a general digraph, for any fixed constant $k > 4$, remains out of reach, we believe that our techniques provide a promising step toward this goal. 
The main obstacles to attaining such a bound within our framework are the lack of (i) a fast decomposition of a general digraph into $(k-1)$-edge-connected components, and (ii) an efficient algorithm for computing a good partition for a vertex $v$ with $\lambda(v, s) \le k-2$.

\ignore{
As discussed earlier, for the case $k=3$, the use of gadgets was essential in bounding the running time in~\cite{GKPP:3ECC}. To extend the framework of \cite{GKPP:3ECC} to a larger $k$, it is necessary that, given a small \outcomp{k} $S$, one can construct a gadget for both $S$ \emph{and} $V \setminus S$ in time proportional to the volume of $S$.
While our algorithm does not rely on gadgets, our framework permits their construction in $O(m)$ time (since we can afford $O(m)$ time per sampled vertex), which could prove useful for general $k$.
Moreover, it appears that our framework yields an $\widetilde{O}(m\sqrt{n})$-time algorithm for computing the \conn{k} components of a \conn{(k-2)} digraph.

More generally, if one is given a linear-time algorithm for computing the \conn{(k-1)} components of an arbitrary digraph $G$, along with a decomposition into $(k-1)$-edge-connected components, then our framework implies an $\Ot(m\sqrt{n})$-time algorithm for computing the \conn{k} components of $G$. However, we are not aware of whether such preconditions can be achieved for $k > 4$.
}

\ignore{
\paragraph{More generalizable.} As mentioned before, for $k=3$ the use of gadgets was crucial to bound the running time in \cite{GKPP:3ECC}. For this to work correctly, given a small \outcomp{k} $S$, one should be able to build a gadget both for $S$ {\emph{and}} for $V\setminus S$, in time proportional only to the volume of $S$. Although our algorithm does not utilize gadgets, our framework would still allow for their construction in $O(m)$ time (since we can afford $O(m)$ time for each sampled vertex), which could be useful for general $k$.\todo{We should check this!}

\paragraph{Other extensions.} It seems that our framework can give an $\widetilde{O}(m\sqrt{n})$ algorithm for \conn{k} components in \conn{(k-2)} graphs. More generally, if we are provided with a linear-time algorithm for computing \conn{(k-1)} components of a digraph $G$ (not necessarily $(k-2)$-edge-connected) together with a $(k-1)$-connected-components decomposition, our framework implies an $O(m\sqrt{n})$-time algorithm for computing the \conn{k} components of $G$. However, we are not aware if the stated preconditions can be justified for $k>4$. 
%
}

%% file: technical.tex
\section{Technical Overview}
\label{sec:technical}

In this section, we provide a high-level description of our algorithm and present the proof of our main result.

\subsection{Preliminaries and notation}

Let $G$ be a strongly connected graph with $n$ vertices and $m$ edges. 
In everything that follows, we assume that we work on the graph $G$, and all graph elements (e.g., vertices, edges, cuts, etc.) refer to $G$.

Let $S$ be a set of vertices. Every edge of the form $(x,y)$ with $x\in S$ and $y\notin S$ is called an \emph{outgoing edge} of $S$. The number of all outgoing edges of $S$ is denoted as $\mathit{out}(S)$. If $\mathit{out}(S)=k$, then $S$ is called a \emph{$k$-out set}. The tail of every outgoing edge of $S$ is called a \emph{boundary point} of $S$, and the head of every outgoing edge of $S$ is called an \emph{exit point} of $S$. We let $\mathit{vol}(S)$ denote \emph{the volume} of $S$: i.e., the number of edges whose tail is in $S$. When we say that an edge $e$ lies in a vertex set $U$, we mean that both endpoints of $e$ are in $U$. Otherwise, we say that $e$ does not lie in $U$.

Let $X$ and $Y$ be two disjoint sets of vertices. Then, every set of vertices $S$ with $X\subseteq S$ and $S\cap Y=\emptyset$ is called an $(X,Y)$-cut. (Notice that the order of $X$ and $Y$ here is important.) We may also say that $S$ separates $X$ and $Y$ (where, again, the order of $X$ and $Y$ in this expression is important). If $X$ or $Y$ consists of a single vertex, we may substitute it with the vertex that it consists of, and so we may speak of $(x,Y)$-cuts, or $(x,y)$-cuts, or $(X,y)$-cuts, where $x$ and $y$ are vertices. 
We denote by $\lambda(X,Y)$ the size of the minimum $(X,Y)$-cut in $G$.

Let $\mathcal{P}$ be a partition of $V$. If two vertices $x$ and $y$ belong to different sets from $\mathcal{P}$, then we say that \emph{$\mathcal{P}$ separates $x$ and $y$}. We say that \emph{$\mathcal{P}$ maintains the $(k+2)$-edge-connected components} if every two $(k+2)$-edge-connected vertices of $G$ are in the same set of $\mathcal{P}$.

\subsection{$(k+1)$-edge-connected component decomposition}

Our main approach is to exploit the notion of \emph{minimal \outcomp{(k+1)}s}, defined with respect to a fixed source vertex $s$. To make this scheme work, we must ensure that every other ordinary vertex is $(k+1)$-edge-connected to $s$. We achieve this via the following decomposition.

\begin{theorem}[$(k+1)$-Edge-Connected Component (\ecc{(k+1)}) Decomposition]
\label{theorem:keccdecomp}
Let $G$ be a $k$-edge-connected digraph with $n$ vertices and $m$ edges. In 
$\Ot(k^2 m)$ time we can construct a collection $H_1,\dots,H_t$ of graphs such that:
\begin{itemize}
\item{All graphs $H_1,\dots,H_t$ are $k$-edge-connected.} 
\item{The graphs $H_1,\dots,H_t$ have $O(n)$ vertices and $O(m)$ edges in total.}
\item{The vertices $V(H_i)$ of each graph $H_i$ are partitioned into two sets of vertices, ordinary and auxiliary.}
\item{For each $i \in \{1,\ldots,t\}$, the ordinary vertices of $V(H_i)$ are $(k+1)$-edge-connected.}
\item{For every vertex of $G$, there is exactly one graph among $H_1,\dots,H_t$ that contains it as an ordinary vertex.}
\item{Every two vertices $u$ and $v$ of $G$ are $(k+2)$-edge-connected if and only if there is an $i\in\{1,\dots,t\}$ such that $u$ and $v$ are $(k+2)$-edge-connected ordinary vertices of $H_i$.}
\end{itemize}
The algorithm is randomized and returns a correct decomposition with probability at least $1 - \delta$, where $\delta \in (0,1)$ is a parameter specified by the user. If the algorithm fails to compute the correct decomposition, it detects this event and reports an error.
\end{theorem}

We note that for any constant $k$, we can avoid the use of randomization in the algorithm of \Cref{theorem:keccdecomp}. Hence, we obtain a deterministic near-linear-time algorithm that 
computes a $(k+1)$-edge-connected-component (\ecc{(k+1)}) decomposition of a $k$-edge-connected digraph, for any constant $k$. Our algorithm exploits the computation of the nodes corresponding to a poset representation of the minimal \outcomp{k}s~\cite{Gabow:Poset:TALG} combined with a contraction operation. 

For $k=2$, we give a deterministic near-linear-time algorithm that computes a $3$-edge-connected-component decomposition of a \emph{general} digraph. 

\begin{theorem}[3-Edge-Connected Component (\ecc{3}) Decomposition]
\label{theorem:3eccdecomp}
Let $G$ be a digraph with $n$ vertices and $m$ edges. In $\Ot(m+n)$ time, we can construct a collection $H_1,\dots,H_t$ of graphs such that:
\begin{itemize}
\item{All graphs $H_1,\dots,H_t$ are strongly connected.}
\item{The graphs $H_1,\dots,H_t$ have $O(m)$ edges in total.}\lnote{Added the total number of edges. I think we can have $O(m)$ auxiliary vertices (and $n$ ordinary).}
\item{The vertices $V(H_i)$ of each graph $H_i$ are partitioned into two sets of vertices, ordinary and auxiliary.}
\item{For each $i \in \{1,\ldots,t\}$, the ordinary vertices of $V(H_i)$ are $3$-edge-connected.}
\item{For every vertex of $G$, there is exactly one graph among $H_1,\dots,H_t$ that contains it as an ordinary vertex.}
\item{Every two vertices $u$ and $v$ of $G$ are $4$-edge-connected if and only if there is an $i\in\{1,\dots,t\}$ such that $u$ and $v$ are $4$-edge-connected ordinary vertices of $H_i$.}
\end{itemize}
\end{theorem}

Similarly to the \ecc{(k+1)} decomposition, we take advantage of the computation of the nodes corresponding to a poset representation of the minimal \outcomp{2}s~\cite{Linear3ECC} combined with the gadget substitution operation of \cite{GKPP:3ECC}. 
The details are provided in \Cref{sec:decompositions}.

\subsection{Computing the $(k+2)$-edge-connected components through the minimum $(k+1)$-out sets}
\label{section:computing_(k+2)-eccs}

\Cref{theorem:keccdecomp} implies that the computation of the $(k+2)$-edge-connected components of $G$ reduces to that of the $(k+2)$-edge-connected components of the graphs $H_1,\dots,H_t$. Thus, from now on we may assume that the input graph $G$ satisfies the property that its vertex set can be partitioned into two kinds of vertices: ordinary and auxiliary. The ordinary vertices of $G$ are $(k+1)$-edge-connected, and the goal is to compute a partition $\mathcal{P}$ of $V(G)$ such that two ordinary vertices are $(k+2)$-edge-connected if and only if they belong to the same set in $\mathcal{P}$. 

We achieve our goal through a randomized algorithm which may have a one-sided error: with probability at most $\delta$ (for any fixed $\delta>0$ specified by the user), it may be that two ordinary vertices of $G$ that are not $(k+2)$-edge-connected belong to the same set in $\mathcal{P}$. However, if two ordinary vertices are $(k+2)$-edge-connected, then it is certain that they appear in the same set in $\mathcal{P}$. The precise formulation of our main result is given in the following (where we assume that $k$ is a fixed constant):

\begin{restatable}{theorem}{maintheorem}
\label{theorem:main}
Let $\delta$ be a number with $0<\delta<1$. There is an algorithm (see \Cref{algorithm:main}) that runs in $O(m\sqrt{n}\log(n/\delta))$ time, where $n$ is the number of the ordinary vertices of $G$, and outputs a partition $\mathcal{P}$ of $V(G)$ with the following guarantees:
\begin{itemize}
\item{Every two $(k+2)$-edge-connected ordinary vertices of $G$ are in the same set from $\mathcal{P}$.}
\item{With probability at least $1-\delta$, every two ordinary vertices of $G$ that are not $(k+2)$-edge-connected are separated by $\mathcal{P}$.}
\end{itemize} 
\end{restatable}

In order to establish \Cref{theorem:main}, we utilize the concept of the minimum $(k+2)$-out sets. Specifically, we first fix an arbitrary ordinary vertex $s$ of $G$. Then, for every ordinary vertex $v$ such that there is a $(k+1)$-out set that separates $v$ and $s$, we let $M(v)$ denote the (inclusion-wise) minimum \outcomp{(k+1)} that contains $v$ but not $s$. (The uniqueness of $M(v)$ is guaranteed by \Cref{corollary:minimal-exists}.) If for an ordinary vertex $v$ no such \outcomp{(k+1)} exists, then we let $M(v)=\bot$. (In particular, we have $M(s)=\bot$.) We use ``$M_R$'' to denote the same concept as ``$M$'' in the reverse graph $G^R$ (where we have fixed the same ordinary vertex $s$ in $G^R$).

It should be clear that, if two ordinary vertices $u$ and $w$ are $(k+2)$-edge-connected, then we have $M(u)=M(w)$ and $M_R(u)=M_R(w)$. On the other hand, if $u$ and $w$ are not $(k+2)$-edge-connected, then \Cref{lemma:distinctM} implies that either $M(u)\neq M(w)$ or $M_R(u)\neq M_R(w)$.

Thus, our approach is the following. First, we want to compute a partition $\mathcal{P}_1$ of $V(G)$ with the property that two ordinary vertices have the same $M$-set if and only if they belong to the same set in $\mathcal{P}_1$. Then, we want to compute a partition $\mathcal{P}_2$ of $V(G^R)$ with the property that two ordinary vertices have the same $M_R$-set if and only if they belong to the same set in $\mathcal{P}_2$. Then, the output $\mathcal{P}$ is the common refinement of $\mathcal{P}_1$ and $\mathcal{P}_2$. The computation of $\mathcal{P}_1$ and $\mathcal{P}_2$ is performed independently on $G$ and $G^R$, using the same procedure. Thus, it is sufficient to describe the idea for computing $\mathcal{P}_1$.

Now we distinguish two types of $M$-sets w.r.t. their volume. We call a vertex set ``small'' if it has volume at most $m/\sqrt{n}$. Otherwise, we call it ``large''. The small $M$-sets can be computed explicitly using the following local search based procedure (which follows from an adaptation of techniques from \cite{chechik2017faster}):

\begin{restatable}{proposition}{localsearch}
\label{proposition:local_search}
Let $v$ be a vertex with $v\neq s$ and $\lambda(v,s)\geq k+1$, and let $\Delta\geq 1$ be an integer. There is an algorithm $\mathtt{LocalSearchForMSet}(G,v,s,k+1,\Delta)$ which runs in $O(2^k(k+1)!\Delta)$ time and returns a set of vertices $S$ (which may be $\emptyset$) with the following guarantees:
\begin{itemize}
\item{If $S\neq\emptyset$, then $S=M(v)$.}
\item{If $S=\emptyset$, then either $M(v)=\bot$ or $\mathit{vol}(M(v))>\Delta$.}
\end{itemize}
\end{restatable}
\begin{proof}
See \Cref{sec:proposition:local_search}. 
\end{proof}

Assuming that $k$ is a fixed constant for our problem, we can use  \Cref{proposition:local_search}, in order to find all small $M$-sets in time $O(n\cdot (m/\sqrt{n}))=O(m\sqrt{n})$. Thus, we get a first partition of $V(G)$ with the property that two ordinary vertices have the same small $M$-set if and only if they belong to the same set of that partition. 

Notice that the dependency of the running time on $k$ is exponential, but this can be improved to a polynomial dependence by using a randomized algorithm to perform the local searches (although this will incur a polylogarithmic overhead on the running time, in order to guarantee a sufficiently high probability of success). 

\begin{restatable}{proposition}{randomizedlocalsearch}
\label{proposition:randomized_local_search}
Given that $\lambda(v,s)\geq k$, there is a procedure (shown in \Cref{algorithm:randomizedlocalsearch}) that has a running time of $O(k^2\Delta)$ and returns a (possibly empty) set of vertices $S$ with the following guarantees:
\begin{itemize}
\item{If $S\neq\emptyset$, then $S=M(v)$.}
\item{If $\lambda(v,s)=k$ and $\mathit{vol}(M(v))\leq\Delta$, then, with probability at least $1/2$, $S\neq\emptyset$.}
\end{itemize}
%
\end{restatable}

\begin{proof}
See \Cref{sec:proposition:random_local_search}. 
\end{proof}

For a randomized version of the local search, see \Cref{sec:proposition:random_local_search}. We note that the local search procedure of choice can be plugged in as a black-box in Line~\ref{line:localsearchforkcut} of \Cref{algorithm:main}.


The case of the large $M$-sets is highly involved, because we cannot afford to compute all of them explicitly. We discuss how to handle those sets in the following section.


\subsection{Sampling for large $M$-sets and computing good partitions}

In order to account for the partition of the ordinary vertices due to the large $M$-sets, we rely on structural properties of the $(k+1)$-edge cuts of the graph. Our first step is to sample enough edges of the graph (uniformly at random, with repetitions allowed), so that, with sufficiently high probability, for every large $M$-set $U$, we have sampled at least one edge whose tail is in $U$. For this purpose, it is sufficient to sample $\tilde{O}(\sqrt{n})$ edges, due to the large volume of the large $M$-sets (see proof of \Cref{theorem:main}).

Now, for every tail $v$ of every sampled edge, we distinguish two possibilities: either $\lambda(v,s)=k+1$, or $\lambda(v,s) \le k$. (In the case $\lambda(v,s) > k+1$ we do nothing, because this implies that $v$ is not included in a $(k+1)$-out set that separates $v$ and $s$.)

In each of those cases, our goal is to provide a partition $\mathcal{P}$ of $V(G)$ with the property that $(1)$ it maintains the $(k+2)$-edge-connected components, and $(2)$ for every ordinary vertex $u$ with $M(u)\neq\bot$ and $v\in M(u)$, and every ordinary vertex $w\notin M(u)$, we have that $u$ and $w$ are separated by $\mathcal{P}$. 
More formally, we have:


\begin{definition}[Good Partition]
\label{definition:goodpartition}
Let $v$ be any vertex such that $\lambda(v,s) \le k+1$. A partition $\mathcal{P}$ of $V$ is \emph{good for $v$} if it satisfies the following properties:
\begin{itemize}
\item{$\mathcal{P}$ maintains the $(k+2)$-edge-connected components.}
\item{For every ordinary vertex $u$ with $v\in M(u)$, and every ordinary vertex $w\notin M(u)$, we have that $u$ and $w$ are separated by $\mathcal{P}$.}
\end{itemize}
\end{definition}

We provide a linear-time algorithm for computing a good partition when $\lambda(v,s)=k+1$, by exploiting the DAG representation of all minimum cuts introduced by Picard and Queyranne~\cite{PicardQ82}.

\begin{restatable}{proposition}{mainordinary}
\label{proposition:main_ordinary}
Let $v$ be a vertex with $\lambda(v,s)=k+1$. Then there is an algorithm that runs in $O(km)$ time and returns a
good partition for $v$.
\end{restatable}

To obtain \Cref{proposition:main_auxiliary}, we construct the Picard--Queyranne DAG representation of all minimum $v$-$s$ cuts. This representation corresponds to the residual graph obtained after computing any maximum $(v,s)$-flow---which has value $k + 1$, since $\lambda(v, s) = k + 1$---and contracting its strongly connected components (SCCs). The resulting SCCs define a partition of $V$, which we show satisfies \Cref{definition:goodpartition}.  
The connection between the Picard--Queyranne DAG and the $M$-sets follows from a simple observation: for any ordinary vertex $u$ with $\lambda(u, s) = k + 1$, if $v \in M(u)$, then $M(u)$ is also a minimum $v$-$s$ cut, and is therefore represented in the DAG.  
Full details are provided in \Cref{section:proof_of_prop_main_ordinary}.

The case where $\lambda(v, s) = k$ is more involved. In this setting, we leverage structural properties of latest minimum cuts.


\begin{restatable}{proposition}{mainauxiliary}
\label{proposition:main_auxiliary}
Let $v$ be a vertex with $\lambda(v,s) = k$. Then there is an algorithm that runs in $O(k^2m)$ time and returns a
good partition for $v$. 
\ignore{partition $\mathcal{P}$ of $V$ with the following properties:
\begin{itemize}
\item{$\mathcal{P}$ maintains the $4$-edge-connected components.}
\item{For every ordinary vertex $u$ with $v\in M(u)$, and every ordinary vertex $w\notin M(u)$, we have that $u$ and $w$ are separated by $\mathcal{P}$.}
\end{itemize}
}
\end{restatable}

Let $S$ be the latest mincut that separates $v$ and $s$. This is the (inclusion-wise) maximum $\lambda(v,s)$-out set that separates $v$ and $s$. The reason that we consider the latest mincut is that, first, this can be computed in linear time, and second, it captures the notion of having made the utmost progress before increasing the edge-connectivity. (This idea is made precise in \Cref{lemma:include_outgoing_edge}.) 

Note that $S$ contains only auxiliary vertices (that is, $\lambda(x, s) = k$ for all $x \in S$), so we may contract $S$ into a single auxiliary vertex $z$ with out-degree $k$. 
Now consider a set $M(u)$ such that $v \in M(u)$, and let $U = M(u) \cup S$. 
For any ordinary vertex $w \notin M(u)$, we also have $w \notin U$. 
To separate $u$ and $w$, we distinguish two cases: 
(1) all $k$ outgoing edges $e_1, \ldots, e_k$ of $z$ belong to $E(U, V \setminus U)$; or 
(2) at least one outgoing edge $e_i$ of $z$ lies inside $U$. 
In the former case, it suffices to compute the $2$-edge-connected components of $G \setminus \{e_1, \ldots, e_k\}$, which can be done in linear time~\cite{georgiadis20162edge,GIP20:SICOMP}. 
In the latter case, we can contract $e_i$, thereby reducing the problem to the case $\lambda(v, s) = k + 1$.


\ignore{
We distinguish the cases $\lambda(v,s)=k+1$ and $\lambda(v,s) \le k$ because we handle them with different techniques. We will give a brief overview of those techniques here, and the full details can be found in Sections~\ref{section:proof1} and \ref{section:proof2}, respectively.

In the case where $\lambda(v,s)=3$, we initiate a procedure which we call ``triggered search''. (See Proposition~\ref{proposition:triggered}.) This makes repeated calls to the local-search procedure of Proposition~\ref{proposition:local_search}, and  it iteratively explores $3$-out sets on top of previously found $3$-out sets. More specifically, the triggered search is initiated at $v$, the first $3$-out set discovered is $M(v)$, and every subsequent local search starts from the most recently found $3$-out set contracted with one of its exit points. Thus, triggered search explores a sequence of $3$-out sets $S_1\subseteq S_2\subseteq\dots\subseteq S_t$, with $S_1=M(v)$. The total time to execute the triggered search procedure is $O(m)$.
For every new $3$-out set discovered by the triggered search, we construct an auxiliary graph which we call \emph{\intermediate{} graph}. We build this graph by contracting the previously found $3$-out set, by keeping intact the new part that was discovered in order to form the new $3$-out set, and by adding some auxiliary edges (For the precise definition of the {\intermediate{} graphs}, see Section~\ref{section:intermediate_graphs}.) Thus, the {\intermediate{} graph} corresponding e.g. to the exploration of the $3$-out set $S_i$ is formed by contracting $S_{i-1}$ into a single vertex, by keeping $S_i\setminus S_{i-1}$ intact, and by adding some auxiliary edges. The total size of the {\intermediate{} graphs} generated by a triggered search is $O(m)$.

Now consider an $M$-set $U$ with $v\in U$. (I.e., $U$ is the $M$-set of some ordinary vertex.) First, we show that the last $3$-out set explored by the triggered search -- i.e., following the notation of the previous paragraph, $S_t$ -- contains the set of vertices that are reachable from $v$ in $G[U]$. This implies that, if we remove the three outgoing edges of $S_t$, and compute the strongly connected components of the remaining graph, we have that every vertex in $U$ is separated from every vertex in $V\setminus S_t$. However, we also want to separate every ordinary vertex in $U$ from every ordinary vertex in $S_t\setminus U$. This is where we utilize the {\intermediate{} graphs} that we generate from the output of the triggered search. Specifically, for every {\intermediate{} graph} $H$, we compute a partition of its vertex set, which is given by the $3$-edge-connected components of $H$ after removing a single edge. (More precisely, there are three special edges in $H$ for which we perform this computation, and then we return the common refinement of the three resulting partitions.) This is where we use the linear-time algorithm of \cite{Linear3ECC} for computing the $3$-edge-connected components in linear time. The resulting partition of $V(H)$ has the property that it maintains the ordinary vertices that are $4$-edge-connected (in $G$) in the same cluster. Furthermore, (and this is the crucial part for which we utilize properties of $3$-edge cuts), if $H$ contains two ordinary vertices $u$ and $w$ such that $u\in U$ and $w\notin U$, then $u$ and $w$ end up in different clusters in this partition. (Hence the name ``Partition-Witness graphs''.) In this way we provide the partition of the ordinary vertices induced by every large $M$-set for which we have sampled a vertex $v$ in it with $\lambda(v,s)=3$. Observe that the total time here can be easily bounded by $\tilde{O}(m\sqrt{n})$.

In the case where $\lambda(v,s)<3$, we utilize the concept of the latest mincut $S$ that separates $v$ and $s$. This is the (inclusion-wise) maximum $\lambda(v,s)$-out set that separates $v$ and $s$. The reason that we consider the latest mincut is that, first, this can be computed in linear time, and second, it captures the notion of having made the utmost progress before increasing the edge-connectivity. (This idea is made precise in Lemma~\ref{lemma:include_outgoing_edge}.) 

Now consider again an $M$-set $U$ with $v\in U$. Here there are two possibilities: either there is an outgoing edge $e$ of $S$ which lies outside of $U$, or every outgoing edge of $S$ lies inside of $U$. In the first case, we show that it is sufficient to remove $e$,  and compute the $3$-edge-connected components of the resulting graph. Then we have the property that every ordinary vertex in $U$ is separated from every vertex in $V\setminus U$ (while maintaining the $4$-edge-connected components). In the second case, we show that it is sufficient to contract $S$ together with one of its exit points into a single vertex $r$, and recurse. That is, we now deal with $r$ in the same way that we processed the sampled vertex $v$. Since $S$ is the latest mincut that separates $v$ from $s$, we have that $\lambda(r,s)>\lambda(v,s)$. Thus, the total time for processing all sampled vertices $v$ with $\lambda(v,s)<3$ can be easily bounded again by $\tilde{O}(m\sqrt{n})$. 
}



\subsection{Proof of the main theorem}
\label{section:main-result}

We now present the proof of our first main result, stated in \Cref{theorem:keccs}. 
We begin with the case where $k$ is a fixed constant. 
In this setting, we can apply the deterministic version of the \ecc{(k+1)}-decomposition (see the paragraph right after the statement of \Cref{theorem:keccdecomp}), which runs in $\Ot(m)$ time.

After this decomposition, we can consider a $k$-edge-connected graph $G$ with $n$ vertices and $m$ edges, that consists of two kinds of vertices: ordinary and auxiliary. The ordinary vertices of $G$ are $(k+1)$-edge-connected. 
We prove the following:

\maintheorem*
\begin{proof}
Let $s$ be a fixed ordinary vertex of $G$. First, we compute all $M$-sets with volume at most $m/\sqrt{n}$. Specifically, for every ordinary vertex $v\neq s$, we apply the algorithm $\mathtt{LocalSearchForMSet}(G,v,s,k+1,\Delta)$, whose guarantees are stated in \Cref{proposition:local_search}. Thus, we get a vertex set $S(v)$ such that, if $M(v)\neq\bot$ and $\mathit{vol}(M(v))\leq m/\sqrt{n}$, then $S(v)=M(v)$. Otherwise, we may still get $M(v)$, or $\emptyset$. (In any case, we will not get a non-empty output which is distinct from $M(v)$.) Since we apply this procedure for every ordinary vertex $v\neq s$ (and the number of such ordinary vertices may be as large as $n-1$), by \Cref{proposition:local_search} we get an $O((n-1)\cdot m/\sqrt{n})=O(m\sqrt{n})$ time bound for this step. The total output is a collection of $O(n)$ sets of vertices with $O(m/\sqrt{n})$ vertices in each set. We can sort those sets in $O(n\cdot m/\sqrt{n})=O(m\sqrt{n})$ time with bucket sort, and determine the partition $\mathcal{P}_1$ of the ordinary vertices that have the same $M$-set with volume $O(m/\sqrt{n})$. To be more precise, if for two distinct ordinary vertices $v$ and $v'$ we have $S(v)=S(v')$, then we put $v$ and $v'$ in the same set in $\mathcal{P}_1$. (And we can put all auxiliary vertices of $G$ in a single set in this partition.) Thus, if two ordinary vertices are $(k+2)$-edge-connected, then they appear in the same set in $\mathcal{P}_1$.

Now we have to consider the separations induced by the $M$-sets of ordinary vertices with volume larger than $m/\sqrt{n}$. To do this, we first sample enough edges (uniformly, with repetitions allowed) so that, with probability at least $1-\delta$, we have that, for every such ``large'' $M$-set $U$, we have sampled at least one edge whose tail is in $U$. Thus, it is sufficient to sample $N=\lceil\sqrt{n}\log_{2}(n/\delta)\rceil$ edges. To see this, consider an $M$-set $U$ with volume more than $m/\sqrt{n}$. This means that there are more than $m/\sqrt{n}$ edges whose tail is in $U$. Then, the probability that a randomly chosen edge does not have its tail in $U$ is less than $1-(m/\sqrt{n})/m = 1-1/\sqrt{n}$. Thus, the probability that, after $N$ samples, we have not sampled an edge whose tail is in $U$ is less than 
$$
(1-1/\sqrt{n})^N\leq 2^{-N/\sqrt{n}}\leq 2^{-\sqrt{n}\log_{2}(n/\delta)/\sqrt{n}} = 2^{-\log_{2}(n/\delta)}=\delta/n.$$ 
Thus, since $n$ is a trivial upper bound on the number of distinct $M$-sets, by the union bound we have that: the probability that there is an $M$-set $U$ with volume more that $m/\sqrt{n}$ for which we have not sampled an edge whose tail is in $U$ is less than $\delta$. 

Thus, from now on, we may assume that: for every $M$-set $U$ with volume more than $m/\sqrt{n}$, we have sampled at least one edge whose tail is in $U$ $(*)$. 

Now, for every tail $v$ of an edge that we have sampled, we do the following. First, we perform at most $k+2$ iterations of Ford-Fulkerson's algorithm, in order to determine whether $\lambda(v,s) \le k$, or $\lambda(v,s)=k+1$, or $\lambda(v,s)>k+1$. (This takes time $O(km)$.) In the first and second case, we apply \Cref{proposition:main_auxiliary} and \Cref{proposition:main_ordinary}, respectively, on $v$, in order to get a partition $\mathcal{P}(v)$ of $V(G)$ with the guarantess provided by the respective proposition. (Again, this step takes $O(m)$ time.) In the third case (i.e., if $\lambda(v,s)>k+1$), we do nothing. Thus, we get a collection of $O(N)$ partitions of $V(G)$ in total time $O(Nm)=O(\sqrt{n}\log(n/\delta)m)$. We compute the common refinement $\mathcal{P}_2$ of those partitions with bucket sort, and this takes time $O(nN)$. Then, we return the common refinement of $\mathcal{P}_1$ and $\mathcal{P}_2$, which takes $O(n)$ time to be computed. This is how we get a partition $\mathcal{Q}$. Notice that, by the guarantees of \Cref{proposition:main_auxiliary,proposition:main_ordinary}, we have that $\mathcal{Q}$ maintains the $(k+2)$-edge-connected components of $G$.

Now we repeat the same process on the reverse graph $G^R$, with the same fixed vertex $s$, and thus we get a partition $\mathcal{Q}_R$. Our final output is the common refinement of $\mathcal{Q}$ and $\mathcal{Q}_R$.

In order to establish correctness, it remains to show that: if two ordinary vertices $u$ and $w$ are not $(k+2)$-edge-connected, then they appear in distinct sets in the final output. So let $u$ and $w$ be two ordinary vertices that are not $(k+2)$-edge-connected. Then, by \Cref{lemma:distinctM} we have that either $M(u)\neq M(w)$ or $M_R(u)\neq M_R(w)$. Let us assume, w.l.o.g., that $M(u)\neq M(w)$. Now, if either $\mathit{vol}(M(u))\leq m/\sqrt{n}$ or $\mathit{vol}(M(w))\leq m/\sqrt{n}$, then $u$ and $w$ are separated by $\mathcal{P}_1$, and thus by the final output. Otherwise, for each of $u$ and $w$ we have that its $M$-set is either $\bot$, or it has volume more than $m/\sqrt{n}$, but they cannot both be $\bot$. Now, since $M(u)\neq M(w)$, by \Cref{lemma:distinctM2} we have that either $M(u)\neq\bot$ and $w\notin M(u)$, or $M(w)\neq\bot$ and $u\notin M(w)$. Thus, we may assume, w.l.o.g., that $M(u)\neq\bot$ and $w\notin M(u)$. Now, since $\mathit{vol}(M(u))>m/\sqrt{n}$, our assumption $(*)$ implies that we have sampled at least one edge whose tail $v$ is in $M(u)$. Then, we obviously have that $\lambda(v,s)\leq k+1$, and therefore one of \Cref{proposition:main_auxiliary,proposition:main_ordinary} (depending on whether $\lambda(v,s) \le k$ or $\lambda(v,s)=k+1$, respectively), implies that $\mathcal{P}(v)$ separates $u$ and $w$. Thus, $u$ and $w$ are separated by the final output.

(Notice that our assumption $(*)$ holds on $G$ with probability more than $1-\delta$. However, we want the same assumption to hold simultaneously for $G$ and $G^R$ with probability more than $1-\delta$. Thus, we should have chosen $\delta$ to be half of the desired probability of failure, so that, by the union bound, this assumption is satisfied simultaneously in $G$ and $G^R$ with the desired probability of success.)
\end{proof}

\begin{algorithm}[h!]
\caption{\textsf{Compute the $(k+2)$-edge-connected components of $G$, with probability at least $1-\delta$}}
\label{algorithm:main}
\LinesNumbered
\DontPrintSemicolon
fix an ordinary vertex $s$\;
\ForEach{ordinary vertex $v\neq s$}{
  apply procedure $\mathtt{LocalSearchForMSet}(G,v,s,k+1,m/\sqrt{n})$, and let $S(v)$ be the output\;
  \label{line:localsearchforkcut}
  \tcp{See \Cref{proposition:local_search}}
}
let $\mathcal{P}_1$ be the partition of $V$ that is formed by the clusters of vertices that have the same $S$ set computed in the previous step (let $S(u)=\emptyset$ for every auxiliary vertex $u$)\;
sample $\lceil\sqrt{n}\log_{2}(2n/\delta)\rceil$ edges (uniformly, with repetitions allowed)\;
\label{line:algorithm:main:sample}
\ForEach{tail $v$ of a sampled edge}{
  \If{$\lambda(v,s)=k+1$}{
  \label{line:lambda=k+1}
    apply the algorithm described in \Cref{proposition:main_ordinary}, and let $\mathcal{P}(v)$ be the output partition of $V$\;
  }
  \ElseIf{$\lambda(v,s) \leq k$}{
  \label{line:lambda<=k}
    apply the algorithm described in \Cref{proposition:main_auxiliary}, and let $\mathcal{P}(v)$ be the output partition of $V$\;
    \label{line:prop_for_auxiliary}
  }
}
let $\mathcal{P}_2$ be the common refinement of all partitions of the form $\mathcal{P}(v)$ computed in the previous step\;
let $\mathcal{Q}$ be the common refinement of $\mathcal{P}_1$ and $\mathcal{P}_2$\;
let $\mathcal{Q}_R$ be the result of applying the same procedure on $G^R$\;
\textbf{return} the common refinement of $\mathcal{Q}$ and $\mathcal{Q}_R$\;
\end{algorithm}

If we drop the assumption that $k$ is a fixed constant, then, by relying throughout on the deterministic local search procedure of \Cref{proposition:local_search}, there is an exponential dependency on $k$ in the time bound for computing the $(k+2)$-edge-connected components. We can improve this dependency to polynomial, by using instead the randomized local search procedure of \Cref{proposition:randomized_local_search}. The local search procedure is used in two different places in our framework: in the algorithm for computing the \ecc{(k+1)}-decomposition (\Cref{theorem:keccdecomp}), and in the search for small $M$-sets in Line~\ref{line:localsearchforkcut} of \Cref{algorithm:main}. 

In particular, if we want to have the same guarantees for \Cref{algorithm:main} as in the statement of \Cref{theorem:main}, but with a running time of $O(k^2m\sqrt{n}\log(n/\delta))$, then we must replace the local search procedure in Line~\ref{line:localsearchforkcut} with that provided by \Cref{proposition:randomized_local_search}, and repeat it for $\lceil \log{(4n/\delta)} \rceil$ times for every vertex $v\neq s$. This is in order to get a $\delta/(4n)$ bound on the probability of failure for every vertex $v\neq s$, for the search for small $M$-sets in each direction (i.e., in $G$ and in $G^R$), so that the union bound gives a $\delta/2$ bound on the probability of failure for this search in both directions. Furthermore, we must increase the sampling rate in Line~\ref{line:algorithm:main:sample} from $\lceil\sqrt{n}\log_{2}(2n/\delta)\rceil$ to $\lceil\sqrt{n}\log_{2}(4n/\delta)\rceil$, in order to also get a $\delta/2$ bound on the probability of failure for the search for large $M$-sets (in both directions). Thus, by the union bound, the total probability of failure of \Cref{algorithm:main} can be bounded by $\delta$.

Overall, in order to establish \Cref{theorem:keccs}, we need as input the graph $G$, and the parameters $k$ and $\delta$. First, we employ the randomized decomposition of \Cref{theorem:keccdecomp}, which runs in $\Ot(k^2 m)$ time and produces a correct decomposition $H_1,\dots, H_t$ with probability at least $1 - \delta/2$. Then, for each graph $H_i$, we apply \Cref{algorithm:main} with bound $\delta/(2n)$ on the probability of failure. This is because: first, there are at most $n$ graphs $H_i$ on which we apply this algorithm, and second, each of those graphs has at most $n$ ordinary vertices (that correspond to vertices of the original graph). Thus, the union bound establishes the $\delta$ bound on the probability of failure for computing the $(k+2)$-edge-connected components of the original graph $G$.

\ignore{
For non-constant values of~$k$, we employ the randomized decomposition of \Cref{theorem:keccdecomp}, which runs in $\Ot(k^2 m)$ time and produces a correct decomposition with probability at least $1 - \delta_1$, 
for a user-specified parameter $0 < \delta_1 < 1$.
If the construction fails (which occurs with probability at most~$\delta_1$), the algorithm terminates and reports an error.

Assuming that the \ecc{(k+1)}-decomposition succeeds, we proceed to execute \Cref{algorithm:main}. 
In line~\ref{line:localsearchforkcut}, we invoke the randomized local search described in \Cref{proposition:randomized_local_search}. 
To amplify the success probability for correctly identifying the small $M$-sets, we repeat the local search from each vertex~$v$ a total of 
$\lceil \log{(2n/\delta_2)} \rceil$ times, for a user-specified parameter $0 < \delta_2 < 1$. 
Since at most $2(n-1)$ local searches are performed for small $M$-sets (at most $n-1$ in the forward direction and at most $n-1$ in the reverse direction), a union bound implies that the probability of any such search failing is at most~$\delta_2$.
Moreover, by the proof of \Cref{theorem:main}, the probability that the search for large $M$-sets fails (in either direction) is at most~$\delta_3$, for a user-specified $0 < \delta_3 < 1$.
Consequently, the overall success probability of our algorithm is at least $1 - \delta$, where $\delta = \delta_1 + \delta_2 + \delta_3$.

This completes the proof of \Cref{theorem:keccs}.\lnote{Added a formal proof of Theorem 1. Please check!}
}

\subsection{Computing the $4$-edge-connected components}
In order to establish \Cref{theorem:4eccs}, we use the same framework that we applied in order to establish \Cref{theorem:keccs}, with some crucial differences that we will discuss here. (Notice that \Cref{theorem:4eccs} is not an immediate consequence of \Cref{theorem:keccs}, because here we assume that the input graph $G$ is any \emph{general} digraph, i.e., not necessarily $k$-edge-connected.)

First, \Cref{theorem:3eccdecomp} implies that the computation of the $4$-edge-connected components of $G$ reduces to that of the $4$-edge-connected components of the graphs (of the $3$-ECC decomposition) $H_1,\dots,H_t$. Thus, from now on we may assume that the input graph $G$ is strongly connected and has two kinds of vertices: ordinary and auxiliary, where the ordinary vertices are $3$-edge-connected. The goal is to compute a partition $\mathcal{P}$ of $V(G)$ such that two ordinary vertices are $4$-edge-connected if and only if they belong to the same set in $\mathcal{P}$.

Now we fix an ordinary vertex $s$, and we apply \Cref{algorithm:main}, with the following modification in Line~\ref{line:prop_for_auxiliary}: instead of using \Cref{proposition:main_auxiliary}, we use \Cref{proposition:main_auxiliary_4}. This is because here we may have $\lambda(v,s)\in\{1,2\}$ (whereas with the assumption that the graph is $k$-edge-connected, only the case $\lambda(v,s)=k$ can arise).

\begin{restatable}{proposition}{mainaux}
\label{proposition:main_auxiliary_4}
Let $G$ be a strongly connected graph that has two kinds of vertices: ordinary and auxiliary, where the ordinary vertices are $3$-edge-connected. Let $s$ be an ordinary vertex, and let $v$ be a vertex with $\lambda(v,s)\leq 2$. Then there is an algorithm that runs in $O(m)$ time and returns a partition $\mathcal{P}$ of $V(G)$ with the following properties:
\begin{itemize}
\item{$\mathcal{P}$ maintains the $4$-edge-connected components of $G$.}
\item{For every two ordinary vertices $u$ and $w$ with $v\in M(u)$ and $w\notin M(u)$, we have that $u$ and $w$ are separated by $\mathcal{P}$. }
\end{itemize}
(Here ``$M(u)$'' is the minimum $3$-out set that separates $u$ from $s$.)
\end{restatable}

We note that Line~\ref{line:prop_for_auxiliary} is the only place in which the use of \Cref{algorithm:main} for the computation of the $4$-edge-connected components differs from that for the computation of the $(k+2)$-edge-connected components. In particular, in Lines~\ref{line:localsearchforkcut}, \ref{line:lambda=k+1} and \ref{line:lambda<=k}, we just replace ``$k$'' with ``$2$''.
The analysis is the same as in \Cref{section:computing_(k+2)-eccs,section:main-result}, where the ``$M$'' sets here have the meaning of minimal $3$-out sets that separate ordinary vertices from $s$. The proof of \Cref{proposition:main_auxiliary_4} is discussed in \Cref{section:proof_of_prop_main_auxiliary}.

%% file: text.tex
\ignore{
\section{Main result}
\label{section:main-result}
\todo[inline]{We should re-organize. Title: Technical overview, definitions, statement of theorem for 3ecc-decomposition, statements of propositions 11,12 and 27, proof of main theorem}

Let $G$ be a strongly connected graph with $n$ vertices and $m$ edges (thus, we have $m\geq n$). We assume that $G$ consists of two kinds of vertices: ordinary and auxiliary. The ordinary vertices of $G$ are $3$-edge-connected. Our main result is the following:

\begin{theorem}
Let $\delta$ be a number with $0<\delta<1$. There is an algorithm (see \Cref{algorithm:main}) that runs in $O(m\sqrt{n}\log(n/\delta))$ time and outputs a partition $\mathcal{P}$ of $V(G)$ with the following guarantees:
\begin{itemize}
\item{Every two $4$-edge-connected ordinary vertices of $G$ are in the same set from $\mathcal{P}$.}
\item{With probability at least $1-\delta$, every two ordinary vertices of $G$ that are not $4$-edge-connected are separated by $\mathcal{P}$.}
\end{itemize} 
\end{theorem}
\begin{proof}
Let $s$ be a fixed ordinary vertex of $G$. First, we compute all $M$-sets with volume at most $m/\sqrt{n}$. Specifically, for every ordinary vertex $v\neq s$, we apply the algorithm $\mathtt{LocalSearchFor3Cut}(G,v,s,\Delta)$, whose guarantees are stated in \Cref{proposition:local_search}. Thus, we get a vertex set $S(v)$ such that, if $M(v)\neq\bot$ and $\mathit{vol}(M(v))\leq m/\sqrt{n}$, then $S(v)=M(v)$. Otherwise, we may still get $M(v)$, or $\emptyset$. (In any case, we will not get a non-empty output which is distinct from $M(v)$.) Since we apply this procedure for every ordinary vertex $v\neq s$ (and the number of such ordinary vertices may be as large as $n-1$), by \Cref{proposition:local_search} we get an $O((n-1)\cdot m/\sqrt{n})=O(m\sqrt{n})$ time bound for this step. The total output is a collection of $O(n)$ sets of vertices with $O(m/\sqrt{n})$ vertices in each set. We can sort those sets in $O(n\cdot m/\sqrt{n})=O(m\sqrt{n})$ time with bucket sort, and determine the partition $\mathcal{P}_1$ of the ordinary vertices that have the same $M$-set with volume $O(m/\sqrt{n})$. To be more precise, if for two distinct ordinary vertices $v$ and $v'$ we have $S(v)=S(v')$, then we put $v$ and $v'$ in the same set in $\mathcal{P}_1$. (And we can put all auxiliary vertices of $G$ in a single set in this partition.) Thus, if two ordinary vertices are $4$-edge-connected, then they appear in the same set in $\mathcal{P}_1$.

Now we have to consider the separations induced by the $M$-sets of ordinary vertices with volume larger than $m/\sqrt{n}$. To do this, we first sample enough edges (uniformly, with repetitions allowed) so that, with probability at least $1-\delta$, we have that, for every such ``large'' $M$-set $U$, we have sampled at least one edge whose tail is in $U$. Thus, it is sufficient to sample $N=\lceil\sqrt{n}\log_{2}(n/\delta)\rceil$ edges. To see this, consider an $M$-set $U$ with volume more than $m/\sqrt{n}$. This means that there are more than $m/\sqrt{n}$ edges whose tail is in $U$. Then, the probability that a randomly chosen edge does not have its tail in $U$ is less than $1-(m/\sqrt{n})/m = 1/\sqrt{n}$. Thus, the probability that, after $N$ samples, we have not sampled an edge whose tail is in $U$ is less than 
$$
(1-1/\sqrt{n})^N\leq 2^{-N/\sqrt{n}}\leq 2^{-\sqrt{n}\log_{2}(n/\delta)/\sqrt{n}} = 2^{-\log_{2}(n/\delta)}=\delta/n.$$ 
Thus, since $n$ is a trivial upper bound on the number of distinct $M$-sets, by the union bound we have that: the probability that there is an $M$-set $U$ with volume more that $m/\sqrt{n}$ for which we have not sampled an edge whose tail is in $U$ is less than $\delta$. 

Thus, from now on, we may assume that: for every $M$-set $U$ with volume more than $m/\sqrt{n}$, we have sampled at least one edge whose tail is in $U$ $(*)$. 

Now, for every tail $v$ of an edge that we have sampled, we do the following. First, we perform at most four iterations of Ford-Fulkerson's algorithm, in order to determine whether $\lambda(v,s)<3$, or $\lambda(v,s)=3$, or $\lambda(v,s)>3$. (This takes time $O(m)$.) In the first and second case, we apply \Cref{proposition:main_auxiliary} and \Cref{proposition:main_ordinary}, respectively, on $v$, in order to get a partition $\mathcal{P}(v)$ of $V(G)$ with the guarantess provided by the respective proposition. (Again, this step takes $O(m)$ time.) In the third case (i.e., if $\lambda(v,s)>3$), we do nothing. Thus, we get a collection of $O(N)$ partitions of $V(G)$ in total time $O(Nm)=O(\sqrt{n}\log(n/\delta)m)$. We compute the common refinement $\mathcal{P}_2$ of those partitions with bucket sort, and this takes time $O(nN)$. Then, we return the common refinement of $\mathcal{P}_1$ and $\mathcal{P}_2$, which takes $O(n)$ time to be computed. This is how we get a partition $\mathcal{Q}$. Notice that, by the guarantees of \Cref{proposition:main_auxiliary,proposition:main_ordinary}, we have that $\mathcal{Q}$ maintains the $4$-edge-connected components of $G$.

Now we repeat the same process on the reverse graph $G^R$, with the same fixed vertex $s$, and thus we get a partition $\mathcal{Q}_R$. Our final output is the common refinement of $\mathcal{Q}$ and $\mathcal{Q}_R$.

In order to establish correctness, it remains to show that: if two ordinary vertices $u$ and $w$ are not $4$-edge-connected, then they appear in distinct sets in the final output. So let $u$ and $w$ be two ordinary vertices that are not $4$-edge-connected. Then, by \Cref{lemma:distinctM} we have that either $M(u)\neq M(w)$ or $M_R(u)\neq M_R(w)$. Let us assume, w.l.o.g., that $M(u)\neq M(w)$. Now, if either $\mathit{vol}(M(u))\leq m/\sqrt{n}$ or $\mathit{vol}(M(w))\leq m/\sqrt{n}$, then $u$ and $w$ are separated by $\mathcal{P}_1$, and thus by the final output. Otherwise, for each of $u$ and $w$ we have that its $M$-set is either $\bot$, or it has volume more than $m/\sqrt{n}$, but they cannot both be $\bot$. Now, since $M(u)\neq M(w)$, by \Cref{lemma:distinctM2} we have that either $M(u)\neq\bot$ and $w\notin M(u)$, or $M(w)\neq\bot$ and $u\notin M(w)$. Thus, we may assume, w.l.o.g., that $M(u)\neq\bot$ and $w\notin M(u)$. Now, since $\mathit{vol}(M(u))>m/\sqrt{n}$, our assumption $(*)$ implies that we have sampled at least one edge whose tail $v$ is in $M(u)$. Then, we obviously have that $\lambda(v,s)\leq 3$, and therefore one of \Cref{proposition:main_auxiliary,proposition:main_ordinary} (depending on whether $\lambda(v,s)<3$ or $\lambda(v,s)=3$, respectively), implies that $\mathcal{P}(v)$ separates $u$ and $w$. Thus, $u$ and $w$ are separated by the final output.

(Notice that our assumption $(*)$ holds on $G$ with probability more than $1-\delta$. However, we want the same assumption to hold simultaneously for $G$ and $G^R$ with probability more than $1-\delta$. Thus, we should have chosen $\delta$ to be half of the desired probability of failure, so that, by the union bound, this assumption is satisfied simultaneously in $G$ and $G^R$ with the desired probability of success.)
\end{proof}

\begin{algorithm}[h!]
\caption{\textsf{Compute the $4$-edge-connected components of $G$, with probability at least $1-\delta$}}
\label{algorithm:main}
\LinesNumbered
\DontPrintSemicolon
fix an ordinary vertex $s$\;
\ForEach{ordinary vertex $v\neq s$}{
  apply procedure $\mathtt{LocalSearchFor3Cut}(G,v,s,m/\sqrt{n})$, and let $S(v)$ be the output\;
  \tcp{See \Cref{proposition:local_search}}
}
let $\mathcal{P}_1$ be the partition of $V$ that is formed by the clusters of vertices that have the same $S$ set computed in the previous step (let $S(u)=\emptyset$ for every auxiliary vertex $u$)\;
sample $\lceil\sqrt{n}\log_{2}(2n/\delta)\rceil$ edges (uniformly, with repetitions allowed)\;
\ForEach{tail $v$ of a sampled edge}{
  \If{$\lambda(v,s)=3$}{
    apply the algorithm described in \Cref{proposition:main_ordinary}, and let $\mathcal{P}(v)$ be the output partition of $V$\;
  }
  \ElseIf{$\lambda(v,s)<3$}{
    apply the algorithm described in \Cref{proposition:main_auxiliary}, and let $\mathcal{P}(v)$ be the output partition of $V$\;
  }
}
let $\mathcal{P}_2$ be the common refinement of all partitions of the form $\mathcal{P}(v)$ computed in the previous step\;
let $\mathcal{Q}$ be the common refinement of $\mathcal{P}_1$ and $\mathcal{P}_2$\;
let $\mathcal{Q}_R$ be the result of applying the same procedure on $G^R$\;
\textbf{return} the common refinement of $\mathcal{Q}$ and $\mathcal{Q}_R$\;
\end{algorithm}
}

\section{Basic definitions}

\ignore{
\subsection{Basic definitions}
Let $G$ be a strongly connected graph with $n$ vertices and $m$ edges, that consists of two kinds of vertices: ordinary and auxiliary. The ordinary vertices of $G$ are $3$-edge-connected. In everything that follows, we assume that we work on the graph $G$, and all graph elements (e.g., vertices, edges, cuts, etc.) refer to $G$.

Let $S$ be a set of vertices. Every edge of the form $(x,y)$ with $x\in S$ and $y\notin S$ is called an \emph{outgoing edge} of $S$. The number of all outgoing edges of $S$ is denoted as $\mathit{out}(S)$. If $\mathit{out}(S)=k$, then $S$ is called a \emph{$k$-out set}. The tail of every outgoing edge of $S$ is called a \emph{boundary point} of $S$, and the head of every outgoing edge of $S$ is called an \emph{exit point} of $S$. We let $\mathit{vol}(S)$ denote \emph{the volume} of $S$: i.e., the number of edges whose tail is in $S$. When we say that an edge $e$ lies in a vertex set $U$, we mean that both endpoints of $e$ are in $U$. Otherwise, we say that $e$ does not lie in $U$.

Let $X$ and $Y$ be two disjoint sets of vertices. Then, every set of vertices $S$ with $X\subseteq S$ and $S\cap Y=\emptyset$ is called an $(X,Y)$-cut. (Notice that the order of $X$ and $Y$ here is important.) We may also say that $S$ separates $X$ and $Y$ (where, again, the order of $X$ and $Y$ in this expression is important). If $X$ or $Y$ consists of a single vertex, we may substitute it with the vertex that it consists of, and so we may speak of $(x,Y)$-cuts, or $(x,y)$-cuts, or $(X,y)$-cuts, where $x$ and $y$ are vertices. 

Let $\mathcal{P}$ be a partition of $V$. If two vertices $x$ and $y$ belong to different sets from $\mathcal{P}$, then we say that \emph{$\mathcal{P}$ separates $x$ and $y$}. We say that \emph{$\mathcal{P}$ maintains the $4$-edge-connected components} if every two $4$-edge-connected ordinary vertices of $G$ are in the same set of $\mathcal{P}$.
}

\subsection{The operator of local edge connectivity}
Let $X$ and $Y$ be two disjoint sets of vertices. If there is a $k$-out set $S$ with $X\subseteq S$ and $S\cap Y=\emptyset$, then we say that $S$ is \emph{a $k$-cut that separates $X$ and $Y$}. (We note that the order of $X$ and $Y$ here is important.)  Furthermore, if there is no $k'$-out set $S'$ with $k'<k$ such that $X\subseteq S'$ and $S'\cap Y=\emptyset$, then we say that $S$ is an \emph{$(X,Y)$-mincut}. If there is an $(X,Y)$-mincut $S$ with $\mathit{out}(S)=k$, then we write $\lambda(X,Y)=k$. (We note that the operator $\lambda$ does not necessarily act symmetrically on sets of vertices.)

If either of $X$ and $Y$ consists of a single vertex, then we may substitute it in the expression ``$\lambda(X,Y)$'' with the vertex that it consists of. (E.g., if $X=\{x\}$, then we may denote $\lambda(X,Y)$ simply as $\lambda(x,Y)$.) Furthermore, we may write ``$\lambda_G$'' instead of just ``$\lambda$'', if we want to specify the reference graph $G$.

\begin{lemma}
\label{lemma:submodularity3}
Let $X$ and $Y$ be two disjoint sets of vertices with $\lambda(X,Y)= k$, and let $S$ and $S'$ be two $k$-out sets that separate $X$ and $Y$. Then, both $S\cup S'$ and $S\cap S'$ are $k$-out sets.
\end{lemma}
\begin{proof}
Due to the submodularity of the cut function, we have $\mathit{out}(S\cap S')+\mathit{out}(S\cup S')\leq\mathit{out}(S)+\mathit{out}(S')$, and therefore $\mathit{out}(S\cap S')+\mathit{out}(S\cup S')\leq 2k$. Notice that both $S\cap S'$ and $S\cup S'$ are cuts that separate $X$ and $Y$. Thus, since $\lambda(X,Y)= k$, we have $\mathit{out}(S\cap S')\geq k$ and $\mathit{out}(S\cup S')\geq k$. Therefore, we have $2k\leq\mathit{out}(S\cap S')+\mathit{out}(S\cup S')$, and thus we infer that $\mathit{out}(S\cap S')=k$ and $\mathit{out}(S\cup S')=k$.
\end{proof}

We will be using the following corollary throughout, without explicitly invoking it.

\begin{corollary}
\label{corollary:minimal-exists}
Let $X$ and $Y$ be two disjoint sets of vertices with $\lambda(X,Y)=k$. Then there is an inclusion-wise minimum $k$-out set that separates $X$ and $Y$. Furthermore, there is an inclusion-wise maximum $k$-out set that separates $X$ and $Y$.
\end{corollary}
\begin{proof}
Consider the collection $\mathcal{S}$ of all $k$-out sets that separate $X$ and $Y$. Then, \Cref{lemma:submodularity3} implies that $\bigcap{\mathcal{S}}$ is a $k$-out set that separates $X$ and $Y$. Similarly, \Cref{lemma:submodularity3} implies that $\bigcup{\mathcal{S}}$ is a $k$-out set that separates $X$ and $Y$.
\end{proof}


Let $Z$ be a set of vertices of $G$, and let $G_Z$ be the graph that is formed by contracting $Z$ into a single vertex $z$. If $X$ is a set of vertices of $G_Z$, then we let $X'$ denote $X$ if $z\notin X$, and $(X\setminus\{z\})\cup{Z}$ otherwise. (I.e., $X'$ is the ``uncontraction'' of $X$ in $G$.) Then it is easy to see that, for every set of vertices $S$ of $G_Z$, we have $\mathit{out}_{G_Z}(S)\geq\mathit{out}_G(S')$. This implies that, for every two disjoint sets of vertices $X$ and $Y$ of $G_Z$, we have $\lambda_{G_Z}(X,Y)\geq\lambda_G(X',Y')$. (I.e., the process of contracting vertices can only increase the edge connectivity.)

\subsection{Minimum $(k+1)$-out sets}
In everything that follows, we will use ``$s$'' to denote a fixed ordinary vertex of $G$. We will use $s$ in order to consider minimum $(k+1)$-out sets that separate vertices from it. Specifically, let $v$ be a vertex with $\lambda(v,s)\geq k+1$. If there is a $(k+1)$-out set $S$ with $v\in S$ and $s\notin S$, then we let $M(v)$ denote the (inclusion-wise) minimum such $(k+1)$-out set. Otherwise, we let $M(v)=\bot$. We use ``$M_R$'' to denote the same concept as ``$M$'' in the reverse graph $G^R$.

The reason that we consider the $M$-sets is twofold. First, those sets are sufficient in order to determine the relation of $(k+1)$-edge-connectivity, as shown in \Cref{lemma:distinctM}.\footnote{A similar observation was utilized in \cite{Linear3ECC}, in order to compute the $3$-edge-connected components in linear time. Specifically, \cite{Linear3ECC} also uses a concept of $M$-sets, which are the minimum $2$-in sets that separate $s$ from a vertex. \Cref{lemma:distinctM} is analogous to Proposition III.5 in \cite{Linear3ECC}.} And second, if the $M$-set of a vertex exists, then it can be computed in time proportional to its volume using a local search procedure (see \Cref{proposition:local_search}).

\begin{lemma}
\label{lemma:distinctM}
Let $u$ and $w$ be two ordinary vertices that are not $(k+2)$-edge-connected. Then, either $M(u)\neq M(w)$, or $M_R(u)\neq M_R(w)$.
\end{lemma}
\begin{proof}
Since $u$ and $w$ are not $(k+2)$-edge-connected (but they are $(k+1)$-edge-connected), we may assume, w.l.o.g., that there is a $(k+1)$-cut $S$ that separates $u$ and $w$. Now there are two possibilities: either $s\notin S$, or $s\in S$. Let us first consider the case that $s\notin S$. Then, $S$ is a $(k+1)$-cut that separates $u$ and $s$, and therefore $M(u)$ exists and $M(u)\subseteq S$. Then, since $w\notin S$, we have $w\notin M(u)$. Therefore, we infer that either $M(w)$ does not exist, or $M(w)\neq M(u)$ (because, if $M(w)$ exists, it satisfies $w\in M(w)$). Now let us assume that $s\in S$. Then, $V\setminus S$ is a $(k+1)$-cut in $G^R$ that separates $w$ and $s$. Thus, we have $M_R(w)\subseteq V\setminus S$, and $u\notin M_R(w)$. Therefore, we infer, as previously, that either $M_R(u)$ does not exist, or $M_R(u)\neq M_R(w)$.  
\end{proof}

\begin{lemma}
\label{lemma:distinctM2}
Let $u$ and $w$ be two ordinary vertices such that $M(u)\neq M(w)$. Then, either $M(u)\neq\bot$ and $w\notin M(u)$, or $M(w)\neq\bot$ and $u\notin M(w)$. 
\end{lemma}
\begin{proof}
Since $M(u)\neq M(w)$, we cannot have $M(u)=\bot$ and $M(w)=\bot$. Thus, we may assume, w.l.o.g., that $M(u)\neq\bot$. Now, if $w\notin M(u)$, then we are done. So let us assume that $w\in M(u)$. Then, $M(u)$ is a $(k+1)$-out set that separates $w$ and $s$, and therefore $M(w)\neq\bot$. Then, due to the minimality of $M(w)$, we have $M(w)\subseteq M(u)$. Now, if we assume that $u\in M(w)$, then, due to the minimality of $M(u)$, we get $M(u)\subseteq M(w)$, and therefore we have $M(u)=M(w)$, which contradicts the assumption in the statement of the lemma. Thus, we conclude that $u\notin M(w)$.
\end{proof}

In order to compute $M$-sets of ``small'' volume (and, later on, separations of vertices induced by $M$-sets of ``large'' volume), we rely on the following local search procedure.

\localsearch*
\ignore{
\begin{restatable*}{proposition}{localsearch}
\label{proposition:local_search}
Let $v$ be a vertex with $v\neq s$ and $\lambda(v,s)\geq 3$, and let $\Delta\geq 1$ be an integer. There is an algorithm $\mathtt{LocalSearchFor3Cut}(G,v,s,\Delta)$ which runs in $O(\Delta)$ time and returns a set of vertices $S$ (which may be $\emptyset$) with the following guarantees:
\begin{itemize}
\item{If $S\neq\emptyset$, then $S=M(v)$.}
\item{If $S=\emptyset$, then either $M(v)=\bot$ or $\mathit{vol}(M(v))>\Delta$.}
\end{itemize}
\end{restatable*}
}

(Notice that the guarantees of \Cref{proposition:local_search} imply that if $M(v)\neq\bot$ and $\mathit{vol}(M(v))\leq\Delta$, then the output $S$ will be $M(v)$.)

We note that \Cref{proposition:local_search} follows from a straightforward adaptation of ideas contained in \cite{chechik2017faster}, for computing minimal sets of vertices with a bounded number of outgoing edges, through a local search procedure. In order to understand the relation between \Cref{proposition:local_search} and the local search procedure of \cite{chechik2017faster}, we recall the following concept from \cite{chechik2017faster}. Let $u$ be a vertex of a graph $G$, and let $k\geq 0$ be an integer. Then, a set of vertices $S$ is called a $k$-edge-out component of $u$ if $u\in S$ and $\mathit{out}(S)\leq k$, and there is no set of vertices $S'$ with $u\in S'\subset S$ such that $\mathit{out}(S')\leq\mathit{out}(S)$.\lnote{Maybe we should note that this concept differs from the he \outcomp{(\le k)} mentioned earlier.}
The number of edges of $G[S]$ is called the volume of $S$.  Now, given a vertex $u$ and two integers $k$ and $\Delta$, we have the following:


\begin{lemma}[\cite{chechik2017faster}]
\label{lemma:local-search-chilp}
In $O((2k)^{k+1}\cdot\Delta)$ time we can find a $(k-1)$-edge-out component of $u$ with volume less than $(2k-1)(\Delta+1)$, or determine that there is no $(k-1)$-edge-out component of $u$ with volume at most $\Delta$.
\end{lemma}
 
The idea behind \Cref{lemma:local-search-chilp} is roughly the following. If there is a $(k-1)$-edge-out component $S$ of $u$ with volume at most $\Delta$, then there is a DFS-based procedure, that starts from $u$, explores $O(k\Delta)$ edges, and finds a set of $O(k)$ paths with the property that at least one of them starts from $u$ and ends outside of $S$. Then, by reversing the direction of the edges of such a path, $S$ becomes a $(k-2)$-edge-out component in the resulting graph. Thus, if we repeat this process $O(k)$ times, in the end it is sufficient to just explore the reachability set of $u$ in the resulting graph, and this provides $S$.

There is an obvious similarity between the guarantees of \Cref{lemma:local-search-chilp} for $k=k+2$ and \Cref{proposition:local_search}. Specifically, it is easy to see that, if we have a vertex $v\neq s$ with $\lambda(v,s)=k+1$, then $M(v)$ is a $(k+1)$-edge-out component of $v$. Therefore, if $G[M(v)]$ has volume at most $\Delta$, then \Cref{lemma:local-search-chilp} will be able to identify a $(k+1)$-edge-out component of $v$ in $O(\Delta)$ time. However, the $(k+1)$-edge-out component of $v$ returned by \Cref{lemma:local-search-chilp} may have less than $k+1$ outgoing edges. This is because there may exist a $k'$-out set with $k'<k+1$ that contains both $v$ and $s$ and has sufficiently small volume. (In fact, the whole graph is a $0$-edge-out component of $v$, and thus for sufficiently large $\Delta$ we just get the whole graph as output.) But here it helps precisely that we know that $s$ is outside of $M(v)$. Thus, whenever we happen to meet $s$ during the (DFS-based) local search procedure, we can immediately pick and reverse the discovered path from $v$ to $s$. This is the crux of the adaptation that we had to make to the local search procedure from \cite{chechik2017faster}. The full proof of \Cref{proposition:local_search} is given in \Cref{sec:proposition:local_search}. 

\subsection{Latest mincuts}
Let $s$ and $t$ be two distinct vertices with $\lambda(s,t)=k$. By \Cref{corollary:minimal-exists}, there is an inclusion-wise maximum $k$-out set $S$ that separates $s$ and $t$. We call $S$ \emph{the latest $(s,t)$-mincut}. (We note that the order of $s$ and $t$ is important.) 

The intuitive reason that we consider latest mincuts is because these are the furthest that we can cut, without exceeding the value of the minimum cut; a precise formulation of this property is provided in \Cref{lemma:include_outgoing_edge}. Furthermore, the latest mincuts for pairs of vertices with bounded edge-connectivity can be computed in linear time, as shown in the following lemma.
We note that the notion of the latest $(s,t)$-mincut was introduced by Ford and Fulkerson~\cite{FordFulkerson}, who showed how to compute it using the residual graph of the corresponding flow problem. For completeness—and because we will later use similar arguments—we provide a self-contained proof that does not rely on flow-based techniques.

\begin{lemma}
\label{lemma:compute_latest}
Let $s$ and $t$ be two distinct vertices of $G$ with $\lambda(s,t)=k$. Then, the latest $(s,t)$-mincut can be computed in $O(km)$ time.
\end{lemma}
\begin{proof}
Let $S$ be the latest $(s,t)$-mincut. Now we apply the following procedure (essentially $k$ iterations of Ford-Fulkerson's algorithm for computing an $(s,t)$-maxflow). First, we compute a path $P_1$ in $G_0:=G$ from $s$ to $t$. Then we reverse the direction of the edges of $P_1$, and let $G_1$ be the resulting graph. Now suppose that we have computed a path $P_i$ from $s$ to $t$ in $G_{i-1}$ and a graph $G_i$ using the same process, for some $i\in\{1,\dots,k-1\}$. Then, we compute a path $P_{i+1}$ from $s$ to $t$ in $G_i$, we reverse the direction of the edges of $P_{i+1}$ in $G_i$, and let $G_{i+1}$ be the resulting graph. Finally, we compute the set of vertices $S'$ of $G_k$ that cannot reach $t$. Obviously, this whole procedure can be completed in $O(km)$ time.

Our goal is to show that $S'=S$.
First, a repeated application of \Cref{lemma:reverse} implies that $\mathit{out}_{G_k}(S)=0$. Thus, no vertex from $S$ can reach $t$ in $G_k$, and therefore $S\subseteq S'$. Since $S'$ is an $(s,t)$-cut and $S$ is the inclusion-wise maximum $k$-out set in $G$ that separates $s$ and $t$, we infer that $\mathit{out}_{G}(S')\geq k$. Now, a repeated application of \Cref{lemma:reverse} implies that $\mathit{out}_{G_k}(S')=\mathit{out}_{G}(S')-k$. Due to the definition of $S'$, we have $\mathit{out}_{G_k}(S')=0$. This implies that $\mathit{out}_{G}(S')=k$. Therefore, due to the maximality of $S$, we conclude that $S'=S$.
\end{proof}

\begin{lemma}
\label{lemma:reaching_boundaries}
Let $s$ and $t$ be two distinct vertices, and let $S$ be an $(s,t)$-mincut (not necessarily the latest). Then, every boundary point of $S$ in $G$ is reachable by $s$ through a path in $G[S]$.
\end{lemma}
\begin{proof}
Let $R$ be the set of vertices that are reachable by $s$ in $G[S]$. Thus, we have $s\in R$ and $R\subseteq S$. Therefore, $R$ is an $(s,t)$-cut, and therefore $\mathit{out}_{G}(R)\geq\lambda(s,t)$. Furthermore, by the definition of $R$ we have that every outgoing edge of $R$ in $G$ is an outgoing edge of $S$ in $G$. Thus, $\mathit{out}_{G}(R)\leq\mathit{out}_{G}(S)$. Then, since $\mathit{out}_{G}(S)=\lambda(s,t)$, we infer that $\mathit{out}_{G}(R)=\lambda(s,t)$. This shows that the set of the outgoing edges of $R$ in $G$ coincides with the set of the outgoing edges of $S$ in $G$. We conclude that $R$ includes the tails of the outgoing edges of $S$ in $G$ (i.e., the boundary points of $S$).
\end{proof}

For convenience, we use the following notation. Let $v$ be a vertex of a graph $H$. Then $\mathit{reach}_H(v)$ denotes the set of vertices that are reachable from $v$. In other words, $\mathit{reach}_H(v)$ consists of every vertex $x$ of $H$ for which there exists a path that starts from $v$ and ends in $x$.

\begin{lemma}
\label{lemma:include_outgoing_edge}
Let $s$ and $t$ be two distinct vertices, let $S$ be the latest $(s,t)$-mincut, let $R=\mathit{reach}_{G[S]}(s)$, and let $x$ be an exit point of $S$ with $x\neq t$. Then $\lambda(R\cup\{x\},t)>\mathit{out}(S)$.
\end{lemma}
\begin{proof}
Notice that $R\cup\{x\}$ is an $(s,t)$-cut, and therefore $\lambda(R\cup\{x\},t)\geq\mathit{out}(S)$. Now let us suppose, for the sake of contradiction, that $\lambda(R\cup\{x\},t)=\mathit{out}(S)$. This implies that there is an $(R\cup\{x\},t)$-cut $R'$ with $\mathit{out}(R')=\mathit{out}(S)$. Our goal is to show that $S\cup R'$ is an $(s,t)$-cut with $\mathit{out}(S\cup R')=\mathit{out}(S)$, which contradicts the fact that $S$ is the latest $(s,t)$-mincut (because $S\cup R'$ is strictly larger than $S$, since it contains $x$). Throughout this proof, whenever we speak of an outgoing edge of a set of vertices, we mean an outgoing edge in $G$.

Let $U$ be the set of vertices from $S$ that are unreachable by $s$ in $G[S]$. (I.e., $U=S\setminus R$.) Then, we have $S=R\cup U$. Notice that every outgoing edge of $U$ has its head in $R$. (Otherwise, there exists an edge $(z,w)$ with $z\in U$ and $w\notin S$. Thus, $(z,w)$ is an outgoing edge of $S$. But then, since $S$ is an $(s,t)$-mincut, by \Cref{lemma:reaching_boundaries} we have that $s$ can reach $z$ in $G[S]$, in contradiction to the definition of $U$.)

Now, since $S=R\cup U$ and $R\subset R'$, we have $S\cup R' = U\cup R'$. Let $(z,w)$ be an outgoing edge of $S\cup R'$. Then, either $z\in U$ or $z\in R'$. If $z\in U$, then $(z,w)$ is an outgoing edge of $U$, and so $w\in R$. But $R\subset R'$, and so $(z,w)$ cannot be an outgoing edge of $S\cup R'$. Thus, we have $z\in R'$. Then, $(z,w)$ is an outgoing edge of $R'$. This shows that the set of the outgoing edges of $S\cup R'$ is a subset of the set of the outgoing edges of $R'$. Thus, we have $\mathit{out}(S\cup R')\leq\mathit{out}(R')$. But since $\mathit{out}(R')=\mathit{out}(S)$, this implies the desired contradiction. 
\end{proof}

\section{Finding good partitions induced by the large $M$-sets}

In this section we assume that $G$ is a strongly connected graph, where some of its vertices are designated as ``ordinary'' and they are $(k+1)$-edge-connected.  Let $s$ be a fixed ordinary vertex. For every ordinary vertex $u$ with $\lambda(u,s)=k+1$, we let $M(u)$ denote the inclusion-wise minimum $(k+1)$-out set that separates $u$ and $s$.

Our goal is to establish the following three propositions.  

\mainordinary*

\mainauxiliary*

\mainaux*

\ignore{
\begin{restatable*}{proposition}{mainordinary}
\label{proposition:main_ordinary}
Let $v$ be a vertex with $\lambda(v,s)=3$. Then there is an algorithm that runs in $O(m)$ time and returns a partition $\mathcal{P}$ of $V$ with the following properties:
\begin{itemize}
\item{$\mathcal{P}$ maintains the $4$-edge-connected components.}
\item{For every ordinary vertex $u$ with $v\in M(u)$, and every ordinary vertex $w\notin M(u)$, we have that $u$ and $w$ are separated by $\mathcal{P}$.}
\end{itemize}
\end{restatable*}

\begin{restatable*}{proposition}{mainauxiliary}
\label{proposition:main_auxiliary}
Let $v$ be a vertex with $\lambda(v,s)<3$. Then there is an algorithm that runs in $O(m)$ time and returns a partition $\mathcal{P}$ of $V$ with the following properties:
\begin{itemize}
\item{$\mathcal{P}$ maintains the $4$-edge-connected components.}
\item{For every ordinary vertex $u$ with $v\in M(u)$, and every ordinary vertex $w\notin M(u)$, we have that $u$ and $w$ are separated by $\mathcal{P}$.}
\end{itemize}
\end{restatable*}
}

The proofs of \Cref{proposition:main_ordinary,proposition:main_auxiliary} are provided in \Cref{section:proof_of_prop_main_ordinary,section:proof_of_prop_main_auxiliary}, where they appear as \Cref{proposition:13_2} and \Cref{proposition:14_2}, respectively, \sidenote{Charis: added a brief explanation for the different statemens.} because the corresponding properties of the computed partitions are explained formally in their statements. 
The proof of \Cref{proposition:main_auxiliary_4} is discussed in \Cref{section:proof_of_prop_main_auxiliary}, because this is established by using the same techniques as in \Cref{proposition:14_2}.

\subsection{Proof of \Cref{proposition:main_ordinary}
(using the Picard-Queyranne graph)}
\label{section:proof_of_prop_main_ordinary}
Let $a$ and $b$ be two distinct vertices of $G$. If $G$ is a multigraph, then we consider it as a simple capacitated graph, with capacity function $c:E(G)\rightarrow\mathbb{N}$, where $c(e)$ for every edge $e$ equals the multiplicity of $e$ in the original graph. Let us define an $(a,b)$-flow $f$ as a function $f:E(G)\rightarrow\mathbb{Z}$ with the following two properties:

\begin{enumerate}
\item{$0\leq f(e)\leq c(e)$, for every edge $e$ of $G$. (I.e., the flow passing through every edge cannot exceed its capacity.)}
\item{$\sum_{e=(u,v)}{f(e)}=\sum_{e=(v,u)}{f(e)}$, for any vertex $v\in V(G)\setminus\{a,b\}$. (I.e., the amount of flow entering any vertex $v \notin\{a,b\}$ equals the amount of flow exiting $v$.)} 
\end{enumerate}  

Without loss of generality, we may assume that no flow enters the source vertex $a$ and no flow leaves the sink vertex $b$, i.e.,
$f(e)=0$ for any edge $e=(v,s)$ and any edge $e=(t,v)$. \sidenote{Loukas: Added that there is no flow entering the source and exiting the sink.}

We call $\sum_{e=(a,u)}{f(e)}$ the value of $f$. (I.e., the value of $f$ is the amount of flow exiting $a$.) It is well known that an $(a,b)$-flow with value $\lambda(a,b)$ exists, and this is in fact the maximum value of any $(a,b)$-flow. We call such a flow an \emph{$(a,b)$-maxflow}.

\begin{definition}[Picard-Queyranne graph]
\label{definition:pq}
\normalfont
Let $f$ be an $(a,b)$-maxflow of $G$. Then we define a graph $\mathit{PQ}$ with $V(\mathit{PQ})=V(G)$ as follows. For every edge $e=(x,y)$ of $G$, there is an edge $(x,y)$ in $\mathit{PQ}$ if $f(e)<c(e)$, and an edge $(y,x)$ if $f(e)>0$. $\mathit{PQ}$ is called the Picard-Queyranne graph that corresponds to $f$.
\end{definition}

\begin{figure}[ht!]
\begin{center}
\includegraphics[scale=1, clip=true, trim={0cm 0cm 0cm 7cm}, width=\textwidth]{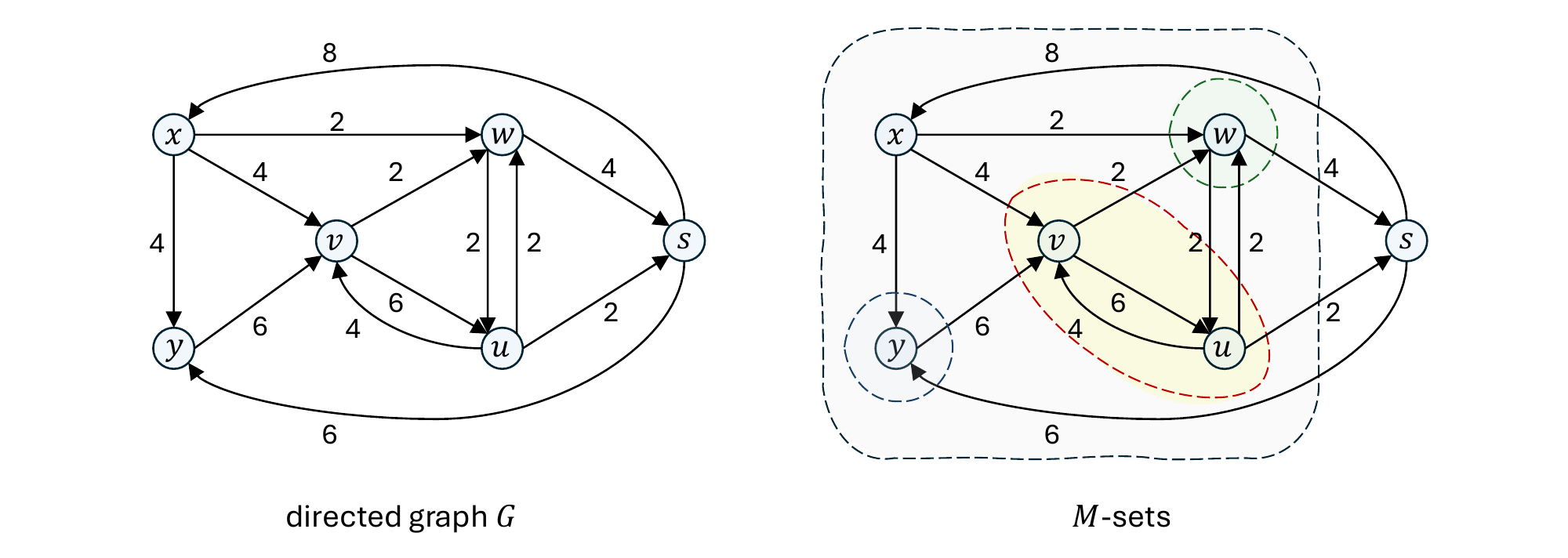}
\includegraphics[scale=1, clip=true, trim={0cm 0cm 0cm 7cm}, width=\textwidth]{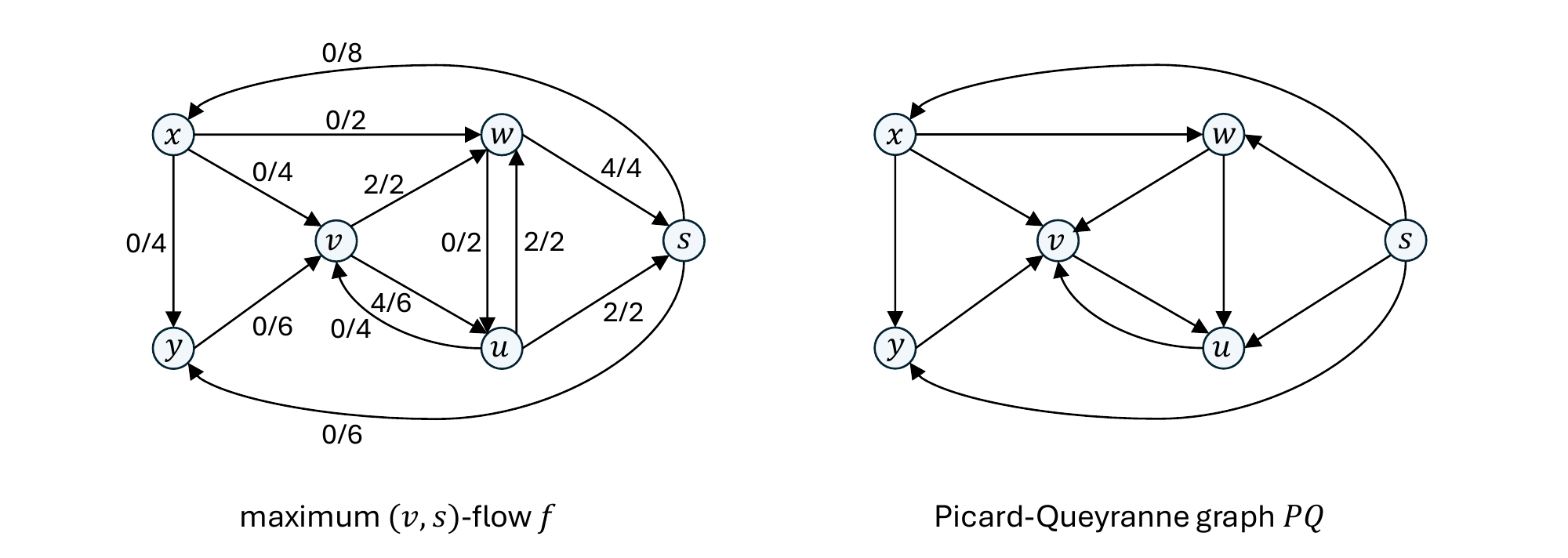}
\end{center}
\caption{Top: A directed graph $G$ with source vertex $s$, such that $\lambda(z,s) =6$, for all vertices $z\not=s$. The $M$-sets of $G$ are $M(v)=M(u)=\{v,u\}$, $M(w)=\{w\}$, $M(y)=\{y\}$, and $M(x)=\{x,y,u,v,w\}$.
Bottom: A maximum $(v,s)$-flow and the corresponding Picard-Queyranne graph $\mathit{PQ}$. We have $x,y,w \not\in M(u)$, so these vertices are not strongly connected with $u$ in $\mathit{PQ}$. 
\label{figure:PQgraph}}
\end{figure}

We note that, since $G$ is strongly connected, it is not difficult to see that $a$ is reachable by all vertices in $\mathit{PQ}$, and $b$ reaches all vertices in $\mathit{PQ}$. (Thus, if we contract every strongly connected component of $\mathit{PQ}$ into a single vertex, then the node containing $b$ is the only source, and the node containing $a$ is the only sink in the resulting DAG.)

Picard and Queyranne~\cite{PicardQ82} have provided the following characterization of the $(a,b)$-mincuts of $G$:

\begin{proposition}[Theorem 1 in \cite{PicardQ82}]
\label{proposition:pq}
A set of vertices $X$ of $G$ is an $(a,b)$-mincut if and only if: it contains $a$, it does not contain $b$, and it is closed w.r.t. the reachability relation in $\mathit{PQ}$ (i.e., for every $x\in X$, we have $\mathit{reach}_{PQ}(x)\subseteq X$).
\end{proposition}

We will prove the following. (See \Cref{figure:PQgraph}.)

\begin{proposition}
\label{proposition:13_2}
Let $v$ be a vertex with $\lambda(v,s)=k+1$. There is an algorithm that runs in $O(km)$ time and returns a partition $\mathcal{P}$ of $V(G)$ with the following properties:
\begin{itemize}
\item{$\mathcal{P}$ maintains the $(k+2)$-edge-connected components of $G$.}
\item{For every two ordinary vertices $u$ and $w$ with $v\in M(u)$ and $w\notin M(u)$, we have that $u$ and $w$ are separated by $\mathcal{P}$.}
\end{itemize}
\end{proposition}
\begin{remark}
\normalfont
We note that the second point holds even if we replace ``$M(u)$'' with ``a $(k+1)$-out set $U$ that separates $u$ and $s$''. The proof is precisely the same (i.e., every instance of ``$M(u)$'' is just replaced with ``$U$'').    
\end{remark}
\begin{proof}
First we perform $k+1$ iterations of Ford-Fulkerson's augmenting paths algorithm (e.g. using BFS), in order to find a maximum $(v,s)$-flow. Then we construct the corresponding graph $\mathit{PQ}$ as described in \Cref{definition:pq}, and we let $\mathcal{P}$ be the collection of the strongly connected components (viewed as vertex-sets) of $\mathit{PQ}$. Thus, the computation of $\mathcal{P}$ takes $O(km)$ time.  

Now let $u$ and $w$ be two vertices of $G$ that are separated by $\mathcal{P}$. Then, at least one of those vertices is unreachable by the other in $\mathit{PQ}$. 
Thus, we may assume w.l.o.g. that $u$ cannot reach $w$ in $\mathit{PQ}$. This implies that $u$ cannot reach the strongly connected component of $s$ in $\mathit{PQ}$ (because $s$ reaches all vertices of $\mathit{PQ}$), and therefore $X=\mathit{reach}_{PQ}(u)$ is a $(v,s)$-mincut of $G$. Since $\lambda(v,s)=k+1$, we have that $X$ is a $(k+1)$-out set of $G$. Furthermore, $X$ separates $u$ and $w$. Thus, $u$ and $w$ are not $(k+2)$-edge-connected in $G$. This establishes the property that $\mathcal{P}$ maintains the $(k+2)$-edge-connected components of $G$.

Now let $u$ and $w$ be two ordinary vertices of $G$ with $v\in M(u)$ and $w\notin M(u)$. This implies that $M(u)$ is a $(k+1)$-out set that separates $u$ and $s$. Therefore, since $v\in M(u)$, we have that $M(u)$ is a $(v,s)$-mincut. According to \Cref{proposition:pq}, this implies that $\mathit{reach}_{PQ}(u)\subseteq M(u)$. Thus, since $w\notin M(u)$, we have that $u$ cannot reach $w$ in $\mathit{PQ}$, and therefore $u$ and $w$ are separated by $\mathcal{P}$. 
\end{proof}

\subsection{Proof of \Cref{proposition:main_auxiliary}}
\label{section:proof_of_prop_main_auxiliary}

We will prove the following.

\begin{proposition}
\label{proposition:14_2}
Let $v$ be a vertex with $\lambda(v,s)=k$. There is an algorithm that runs in $O(k^2m)$ time and returns a partition $\mathcal{P}$ of $V(G)$ with the following properties:
\begin{itemize}
\item{$\mathcal{P}$ maintains the $(k+2)$-edge-connected components of $G$.}
\item{For every two ordinary vertices $u$ and $w$ with $v\in M(u)$ and $w\notin M(u)$, we have that $u$ and $w$ are separated by $\mathcal{P}$.}
\end{itemize}
\end{proposition}
\begin{proof}
We apply the procedure shown in \Cref{algorithm:14_2}. Thus, first we compute the latest $(v,s)$-mincut $S$. By \Cref{lemma:compute_latest}, this takes $O(km)$ time. Notice that $S$ does not contain any ordinary vertex of $G$ (because the ordinary vertices are $(k+1)$-edge-connected, and $S$ is a $k$-out set). 

Then, we contract $S$ into a single vertex $z$, and let $\widetilde{G}$ be the resulting graph. It is easy to establish the following two facts: $(i)$ for every two vertices $x$ and $y$ of $G$ with $x,y\notin S$, we have $\lambda_{\widetilde{G}}(x,y)\geq\lambda_{G}(x,y)$, and $(ii)$ the out-degree of $z$ in $\widetilde{G}$ is $k$. Let $\{e_1,\dots,e_k\}$ be the set of the outgoing edges of $z$ in $\widetilde{G}$.

Now, as shown in Line~\ref{line:14_2_1}, we first remove all outgoing edges of $z$ from $\widetilde{G}$, and we compute the partition $\mathcal{Q}_0$ of the $2$-edge-connected components of the resulting graph, that we denote by $\widetilde{G}'$. By \cite{georgiadis20162edge}, this takes linear time. Let $\mathcal{P}_0$ be the partition of $V(G)$ that corresponds to $\mathcal{Q}_0$. (I.e., every set of $\mathcal{Q}_0$ is a set of $\mathcal{P}_0$, except for the set $U\in\mathcal{Q}_0$ that contains $z$, which is replaced by $(U\setminus\{z\})\cup S$.) Let $x$ and $y$ be two vertices of $G$ with $x,y\notin S$ and $\lambda_{G}(x,y)\geq k+2$. Then, by $(i)$ we infer that $\lambda_{\widetilde{G}'}(x,y)\geq 2$. Thus, $\mathcal{P}_0$ maintains the $(k+2)$-edge-connected components of $G$. 

Then, for every $i\in\{1,\dots,k\}$ such that the head $x_i$ of $e_i$ is not $s$, we contract $z$ with $x_i$ into a single vertex $z_i$, and let $\widetilde{G}_i$ be the resulting graph (derived from $\widetilde{G}$). Notice that $\widetilde{G}_i$ is essentially the same graph as if we had contracted the vertex set $S\cup\{x_i\}$ of $G$ into $z_i$. Thus, it is not difficult to see that \Cref{lemma:include_outgoing_edge} implies that $\lambda_{\widetilde{G}_i}(z_i,s)>\lambda_{G}(v,s)=k$. 
(To be precise, \Cref{lemma:include_outgoing_edge} implies that $\lambda_{G}(R\cup\{x_i\},s)>k$, where $R=\mathit{reach}_{G[S]}(v)$. But then, since $R\cup\{x_i\}\subseteq S\cup\{x_i\}$, we have $\lambda_{G}(R\cup\{x_i\},s)\leq\lambda_{G}(S\cup\{x_i\},s)$. Thus, we get $\lambda_{G}(S\cup\{x_i\},s)>k$, which implies that $\lambda_{\widetilde{G}_i}(z_i,s)>k$.) 
Now, if $\lambda_{\widetilde{G}_i}(z_i,s)=k+1$, then we can apply \Cref{proposition:13_2} on $\widetilde{G}_i$ (with $v=z_i$), in order to get a partition $\mathcal{Q}_i$. By \Cref{proposition:13_2}, we have that $\mathcal{Q}_i$ maintains the $(k+2)$-edge-connected components of $\widetilde{G}_i$. Then, since $\widetilde{G}_i$ is essentially derived from $G$ by contracting $S\cup\{x_i\}$ into a single vertex, we have that the partition $\mathcal{P}_i$ of $V(G)$ that corresponds to $\mathcal{Q}_i$ maintains the $(k+2)$-edge-connected components of $G$. 
 
The output of \Cref{algorithm:14_2} is the partition $\mathcal{P}$ which is the common refinement of $\mathcal{P}_0$ and all $\mathcal{P}_i$, for $i\in\{1,\dots,k\}$ such that the head $x_i$ of $e_i$ is not $s$ and $\lambda_{G}(S\cup\{x_i\},s)=k+1$. Therefore, so far we have established that $\mathcal{P}$ maintains the $(k+2)$-edge-connected components of $G$.

Now let $u$ and $w$ be two ordinary vertices of $G$ such that $v\in M(u)$ and $w\notin M(u)$. Consider the set of vertices $U=M(u)\cup S$. As $w$ is an ordinary vertex, we have $w \notin U$. 
We will show that $\mathit{out}_{G}(U)=k+1$. 
By the submodularity of the cut function we have 
$$
\mathit{out}_{G}(M(u)\cap S)+\mathit{out}_{G}(M(u)\cup S)\leq\mathit{out}_{G}(M(u))+\mathit{out}_{G}(S).
$$
We have $\mathit{out}_{G}(M(u))=k+1$ and $\mathit{out}_{G}(S)=k$. On the other hand, we have $\mathit{out}_{G}(M(u)\cap S)\geq k$ (because $M(u)\cap S$ is a $(v,s)$-cut), and $\mathit{out}_{G}(M(u)\cup S)\geq k+1$ (because $M(u)\cup S$ is a $(u,s)$-cut). Thus, we infer that $\mathit{out}_{G}(M(u)\cup S)=k+1$.

Now consider the projection $\widetilde{U}$ of $U$ from $G$ into $\widetilde{G}$. Then, since $S\subset U$, we have $\mathit{out}_{\widetilde{G}}(\widetilde{U})=\mathit{out}_{G}(U)$, and therefore $\widetilde{U}$ is a $(k+1)$-out set of $\widetilde{G}$. Notice that $u,z\in\widetilde{U}$ and $w\notin\widetilde{U}$. Now, if the outgoing edges of $z$ in $\widetilde{G}$ are outgoing edges of $\widetilde{U}$, then their removal from $\widetilde{G}$ drops the connectivity from $u$ to $w$ to $1$. Thus, $u$ and $w$ are separated by $\mathcal{Q}_0$, and therefore by $\mathcal{P}_0$, and therefore by $\mathcal{P}$. Otherwise, we have that at least one outgoing edge $e_i$ of $z$ in $\widetilde{G}$ lies entirely within $\widetilde{U}$. Then, since we contract $z$ with the head of $e_i$ into a vertex $z_i$ in order to get $\widetilde{G}_i$, it is easy to see that the projection of $\widetilde{U}$ into $\widetilde{G}_i$ is a $(k+1)$-out set of $\widetilde{G}_i$ that contains $z_i$, and does not contain $w$ or $s$. Then, by \Cref{proposition:13_2} we infer that $u$ is separated from $w$ by $\mathcal{P}_i$, and therefore by $\mathcal{P}$.

Notice that we get the $O(k^2m)$ time bound because we apply at most $k$ times the algorithm that establishes \Cref{proposition:13_2} on graphs that have at most as many edges as $G$.
\end{proof}

In order to establish \Cref{proposition:main_auxiliary_4}, we can use \Cref{algorithm:14_2} with the following modifications (and the proof of correctness is similar to that of \Cref{proposition:14_2}). First, if $\lambda(v,s)=2$, then we just apply this algorithm with $k=2$. However, if $\lambda(v,s)=1$, then the vertex $z$, that appears in Line~\ref{line:algorithm14_2_z} (and is the contraction of the latest $(v,s)$-mincut $S$), has a unique outgoing edge $e_1$. Thus, in Line~\ref{line:14_2_1} we perform a computation of the $3$-edge-connected components of $\widetilde{G}\setminus\{e_1\}$. This can be performed in linear time, by \cite{Linear3ECC}. Then, we have $\lambda'=\lambda_{G}(S\cup\{x_1\},s)>1$, where $x_1$ is the head of $e_1$. And now, if $\lambda'\leq 3$, we work similarly as in Lines~\ref{line:14_2_a} to \ref{line:14_2_b}. The only difference is that, if $\lambda'=2$, then we apply again \Cref{proposition:main_auxiliary_4}, instead of \Cref{proposition:13_2}.

   
\begin{algorithm}[h!]
\caption{\textsf{The procedure that establishes \Cref{proposition:14_2}.}}
\label{algorithm:14_2}
\LinesNumbered
\DontPrintSemicolon

\tcp{input is a vertex $v$ with $\lambda(v,s)=k$}
compute the latest $(v,s)$-mincut $S$\;
contract $S$ into a vertex $z$, and let $\widetilde{G}$ be the resulting graph\;
\label{line:algorithm14_2_z}
let $\{e_1,\dots,e_k\}$ be the outgoing edges of $z$ in $\widetilde{G}$\;
compute the partition $\mathcal{Q}_0$ of the $2$-edge-connected components of $\widetilde{G}\setminus\{e_1,\dots,e_k\}$\;
\label{line:14_2_1}
let $\mathcal{P}_0$ be the partition of $V(G)$ that corresponds to $\mathcal{Q}_0$ (via the inverse of the quotient map from $G$ to $\widetilde{G}$)\;
\ForEach{$i\in\{1,\dots,k\}$}{
  \If{the head $x_i$ of $e_i$ is not $s$ \textbf{and} $\lambda_{G}(S\cup\{x_i\},s)=k+1$}{
    let $\widetilde{G}_i$ be the graph obtained from $\widetilde{G}$ by contracting $z$ and $x_i$ into a vertex $z_i$\;
    \label{line:14_2_a}
    let $\mathcal{Q}_i$ be the partition that is returned by applying \Cref{proposition:13_2} on $\widetilde{G}_i$, with $v:=z_i$\;
    let $\mathcal{P}_i$ be the partition of $V(G)$ that corresponds to $\mathcal{Q}_i$\;
    \label{line:14_2_b}
  }
  \Else{
    let $\mathcal{P}_i$ be $\{V(G)\}$\;
  }
}

\textbf{return} $\mathcal{P}_0$ refined by $\mathcal{P}_1,\dots,\mathcal{P}_k$\; 

\end{algorithm}

%% file: decompositions.tex
\section{Edge-connected components decompositions}
\label{sec:decompositions}

Our main algorithm (\Cref{algorithm:main}) in \Cref{section:main-result} assumes that the ordinary vertices of the input digraph are $(k+1)$-edge-connected. 
In this section, we justify this assumption by proving \Cref{theorem:keccdecomp} and \Cref{theorem:3eccdecomp}. 
To this end, we leverage the computation of the nodes corresponding to a poset representation of the minimal \outcomp{k}s~\cite{Gabow:Poset:TALG,Linear3ECC}, in an appropriate order, combined with either a contraction operation for the \ecc{(k+1)} decomposition or the gadget substitution operation of~\cite{GKPP:3ECC} for the \ecc{3} decomposition.

\input{kecc-decomposition}

\input{3ecc-decomposition}

%% file: kecc-decomposition.tex
\subsection{$(k+1)$-ECC decomposition of a $k$-edge connected digraph}
\label{sec:kecc-decomposition}

\ignore{
Our main algorithm of \Cref{section:main-result} assumes that the ordinary vertices of the input digraph are $(k+1)$-edge-connected.  
Here we justify this assumption by proving \Cref{theorem:keccdecomp}.

\begin{theorem}[$k$-Edge-Connected Component (\ecc{k}) Decomposition]
\label{theorem:keccdecomp}
Let $G$ be a $k$-edge-connected digraph with $m$ edges. In $\widetilde{O}(m)$ time\lnote{This bound refers to constant $k$. We may want to express the exact dependence on $k$}, we can construct a collection $H_1,\dots,H_t$ of graphs such that:
\begin{itemize}
\item{The vertices $V(H_i)$ of each graph $H_i$ are partitioned into two sets of vertices, ordinary and auxiliary.}
\item{For each $i \in \{1,\ldots,t\}$, the ordinary vertices of $V(H_i)$ are $(k+1)$-edge-connected.}
\item{For every vertex of $G$, there is exactly one graph among $H_1,\dots,H_t$ that contains it as an ordinary vertex.}
\item{Every two vertices $u$ and $v$ of $G$ are $(k+2)$-edge-connected if and only if there is an $i\in\{1,\dots,t\}$ such that $u$ and $v$ are $(k+2)$-edge-connected ordinary vertices of $H_i$.}
\end{itemize}
\end{theorem}
}

Let $G=(V,E)$ be a $k$-edge-connected digraph. 
Let $s$ be an arbitrarily selected vertex of $G$.
Recall that a \outcomp{k} (resp., \incomp{k}) $S$ is a set of vertices such that $out(S) = k$ (resp., $in(S) = k$).
For any ordinary vertex $v$, we let \mout{v} denote the minimal \outcomp{k} that contains $v$ but not $s$.\footnote{We use the notation \mout{v} to denote the minimal \outcomp{k}s, in order to distinguish it from $M(v)$, which denotes the minimal \outcomp{(k+1)}s.} 
It follows from~\Cref{lemma:submodularity3} that if any $k$-out set separating $v$ from $s$ exists, then there also exists a minimum one, contained in every other such set (notice that $v,s$ are \conn{k}, therefore $\lambda(u,s)\ge k$).
If $v$ is not contained in any such set then we let $\moutM{v} = \bot$.

Let $[v]$ denote the set of ordinary vertices that have the same minimal \outcomp{k} in $G$, i.e., $[v] = \{ u : \moutM{v} = \moutM{u}\}$. 
Similarly, $[v]_R$ denotes the set of ordinary vertices that have the same minimal \outcomp{k} in $G^R$.
Note that the $[v]$-sets form a partition of the vertices of $G$. In particular, $v \in [s]$ implies $\moutM{v} = \bot$, i.e., there is no cut of size at most $k$ that separates all paths from $v$ to $s$.

\ignore{
\begin{lemma}
\label{lemma:kcut-induces-kout}
Let $v$ be an ordinary vertex, and let $C$ be a set of edges with $|C|=k$ such that $v$ cannot reach $s$ in $G \setminus C$. Then $C=E(S, V\setminus S)$, where $S$ is the set of vertices that $v$ reaches in $G\setminus C$.
\end{lemma}
\begin{proof}
Let $C=\{(x_1,y_1),\ldots,(x_k,y_k)\}$. Since $G$ is $k$-edge-connected and $v$ cannot reach $s$ in $G\setminus C$, we have that all paths from $v$ to $s$ in $G \setminus (C \setminus (x_i,y_i))$, for any $i \in \{1,\ldots,k\}$, use the edge $(x_i,y_i)$. From this we infer that $x_i \in S$ and $y_i \not\in S$, and therefore $(x_i,y_i)\in E(S,V\setminus S)$. Now, if there is an edge $(x_{k+1},y_{k+1})\in E(S,V\setminus S)$ such that $(x_{k+1},y_{k+1})\notin C$, then, by definition of $S$, we have that $v$ can reach $x_{k+1}$ in $G\setminus C$, but $v$ cannot reach $y_{k+1}$ in $G\setminus C$. But then, since $(x_{k+1},y_{k+1})\notin C$, we have that $x_{k+1}$ can reach $y_{k+1}$ in $G\setminus C$, a contradiction. This shows that $E(S,V\setminus S)=C$.
\end{proof}
}

\begin{proposition}
\label{proposition:minimal-in-out-sets}
Let $G$ be a $k$-edge-connected digraph with a fixed start vertex $s$. Then, any two vertices $u$ and $v$ are $(k+1)$-edge-connected if and only if $[u]=[v]$ and $[u]_R=[v]_R$.
\end{proposition}
\begin{proof}
Follows immediately from \Cref{lemma:distinctM} and the definition of the sets $[v]$ and $[v]_R$,
\ignore{
By the definition of the sets $[v]$ and $[v]_R$, it suffices to show that $u \leftrightarrow_{\mathit{k+1}} v$ if and only if $M(u)=M(v)$ and $M_R(u)=M_R(v)$.

Suppose that $u \leftrightarrow_{\mathit{k+1}} v$. Let $S$ be a set of vertices such that $|E(S,V\setminus S)|=k$. Then, since $u$ and $v$ are $k$-edge-connected, we have that $u\in S$ if and only if $v\in S$. This implies that $M(u)=M(v)$. Similarly, we have $M_R(u)=M_R(v)$.

Now assume that $M(u)=M(v)$ and $M_R(u)=M_R(v)$. Let us suppose, for the sake of contradiction, that $u$ and $v$ are not $k$-edge-connected. Then, $u$ and $v$ are separated by a $k$-edge cut. Thus, we may assume w.l.o.g. that there is a set of edges $C$ with $|C|=k$ such that $u$ cannot reach $v$ in $G\setminus C$. 
First, let us assume that $s$ can reach $v$ in $G\setminus C$. Then $u$ cannot reach $s$ in $G\setminus C$. Now let us consider the set of vertices $S$ that $u$ reaches in $G\setminus C$. Then, \Cref{lemma:kcut-induces-kout} implies that $S$ is a \outcomp{k} that contains $u$. Notice that $v\notin S$. But then, due to the minimality of $M(v)$, we infer that $v \notin M(u)$, in contradiction to $M(u)=M(v)$. This shows that $s$ cannot reach $v$ in $G\setminus C$.
Therefore, $v$ cannot reach $s$ in $G^R\setminus C^R$. Since $u$ cannot reach $v$ in $G\setminus C$, we have that $v$ cannot reach $u$ in $G^R\setminus C^R$. But then, with a similar argument, we arrive at a contradiction to $M_R(u)=M_R(v)$. We conclude that $u$ and $v$ are $k$-edge-connected.
}
\end{proof}

\begin{definition}
\label{definition:proper}
An ordering $\prec$ of the sets $[v]$ is \emph{proper} if, for all vertices $x,y$, $\moutM{x} \subset \moutM{y}$ implies $[x] \prec [y]$. 
\end{definition}

\begin{proposition}
\label{proposition:proper}
The sets $[v]$ can be computed in proper order in $O(km \log n)$ time.
\end{proposition}
\begin{proof}
The sets $[v]$ correspond to the strongly connected components (SCCs) of the \emph{labeling graph} $\mathit{LG}$ introduced by Gabow~\cite{Gabow:Poset:TALG}. 
This graph satisfies the property that $u \in \moutM{v}$ if and only if $u$ is reachable from $v$ in $\mathit{LG}$. 
Moreover, its SCCs can be computed in topological order in $O(m)$ time (Lemma~2.16 in~\cite{Gabow:Poset:TALG}). 
The labeling graph $\mathit{LG}$ is defined with respect to a \emph{complete $k$-intersection} of a $k$-edge-connected digraph, which can be constructed in $O(km \log n)$ time~\cite{edge_connectivity:gabow}.
\end{proof}

Let $C_1,\ldots,C_q$ be the $(k+1)$-edge-connected components of $G$.
Our plan is to compute, for each component $C_i$, a graph $G_i$, where the vertices in $C_i$ are ordinary, that satisfies \Cref{theorem:keccdecomp}.
We construct this decomposition in two phases, where the first phase processes the sets $[v]$, corresponding to the minimal \outcomp{k}s of $G$, in proper order, while the second phase processes the sets $[v]_R$, corresponding to the minimal \outcomp{k}s of $G^R$ in proper order.

\begin{definition}\label{def:contraction}
Let $S$ be an \outcomp{k} of a $k$-edge-connected digraph $G$.  
The \emph{contraction of $S$}, denoted by $\langle S \rangle$, is the operation that contracts all vertices in $S$ into a single vertex $v_S$. Specifically:
\begin{itemize}
    \item every edge $(u,v)$ entering $S$, where $u \notin S$ and $v \in S$, is replaced by an edge $(u, v_S)$;
    \item every edge $(v,w)$ leaving $S$, where $v \in S$ and $w \notin S$, is replaced by an edge $(v_S, w)$;
    \item all edges with both endpoints in $S$ are deleted.
\end{itemize}
We denote by $\contract{S}$ the graph obtained from $G$ after the contraction of $S$.
\end{definition}

Our goal is to show that the contraction of a \outcomp{k} $S$ in $G$ maintains the pairwise connectivity of the remaining vertices in $\bar{S}=V \setminus S$. 
For convenience, we will work with a modified version of $G$, where we split the outgoing edges of $S$ by introducing auxiliary intermediate vertices. More formally, we have the following.

\begin{definition}
Let $S$ be a \outcomp{k} of a $k$-edge-connected digraph $G$, with outgoing edges $C=E(S,V \setminus S) = \{(x_1,y_1), \ldots, (x_k,y_k)\}$. 
The \emph{splitting operation on $C$} introduces, for each edge $e_i = (x_i, y_i) \in C$, a new vertex $x(e_i)$, and replaces $e_i$ with the two edges $(x_i, x(e_i))$ and $(x(e_i), y_i)$. 
We denote by \splitop{G}{C} the resulting graph obtained after performing this operation.
\end{definition}

\begin{lemma}
\label{lemma:splitting}
Let $S$ be a \outcomp{k} of a $k$-edge-connected digraph $G$.
Let $u$ and $v$ be any two vertices of $G$. Then, for any $\ell \ge k$, $u$ and $v$ are $\ell$-edge-connected in $G$ if and only if $u$ and $v$ are $\ell$-edge-connected in \splitop{G}{E(S,V \setminus S)}.
\end{lemma}
\begin{proof}
The lemma holds since splitting an edge $(x_i,y_i)$ does not affect the local edge connectivity $\lambda(u,v)$ between any two vertices $u$ and $v$ of $G$. 
\end{proof}

\begin{lemma}
\label{lemma:kout-split}
Let $S$ be a \outcomp{k} of a $k$-edge-connected digraph $G$, with outgoing edges $C=E(S,V\setminus S)=\{(x_1,y_1),\ldots,(x_k,y_k)\}$.
Also, let $G'=\splitop{G}{C}$ and let 
$X=\{x(e_1),\ldots,x(e_k)\}$ be the set of vertices introduced by the splitting operation.
Then, for any vertex $v \in S$, the minimum $(v,X)$-cut in $G'[S \cup X]$ has size at least $k$.
\end{lemma}
\begin{proof}
Suppose, for contradiction, that the claim is not true. Let $C'$ be a minimum $(v,X)$-cut in $G'[S\cup X]$ with size $|C'|=k'<k$.
Let $R$ be the set of vertices that $v$ reaches in $G'[S\cup X] \setminus C$. Then, $R$ is a \outcomp{k'} in $G'$ containing $v$ and, which contradicts \Cref{lemma:splitting} and the fact that $G$ is $k$-edge-connected.
\end{proof}

\ignore{
Now let $S' = S \cup X$, where $X = \{x(e_i) : e_i \in E(S, V \setminus S)\}$ is the set of vertices introduced by the splitting operation. Then $S'$ is a $k$-out set of $G' = \splitop{G}{E(S, V \setminus S)}$. It is also easy to verify that contracting $S'$ in $G'$ yields the same graph as contracting $S$ in $G$.
}

\begin{observation}
\label{obs:equivalence}
Let $S$ be a \outcomp{k} of a $k$-edge-connected digraph $G$. Also, let $G' = \splitop{G}{E(S, V \setminus S)}$ and let $S' = S \cup X$, where $X = \{x(e_i) : e_i \in E(S, V \setminus S)\}$ is the set of vertices introduced by the splitting operation. Then, $S'$ is a $k$-out set of $G'$ and $\contract{S}=\contractp{S'}$.
\end{observation}

\ignore{
\begin{lemma}
\label{lemma:koutboundary}
Let $S$ be a \outcomp{k} of a $k$-edge-connected digraph $G$, with boundary vertices $x_1,\ldots,x_k$. Then, for any vertex $v \in S \setminus \{x_1,\ldots,x_k\}$, the minimum $(v,\{x_1,\ldots,x_{\eta}\})$-cut in $G[S]$ has size at least $k$.
\end{lemma}
\begin{proof}
Suppose, for contradiction, that the claim is not true. Let $C$ be a minimum $(v,\{x_1,\ldots,x_{\eta}\})$-cut in $G[S]$ with size $|C|=k'<k$.
Let $R$ be the set of vertices that $v$ reaches in $G[S] \setminus C$. By \Cref{lemma:kcut-induces-kin}, $R$ is a \outcomp{k'} in $G$, which contradicts the fact that $G$ is $k$-edge-connected.
\end{proof}
}

Now we are ready to prove that the contraction of a \outcomp{k} $S$ in $G$ maintains the pairwise connectivity of the remaining vertices in $\bar{S}=V \setminus S$.

\begin{lemma}
\label{lemma:contraction}
Let $S$ be a \outcomp{k} of a $k$-edge-connected digraph $G$. Then, for any $\ell \ge k$, any two vertices $u,v \not\in S$ are $\ell$-edge-connected in $G$ if and only if $u$ and $v$ are $\ell$-edge-connected in \contract{S}.
\end{lemma}
\begin{proof}
For any graph $H$, let $\lambda_H(u,v)$ denote the size of the minimum $u$-$v$ cut in $H$. We show that $\lambda_G(u,v) = \lambda_{\contract{S}}(u,v)$.
Define $X, S'$ and $G'$ as in \Cref{obs:equivalence}.
By \Cref{lemma:splitting}, $\lambda_{G}(u,v)=\lambda_{G'}(u,v)$ and by \Cref{obs:equivalence}, $\lambda_{\contract{S}}(u,v) = \lambda_{\contractp{S'}}(u,v)$. Hence,
it suffices to argue that 
$\lambda_{G'}(u,v) = \lambda_{\contractp{S'}}(u,v)$.


Let $P'=\{p_1, p_2, \ldots, p_{\zeta}\}$ be a maximum set of edge-disjoint $u$-$v$ paths in $G'$. By Menger's theorem~\cite{menger}, $\zeta = \lambda_{G'}(u,v)$.
Clearly, each path $p_i \in P$ has a corresponding $u$-$v$ path $q_i$ in \contractp{S'}. Moreover, it is easy to observe that  $Q=\{q_1, q_2, \ldots, q_{\zeta}\}$ is a set of edge-disjoint $u$-$v$ paths in \contractp{S'}. Hence, $\lambda_{\contractp{S'}}(u,v) \ge \lambda_{G'}(u,v)$.

Now, we argue that $\lambda_{\contractp{S'}}(u,v) \le \lambda_{G'}(u,v)$. Let $Q=\{q_1, q_2, \ldots, q_{\xi}\}$ be a maximum set of edge-disjoint $u$-$v$ paths in \contractp{S'}, so $\xi = \lambda_{\contractp{S'}}(u,v)$. Assume, without loss of generality, that $Q'=\{q_1, q_2, \ldots, q_{\eta}\}$ is the subset of $Q$ consisting of the paths that contain the vertex $v_{S'}$ into which $S'$ was contracted (\Cref{def:contraction}).
Then, any path $p_i \in Q'$ has the edges $(u_i,v_{S'})$ and $(v_{S'},y_i)$, where $y_i$ is an exit point of $S'$.
Hence, $\eta \le k$.
Let $(u_i,v_i)$ and $(x(e_i),y_i)$ be the edges of $G'$ that correspond to $(u_i,v_{S'})$ and $(v_{S'},y_i)$, respectively.
We claim that the minimum $(v_i, \{x(e_i),\ldots,x(e_\eta)\})$-cut in $G'[S']$ has size at least $k$. Suppose, for contradiction, that the claim is not true. Let $C$ be a $(v_i, \{x(e_i),\ldots,x(e_\eta)\})$-cut in $G'[S']$ with size $|C|=k'<k$.
Let $R$ be the set of vertices that $v_i$ reaches in $G'[S'] \setminus C$. Then, $R$ is a \outcomp{k'} in $G'$, which, by \Cref{lemma:splitting}, contradicts the fact that $G$ is $k$-edge-connected.
This implies $\lambda_{G'}(\{v_1, \ldots, v_{\eta}\}, \{y_1, \ldots, y_{\eta}\}) \ge k$. 
By Menger's theorem, $G'$ therefore contains at least $\eta$ edge-disjoint $u$–$v$ paths passing through $S'$, corresponding to the set of paths $Q'$ in \contractp{S'}. 
Since the paths in $Q \setminus Q'$ remain unchanged in $G'$, we conclude that $\lambda_{\contractp{S'}}(u, v) \le \lambda_{G'}(u, v)$.
\end{proof}

\subsubsection{First phase}

Let $G$ be the given $k$-edge-connected digraph. We choose an arbitrary vertex $s$ as a start vertex and compute the sets $[v]$, for all ordinary vertices $v$, corresponding to the minimal \outcomp{k}s $\moutM{v}$, in proper order (\Cref{definition:proper} and \Cref{proposition:proper}). 

Now, given a proper order of the sets $[v]$, our goal is to compute, for each set $[v]$, an auxiliary graph $G_{[v]}$, 
satisfying the following: 
\begin{definition}[First-phase auxiliary graph]
\label{definition:auxv-k}
An auxiliary graph $G_{[v]}$ is $k$-edge-connected and contains two types of vertices, ordinary and auxiliary, such that:
\begin{enumerate}
    \item The ordinary vertices are precisely the vertices of $[v]$.
    \label{definition:auxv:property1-k}
    \item An auxiliary vertex $x$ is not $(k+1)$-edge-connected with any other vertex $y$, ordinary or auxiliary, in $G_{[v]}$.
    \label{definition:auxv:property2-k}
    \item Any two ordinary vertices $u$ and $v$ of $G_{[v]}$ are $\ell$-edge-connected in $G_{[v]}$, for $\ell \ge k$, if and only if they are $\ell$-edge-connected in $G$.
    \label{definition:auxv:property3-k}
\end{enumerate}
\end{definition}

In the first phase of our algorithm, we construct auxiliary graphs $G_{[v]}$ for all sets $[v]$, 
with a total number of vertices $O \left ( \sum_{[v]}  |V(G_{[v]})| \right) = O(|V(G)|)$ and a total number of edges $O \left ( \sum_{[v]}  |E(G_{[v]})| \right) = O(|E(G)|)$.

%
We construct the graphs $G_{[v]}$ by applying the contraction operation of \Cref{def:contraction}.
This ensures that the connectivity of ordinary vertices is preserved (\Cref{lemma:contraction}) and that the resulting graph is a $k$-edge-connected.

We maintain an evolving graph $G'=(V',E')$, where we apply the contraction of each \outcomp{k} that we identify. Initially, $G'=G$. We
process the sets $[v]$ in proper order, which means that if $\moutM{u} \subset \moutM{v}$ then $[u]$ is computed before $[v]$.

When we process the next set $[v]$, we wish to identify the minimal \outcomp{k} \mpout{v} in the current graph $G'$ that corresponds to \mout{v} and contains all ordinary vertices of $[v]$.
To process the next set $[v]$, we execute the 
randomized local search \Cref{algorithm:randomizedlocalsearch} (presented in \Cref{sec:proposition:local_search}) for exponentially increasing $\Delta = 2^i$, $i \in \{0,\ldots,\lceil \log_2{m} \rceil\}$.\footnote{Even though it is not necessary, one could make the following modification in lines 2 and 9 of \Cref{algorithm:DFS}: one can return the DFS-tree path as soon as we reach any ordinary vertex $w \not \in [v]$. This is because any such $w$ has not been contracted into an auxiliary vertex.}
For constant $k$, we can use the 
local search \Cref{algorithm:localsearch} instead, in order to obtain a deterministic decomposition.
Henceforth, we will assume that the local searches have successfully identified each set \mpout{v}. (Note that this is easy to verify, since for any $[v] \not= [s]$ we have $\mpoutM{v} \not= \emptyset$.)


When we initiate a local search from an ordinary vertex $v$ in the evolving graph $G'$, we find a minimal \outcomp{k} $S=\mpoutM{v}$. At this point, we need to (i) create the auxiliary graph $G_{[v]}$, by performing a contraction of $V(G') \setminus S$, and (ii) update the evolving graph $G'$, by performing a contraction of $S$.

To enable the fast execution of tasks (i) and (ii), we maintain, for each vertex in $V(G')$, a list of incoming edges and a list of outgoing edges, organized in circular doubly-linked lists. This allows us to delete an edge from a list and merge two lists in constant time. Also, we assume that for each instance of an edge $e=(u,v)$, we maintain a pointer to the location of $e$ in the adjacency lists of $u$ and $v$.

When we apply a contraction for a \outcomp{k} (resp., \incomp{k}) $S$, we introduce a new auxiliary vertex $v_S$ with $\mathit{out}(v_S)=k$ (resp., $\mathit{in}(v_S)=k$), that represents the vertices (ordinary and auxiliary) in $S$. We refer to such a vertex $v_S$ as a \emph{$k$-out auxiliary vertex} (resp., \emph{$k$-in auxiliary vertex}).

\paragraph{Construction of $G_{[v]}$.}

Let $S=\mpoutM{v}$ be the next minimal \outcomp{k} found in $G'$, and let $\bar{S}=V(G') \setminus S$. Also, let $C=E(S,\bar{S})=\{e_1=(x_1,y_1), \ldots, e_k=(x_k,y_k)\}$ be the edges of the \cut{k} that defines $S$. (Note that $x_1,\ldots,x_k \in S$ and $y_1,\ldots,y_k \in \bar{S}$).
We apply the contraction operation for $\bar{S}$, which replaces $G'[\bar{S}]$ (the subgraph of $G'$ induced by $\bar{S}$) with a $k$-in auxiliary vertex $v_{\bar{S}}$. Then, we let $G'=\contractp{\bar{S}}$.

Given $S$, it is straightforward to construct $G'$ in $O(\mathit{vol}(S) + \mathit{in}(S))$ time, which is not fast enough for our purposes. Fortunately, we can apply the following reduction rule to obtain a construction with 
$O(\mathit{vol}(S))$ time and space.

\vspace{0.5cm}
\noindent\emph{Reduction Rule}: For each vertex $u \in S$ do the following.
\begin{itemize}
\item Let $\rho=E(\bar{S},u)$, i.e., the number of edges entering $u$ from $\bar{S}$. Then, in $G'$ we keep $\min\{k,\rho\}$ copies of the edge $(v_{\bar{S}},u)$. 
\end{itemize}

The correctness of this reduction rule follows directly from the observation that, for any two vertices $u, w \in S$, there exist at most $k$ edge-disjoint $u$–$w$ paths that traverse $\bar{S}$.

To complete the construction of $G_{[v]}$, we let $G_{[v]} = \contractp{\bar{S}}$, for $S=\mpoutM{v}$, where we execute the reduction rule.

\begin{lemma}
\label{lemma:gadget-replace-time}
Given a \outcomp{k} $S=\mpoutM{v}$ of the evolving graph $G'$, we can construct $G_{[v]}$ in $O(\mathit{vol}(S))$ time, plus $O(m)$ time for all $G_{[v]}$ graphs constructed through the algorithm.
\end{lemma}
\begin{proof}
To construct $G_{[v]}$, we first create a copy of each vertex in $v \in S$, and also copy the list of outgoing edges of $v$ in $G'$. This takes $O(\mathit{vol}(S))$ time. Then, we introduce the $k$-in auxiliary vertex $v_{\bar{S}}$ of $\bar{S}$, and replace each entering edge $(x_i,y_i) \in E(S,\bar{S})$ with $(x_i,v_{\bar{S}})$.

Next, we need to add the necessary outgoing edges $(v_{\bar{S}},u)$, for  $u\in S$, according to the Reduction Rule. To that end, for each vertex $v \in S$, we scan its list of entering edges in $G'$, and we count the number $\rho$ of edges $(z,u)$ such that $z \in \bar{S}$. We stop as soon as we find $k$ such edges, or reach the end of the list of entering edges of $v$.
Hence, we add $\min\{k,\rho\}$ copies of the edge $(v_{\bar{S}},u)$ to the list of outgoing edges of $v_{\bar{S}}$. 
Since we access at most $k$ edges entering $u$ from $\bar{S}$, the overall time to scan all lists is bounded by $O(\mathit{vol}(S)+k|S|)=O(\mathit{vol}(S))$, as claimed.
This bound also includes the time to construct the list of entering edges for each vertex of $G_{[v]}$. 
\end{proof}

\paragraph{Updating $G'$.}

After we have constructed $G_{[v]}$, we need to update the evolving graph by setting $G' = \contractp{S}$. Recall that $S=\mpoutM{v}$ is a \outcomp{k} of the evolving graph $G'$. 
We apply the contraction operation for $S$, which replaces $G'[S]$ (the subgraph of $G'$ induced by $S$) with a $k$-out auxiliary vertex $v_S$. This replacement procedure is analogous to the one applied for the construction of $G_{[v]}$, but slightly different since we operate directly on $G'$.
The important difference with respect to the contraction procedure of $\bar{S}$ for $G_{[v]}$, is that here we do not construct a new graph from scratch, but modify $G'$. Also, we do not need (and cannot afford) to apply the Reduction Rule.

To perform the contraction of $S$ into $v_S$, we need to update the head of the edges $(w,v)$, where $w \in \bar{S}$. 
We do that indirectly with a disjoint set-union (DSU) data structure~\cite{setunion:tvl,dsu:tarjan}, which we use to contract all the vertices of $S$.
The DSU data structure supports the operation $\mathit{unite}(u,v)$, which unites the sets of two given vertices $u$ and $v$, and makes $u$ the representative of the new set. It also supports the query $\mathit{find}(v)$, which returns the representative vertex of the set containing $v$.
To perform the contraction of $S$, we chose a vertex $r \in S$ as representative, and execute $\mathit{unite}(r,v)$ for all vertices $v \in S \setminus r$. Hence, whenever we need to traverse an edge $(w,v)$, we traverse $(w,\mathit{find}(v))$ instead.

Next, we update the list of outgoing and entering edges of $v$. To do this fast, we assume that we have stored a list $L(S)$ of the edges with both endpoints in $S$, which we can do in $O(\mathit{vol}(S))$ time.
Next, we delete each edge $(u,v) \in L(S)$, both from the list of outgoing edges of $u$ and from the list of entering edges of $v$. 
Then, for every vertex $v \in S \setminus r$, we merge the list of entering edges of $v$ with the list of entering edges of $r$.
Finally, we update the list of outgoing edges of $r$, so that it consists of the $k$ outgoing edges of $S$. 

The proof that $G_{[v]}$ satisfies \Cref{definition:auxv-k} is deferred to a later subsection. We proceed to analyze the running time of the procedure.

\begin{lemma}
\label{lemma:contract-time-2}
Given a \outcomp{k} $S=\mpoutM{v}$ of the evolving graph $G'$, we can update $G'$ in $O(\mathit{vol}(S))$ time, plus the time required for the DSU data structure.
\end{lemma}
\begin{proof}
Our procedure receives as input a list of the vertices in $S$.
It computes a list $L(S)$ of the edges with both endpoints in $S$ in $O(\mathit{vol}(S))$ time.
Then, we can remove all edges in $L(S)$ from the adjacency lists of $G'$ in $O(|L(S)|)$ time, and update the list of outgoing edges of $r$, so that it consists of the $k$ outgoing edges of $S$.
Finally, we can merge the list of entering edges of each vertex $v \in S \setminus r$ with that of $r$, in a total of $O(|S|)$ time.
Thus, the total update time is $O(|L(S)|+|S|)=O(\mathit{vol}(S))$, plus the time required by the DSU data structure.
\end{proof}

After processing all sets $[v]$, such that $s \not\in [v]$, we set $G_{[s]}=G'$. The ordinary vertices of $G_{[s]}$ are the ordinary vertices $v \in V(G)$ for which $\moutM{v} = \mpoutM{v} = \bot$.

\subsubsection{Second phase}

The second phase completes the construction of the \ecc{(k+1)} decomposition $H_1, \ldots, H_t$, where each $H_i$ corresponds to a unique $(k+1)$-edge-connected component $C_i$ of $G$. 
To that end, we process each auxiliary graph $G_{[v]}$ of the first phase as follows. Recall that if $s \not\in [v]$, $G_{[v]}$ is the graph that results from the contraction of the \incomp{k} $S=V(G') \setminus \mpoutM{v}$. Since $s \not\in \mpoutM{v}$, the $k$-out auxiliary vertex $v_{S}$ represents the contraction of $S=V(G') \setminus \mpoutM{v}$, which contains the start vertex $s$. Hence, we set $v_S$ as the start vertex of $G_{[v]}$, and apply the same procedure as the first phase, for $G'=(G_{[v]})^R$. For $G_{[s]}$, we apply the procedure of the first phase for $G'=(G_{[s]})^R$ with start vertex $s$.

\subsubsection{Correctness proof}

First, we will prove the correctness of our construction, under the assumption that the local searches have successfully identified each set \mpout{v}. Then, we will bound the probability of failure.

In the following, we prove that in each of the graphs created, the ordinary vertices are all \conn{(k+1)}.
The rest of the properties of \Cref{theorem:keccdecomp} follow directly from our algorithm and the fact that contractions preserve the connectivity of ordinary vertices (\Cref{lemma:contraction}).

Let $C_1,\ldots,C_q$ be the $(k+1)$-edge-connected components of a \conn{k} $G$.

\begin{lemma} \label{lem:correctnessPhaseOne-k}
The graphs $G_{[v]}$ constructed during the first phase of the algorithm satisfy \Cref{definition:auxv-k}.
\end{lemma}
\begin{proof}
We only need to prove Property~\ref{definition:auxv:property1-k} of \Cref{definition:auxv-k}, as the remaining properties follow immediately from the correctness of the contraction operation (\Cref{lemma:contraction}).

To prove Property~\ref{definition:auxv:property1-k}, consider the local search that identified \mpout{v} in $G'$.
All vertices of the induced subgraph of the minimum $k$-out set of any ordinary vertex $v$ are reachable from $v$ (otherwise we could remove the non-reachable vertices and get a smaller set).
Then, when we process $[v]$, for all vertices in $[v]$, the evolving graph contains a $k$-out set $S'$ whose ordinary vertices are exactly the remaining ordinary vertices of $\mu(v)$ in $G'$.
By induction, if there exists any $u$ with $\mu(u) \subset \mu(v)$, the ordinary vertices of $\mu(u)$ have already been processed and thus are no longer ordinary vertices in the evolving graph (this is why we need the proper ordering on the sets $[v]$).
Furthermore, the ordinary vertices $u$ of $\mu(v)$ have $\mu(u) \subseteq \mu(v)$ (otherwise by \Cref{lemma:submodularity3} the set $\mu(u) \cap \mu(v)$ is $k$-out and separates $u$ from $s$, contradicting the minimality of $\mu(u)$).
We conclude that the ordinary vertices of $S'$ are exactly the vertices of $[v]$.

As our local search avoids $s$ by construction, and $s\not \in [v]$, it follows that $s$ is not in the aforementioned set.
Therefore, the existence of $S'$ shows that the ordinary vertices of the minimum $k$-out set we find contains a subset of the remaining ordinary vertices of $[v]$ and no other ordinary vertices.

We now argue that the minimum $k$-out set $\mu'(v)$ we find contains all vertices of $[v]$.
If this was not true, then it would contain some $v_1 \in [v]$ but not some $v_2 \in [v]$, and also not $s$ (by construction of the local search algorithm).
But the contraction operation on the evolving graph guarantees that such a set would exist in the original graph (as we prove in the following claim), which contradicts the definition of $[v]$. Therefore, it suffices to prove:

\begin{claim}\label{clm:persisting2OutIn}
Let $u,v$ be two ordinary vertices in the evolving graph $G'$, and suppose there exists a $k$-out (resp. $k$-in) set $S$ containing $u$ but not $v,s$.
Then, there exists a $k$-out (resp. $k$-in) set in the original graph, containing $u$ but not $v,s$.
\end{claim}
\begin{proof}
Let $G''$ be the evolving graph at some point during the first phase when we found a $k$-out set $S''$ and contracted it into $v_{S''}$, so that $G' = G''_{\langle S'' \rangle}$.
It suffices to prove that in $G''$ there exists a $k$-out or $k$-in set $S$ containing $u$ but not $v,s$, as we can repeat the argument to reach the original graph.
Note that $S''$ does not contain $u,v,s$, as these vertices are ordinary in $G'$.

The contraction operation in the evolving graph does not remove any edge with at most one endpoint in $S''$. 
Therefore, if $S$ does not contain $v_{S''}$, then $S$ exists as is in $G''$, with the same number of incoming and outgoing edges.

On the other hand, if $S$ contains $v_{S''}$, then let $Y = (S \setminus \set{v_{S''}}) \cup S''$. 
If $S$ is a \outcomp{k}, then $Y$ is also a \outcomp{k} that contains $u$ (because $u\in S$) but not $v,s$. Similarly, if $S$ is a \incomp{k}, then $Y$ is also a \incomp{k} that contains $u$ but not $v,s$. In either case the result follows.
\end{proof}
This concludes the proof for all sets except for $[s]$.
However, as the sets $[v]$ form a partition, and after the first phase is done we put all remaining ordinary vertices of the evolving graph in $G_{[s]}$, the result follows.
\end{proof}

Now we focus on the second phase.

\begin{lemma}[Refined Cuts]
\label{lemma:refined-cuts-k}
Let $S$ be a \outcomp{k} of a $k$-edge-connected digraph $G$ that contains two distinct vertices $x$ and $y$. If there is no $(\le k)$-out set $S' \subset S$ that separates $x$ and $y$, then any \outcomp{k} that separates $x$ and $y$ must contain all vertices in $V \setminus S$.
\end{lemma}
\begin{proof}
Let $X$ be a \outcomp{k} that contains $x$ but not $y$. 
Then, $out(S \cap X) > k$, since otherwise $S \cap X$ would be a $(\le k)$-out set contained in $S$ that separates $x$ and $y$.
By submodularity $\mathit{out}(S)+\mathit{out}(X) \ge \mathit{out}(S \cup X) + \mathit{out}(S \cap X)$, and since $\mathit{out}(S) = \mathit{out}(X) = k$, we have $\mathit{out}(S \cup X) < k$. Thus, since $G$ is $k$-edge-connected, $S \cup X$ must contain all vertices of $G$, and so $V \setminus S \subseteq X$.
\end{proof}

\begin{corollary}
\label{cor:refined-cuts-k}
Let $x$ and $y$ two vertices in $[v]$ that are not $(k+1)$-edge-connected in $G$. Then, there is a \outcomp{k} $X$ in $G_s^R$ that separates $x$ and $y$ that does not contain any vertex in $V \setminus \moutM{v}$. In particular, $X$ does not contain $s$.
\end{corollary}
\begin{proof}
We apply \Cref{lemma:refined-cuts-k} for $S=\moutM{v}$. By the definition of $[v]$, and the fact that $x,y \in [v]$, we have that there is no \outcomp{k} $S' \subset \moutM{v}$ that separates $x$ and $y$. Since $x$ and $y$ are not $(k+1)$-edge-connected, we have $x \not\in \moutRM{y}$ or $y \not\in \moutRM{x}$ (\Cref{lemma:distinctM}).
Assume, without loss of generality, that $x \not\in \moutRM{y}$. Then, $X = V \setminus \moutRM{y}$ is a \incomp{k} in $G^R$ that contains $x$ but not $y$.
Hence, $X$ is a \outcomp{k} in $G$ that separates $x$ and $y$, so by \Cref{lemma:refined-cuts-k}, it contains all vertices of $V \setminus \moutM{v}$. Then, \moutR{y} is a \outcomp{k} in $G_s^R$ that separates $x$ and $y$ that does not contain any vertex in $V \setminus \moutM{v}$.
\end{proof}

We are now ready to show that the ordinary vertices in the graphs constructed during the second phase of the algorithm are \conn{(k+1)}.

\begin{lemma}
Let $u,v$ be ordinary vertices in a graph $H_i$ constructed during the second phase of the algorithm.
Then $u,v$ are \conn{(k+1)}.
\end{lemma}
\begin{proof}
If $u,v$ are not \conn{(k+1)} and are not splitted by the first phase of the algorithm, then $[u] = [v]$ by \Cref{lem:correctnessPhaseOne-k}.
\Cref{cor:refined-cuts-k} implies that there exists a $k$-out set $X$ in $G_s^{R}$ that w.l.o.g. contains $u$ and does not contain $v,s$.
Hence, $G_{[v]}$ contains a $k$-out set $S'$ containing $u$ but not $v,s$.
As the second phase is the same as the first phase in the graph $G_{[v]}$ instead of $G$, by \Cref{lem:correctnessPhaseOne-k} we conclude that, in the end, $u$ is in a different graph $H_i$ from $v$.
\ignore{Therefore, the minimum $k$-out set of $u$ in $G_s^{R}$ that does not include $s$ also does not include $v$.
Notice that $u$ reaches everything in this minimum $k$-out set, as otherwise we could remove the unreachable vertices and create an even smaller $k$-out set separating $u$ from $s$.
Therefore we can apply Lemma~6.5 of~\cite{GKPP:3ECC} which implies that $G_{[v]}$ contains a $k$-out set $S'$ containing $u$ but not $v,s$; in turn this means that the minimum $k$-out set of $u$ in $G_{[v]}$ does not contain $v,s$.
As the second phase is the same with the first phase in the graph $G_{[v]}$ instead of $G$, by \Cref{lem:correctnessPhaseOne} we get that in the end $u$ is in a different graph $H_i$ from $v$.}
\end{proof}

We complete the proof of correctness by bounding the probability of error. 

\begin{proposition}\label{lem:correctnessPhaseTwo}
The graphs constructed during the second phase of the algorithm satisfy \Cref{theorem:keccdecomp} with probability at least $1-\delta$, where $\delta$ is a parameter with $0 < \delta < 1$.\lnote{Added this proposition.}
\end{proposition}
\begin{proof}
From the analysis above, the graphs constructed during the second phase of the algorithm satisfy \Cref{theorem:keccdecomp}, provided that each local search correctly identifies the corresponding \mpout{v} set. 
Consider the next set $S = \mpoutM{v}$ that the algorithm attempts to find in the evolving graph $G'$ during the first phase. 
Let $2^{i-1} < \mathit{vol}(S) \le 2^{i}$. 
By \Cref{proposition:randomized_local_search} and the use of exponentially increasing $\Delta$, the execution of \Cref{algorithm:randomizedlocalsearch} with $\Delta = 2^{i}$ fails to identify \mpout{v} (i.e., incorrectly returns $\emptyset$) with probability at most $1/2$.

Note that in this setting, it is straightforward to verify whether a failure has occurred. 
Indeed, for any $[v] \neq [s]$, we have $\mpoutM{v} \neq \emptyset$. 
Therefore, if the local search for $S = \mpoutM{v}$ either returns $S = \emptyset$ 
or succeeds only when the volume parameter satisfies $\Delta \ge 2\,\mathit{vol}(S)$, 
we can immediately terminate the execution and report a failure. 
(This property will also be used to bound the overall running time of the construction.)

To amplify the success probability, we repeat \Cref{algorithm:randomizedlocalsearch} 
$N = \lceil \log_2(2n/\delta) \rceil$ times. 
Then, the probability that all $N$ repetitions fail is at most 
$(1/2)^N \le \frac{\delta}{2n}$.
Since the algorithm identifies fewer than $2n$ distinct \mpout{v} sets in total 
(less than $n$ in each of the two phases), 
a union bound implies that the probability that any local search fails is at most~$\delta$. 
Hence, with probability at least $1 - \delta$, all searches succeed, completing the proof.
\end{proof}

\ignore{
\begin{corollary} \label{lem:correctnessPhaseTwo}
The graphs constructed during the second phase of the algorithm satisfy \Cref{theorem:keccdecomp}.
\end{corollary}
}

\subsubsection{Running time}

Now we bound the running time of our construction and the total size of the constructed auxiliary graphs. We only analyze the construction of the first phase, since the second phase executes the same construction on the reverse graph of each first-phase auxiliary graph.

\begin{lemma}
\label{lemma:first-phase-evolving-k}
The total time to identify all $\mpoutM{v}$ sets and maintain the evolving graph $G'$ during the first phase of the algorithm is $O(k^2m \log_2{(n/\delta)})$, plus the time to perform $O(m)$ operations on the DSU data structure. 
\end{lemma}
\begin{proof}
Directly by our algorithm and the 
running time of our local search, we have that the time to find an $\mpoutM{v}$ is $O(k^2 \log_2{(n/\delta)}\sum_{i=0}^{\lceil \log_2{vol(\mpoutM{v})}\rceil} 2^i) = O(k^2 \log_2{(n/\delta)}vol(\mpoutM{v}))$.
The contraction operation removes all vertices in a \outcomp{k}, and introduces a new vertex with $k$ outgoing edges; as we perform at most $n$ contractions, the extra edges added this way are at most $kn=O(m)$.

We conclude that the running time to identify all $\mpoutM{v}$ sets is at most 
$k^2 \log_2{(n/\delta)}$
times the sum of out-degrees of all vertices of $G$, plus $O(m)$, which gives the desired 
$O(k^2m \log_2{(n/\delta)})$ bound.
The lemma then directly follows from \Cref{lemma:contract-time-2}.
\end{proof}

\begin{lemma}
\label{lemma:first-phase-auxiliary-k}
The total size of all constructed graphs during the first phase of the algorithm is $O(m)$, and the total number of vertices is $O(n)$.
The total time to construct all auxiliary graphs $G_{[v]} = \contractp{V' \setminus \mpoutM{v}}$ during the first phase of the algorithm is $\Ot(k^2m)$. 
\end{lemma}
\begin{proof}
By \Cref{lemma:contract-time-2}, when we identify the next minimal \outcomp{k} \mpout{v} in the evolving graph $G'$, we create $G_{[v]} = \contractp{V' \setminus \mpoutM{v}}$ in $O(\mathit{vol}(\mpoutM{v}))$ time (plus $O(m)$ time in total, which is negligible).
So, the total size of all auxiliary graphs $G_{[v]}$ is bounded by the time to construct them, as they are all explicitly constructed.
This is $O(\sum_{[v]}{\mathit{vol}(\mpoutM{v})}) = O(m)$, exactly as argued in \Cref{lemma:first-phase-evolving-k}.
To bound the total number of vertices, we observe that whenever we identify a minimal \outcomp{k} $S$ in the evolving graph $G'$, we introduce two auxiliary vertices: a $k$-out vertex $v_S$ in $G'$ and a $k$-in vertex $v_{\bar{S}}$ in $G_{[v]}$. Since this process can occur at most $n$ times (once per $[v]$ set), we introduce a maximum of $2n$ auxiliary vertices. Therefore, the total number of vertices in all generated graphs remains $O(n)$.\lnote{Added this.}

The running time is dominated by \Cref{lemma:first-phase-evolving-k}.
It is $\Ot(k^2 m)$, using, e.g., the DSU data structure of~\cite{dsu:tarjan} or of~\cite{setunion:tvl}.
\end{proof}

\ignore{
If for the DSU data structure we use a structure that supports each $\mathit{find}$ operation in worst-case $O(1)$ time and
any sequence of $\mathit{unite}$ operations in total time $O(n \log n)$~\cite{setunion:tvl}, then we obtain an algorithm that constructs the \ecc{(k+1)} decomposition in $O(k^2 m \log{n})$ total time.
Alternatively, we implement $\mathit{unite}$ and $\mathit{find}$ using compressed trees with appropriate heuristics \cite{dsu:tarjan}, and the total time for the DSU operations is $O(m \alpha(m,n))$. 
}

%% file: 3ecc-decomposition.tex
\subsection{3-edge-connected component decomposition of a general digraph}
\label{sec:3ecc-decomposition}


\ignore{
\begin{theorem}[3-Edge-Connected Components (\ecc{3}) Decomposition]
\label{theorem:3eccdecomp}
Let $G$ be a digraph. In linear time, we can construct a collection $H_1,\dots,H_t$ of graphs such that:
\begin{itemize}
\item{The vertices $V(H_i)$ of each graph $H_i$ are partitioned into two sets of vertices, ordinary and auxiliary.}
\item{For each $i \in \{1,\ldots,t\}$, the ordinary vertices of $V(H_i)$ are $3$-edge-connected.}
\item{For every vertex of $G$, there is exactly one graph among $H_1,\dots,H_t$ that contains it as an ordinary vertex.}
\item{Every two vertices $u$ and $v$ of $G$ are $4$-edge-connected if and only if there is an $i\in\{1,\dots,t\}$ such that $u$ and $v$ are $4$-edge-connected ordinary vertices of $H_i$.}
\end{itemize}
\end{theorem}
}

In this section we describe the 
decomposition of \Cref{theorem:3eccdecomp}. 
Here we take advantage of the $2$-connectivity-light graph (\lightgraph{2}) decomposition, introduced in \cite{GKPP:3ECC} and also exploited in \cite{Linear3ECC}.


\begin{definition}[\cite{GKPP:3ECC}]
\label{definition:2clg}
A strongly connected digraph $G$ is called a \emph{$2$-connectivity-light graph} (\lightgraph{2}) if $V(G)$ can be partitioned into two sets of vertices, ordinary and auxiliary, such that:
\begin{enumerate}
\item{All ordinary vertices are $2$-edge-connected.}
\label{definition:2clg:property1}
\item{Every auxiliary vertex has out-degree or in-degree $1$.}
\label{definition:2clg:property2}
\item{For every two vertices $u$ and $v$ with $\mathit{out}(u)>1$ and $\mathit{in}(v)>1$, there are two edge-disjoint paths from $u$ to $v$.}
\label{definition:2clg:property3}
\item{If $u$ is an auxiliary vertex with $\mathit{out}(u)=1$ (resp., $\mathit{in}(u)=1$), then, for every ordinary vertex $v$, the unique outgoing edge from $u$ (resp., the unique incoming edge to $u$) is the only bridge that appears in every path from $u$ to $v$ (resp., from $v$ to $u$).}
\label{definition:2clg:property4}
\end{enumerate}
\end{definition}

Our \ecc{3} decomposition utilizes the following result.

\begin{theorem}[\lightgraph{2} Decomposition~\cite{GKPP:3ECC}, Lemma~3.2]
\label{theorem:2clgdecomp}
Let $G$ be a strongly connected digraph. In linear time, we can explicitly construct a collection $H_1,\dots,H_t$ of 2CLG graphs such that:
\begin{itemize}
\item{For every vertex of $G$, there is exactly one graph among $H_1,\dots,H_t$ that contains it as an ordinary vertex.}
\item{Every two vertices $u$ and $v$ of $G$ are $k$-edge-connected, for $k\ge 2$, if and only if there is an $i\in\{1,\dots,t\}$ such that $u$ and $v$ are $k$-edge-connected ordinary vertices of $H_i$.}
\end{itemize}
\end{theorem}

\ignore{
We point out that a \lightgraph{2} may have parallel edges between vertices.
For each edge, we also store its multiplicity; however, it may be, due to the way we construct the \lightgraph{2}s, that we have multiple copies of an edge, each with a different multiplicity.
The actual multiplicity of this edge in the \lightgraph{2} is the sum of all the multiplicities of the copies.

To make the presentation easier, we assume that a \lightgraph{2} $G$ has the following property.
\begin{property}
\label{property:additional2CLG}
Let $G$ be a \lightgraph{2}. Then, for any two vertices $u,v \in V(G)$, where $in(u)=1$, $G$ has at most one copy of the edge $(u,v)$.
\end{property}
Once again, we point out that this one copy also stores the multiplicity of the corresponding edge.
We can easily satisfy Property~\ref{property:additional2CLG} of \Cref{definition:2clg} in a linear-time preprocessing phase, e.g., by sorting the edges of $G$ lexicographically. 
}

\ignore{
\begin{definition}
\label{definition:3clg}
A strongly connected digraph $G$ is called \emph{$3$-connectivity-light graph} (\lightgraph{3}) if $V(G)$ can be partitioned into two sets of vertices, ordinary and auxiliary, such that:
\begin{enumerate}
\item{All ordinary vertices are $3$-edge-connected.}
\label{definition:3clg:property1}
\item{Every auxiliary vertex has out-degree or in-degree at most $2$.}
\label{definition:3clg:property2}
\item{For every two vertices $u$ and $v$, there are $\min\{out(u),in(v),3\}$
edge-disjoint paths from $u$ to $v$.}
\label{definition:3clg:property3}
\item{If $u$ is an auxiliary vertex with $\mathit{out}(u)=1$ (resp., $\mathit{in}(u)=1$), then, for every ordinary vertex $v$, the unique outgoing edge from $u$ (resp., the unique incoming edge to $u$) is the only bridge that appears in every path from $u$ to $v$ (resp., from $v$ to $u$).}
\label{definition:3clg:property4}
\item{Let $u$ is an auxiliary vertex with $\mathit{out}(u)=2$ with outgoing edges $(u,x_1)$ and $(u,x_2)$. Then, for every ordinary vertex $v$, the unique latest minimum $u$-$v$ cut is one of the following: (i) the two outgoing edges from $u$ if $out(x_1)>1$ and $out(x_2)>1$, (ii) $(u,x_1)$ and the unique outgoing edge of $x_2$ if $out(x_1)>1$ and $out(x_2)=1$, (iii) the unique outgoing edge of $x_1$ and $(u,x_2)$ if $out(x_1)=1$ and $out(x_2)>1$, or (iv) the unique outgoing edge of $x_1$ and the unique outgoing edge of $x_2$ if $out(x_1)=1$ and $out(x_2)=1$.}\sidenote{Loukas: Do we need these last two properties?}
\label{definition:3clg:property5}
\item{Let $u$ is an auxiliary vertex with $\mathit{in}(u)=2$, with incoming edges $(x_1,u)$ and $(x_2,u)$. Then, for every ordinary vertex $v$, the unique earliest minimum $v$-$u$ cut is one of the following: (i) the two incoming edges to $u$ if $in(x_1)>1$ and $in(x_2)>1$, (ii) $(x_1,u)$ and the unique incoming edge to $x_2$ if $in(x_1)>1$ and $in(x_2)=1$, (iii) the unique incoming edge to $x_1$ and $(x_2,u)$ if $in(x_1)=1$ and $in(x_2)>1$, or (iv) the unique incoming edge to $x_1$ and the unique incoming edge to $x_2$ if $in(x_1)=1$ and $in(x_2)=1$.}
\label{definition:3clg:property6}
\end{enumerate}
\end{definition}

Our goal in this section is to prove the following theorem. 

\begin{theorem}[3CLG Decomposition]
\label{theorem:3clgdecomp}
Let $G$ be a strongly connected digraph. In linear time, we can construct a collection $H_1,\dots,H_t$ of \lightgraph{3}s such that:
\begin{itemize}
\item{For every vertex of $G$, there is exactly one graph among $H_1,\dots,H_t$ that contains it as an ordinary vertex.}
\item{Every two vertices $u$ and $v$ of $G$ are $4$-edge-connected if and only if there is an $i\in\{1,\dots,t\}$ such that $u$ and $v$ are $4$-edge-connected ordinary vertices of $H_i$.}
\end{itemize}
\end{theorem}

Note that \cite{Linear3ECC} gave a linear-time decomposition of a strongly connected graph into a collection of \lightgraph{2}s.
}

Given \Cref{theorem:2clgdecomp}, we may assume that our input graph $G$ is a \lightgraph{2}, which we wish to decompose into a collection of graphs as stated in \Cref{theorem:3eccdecomp}.
The general idea is to exploit the algorithm of \cite{Linear3ECC} that computes a compact representation of the minimal \outcomp{2}s of $G$ in linear time.

\subsubsection{\ecc{3} decomposition of a \lightgraph{2}}

Let $G=(V,E)$ be a \lightgraph{2} with at least two ordinary vertices.
Let $G_s$ be the corresponding flow graph with start vertex $s$, where $s$ is an arbitrarily selected ordinary vertex of $G$.
Recall that a \outcomp{2} (resp., \incomp{2}) $S$ is a set of vertices such that $out(S) = 2$ (resp., $in(S)=2$). We will consider only \outcomp{2}s and \incomp{2}s that contain at least one ordinary vertex.

For any ordinary vertex $v$, we let \mout{v} denote the minimal \outcomp{2} that contains $v$ but not $s$.
It follows from~\Cref{lemma:submodularity3} that if any $2$-out set separating $v$ from $s$ exists, then there also exists a minimum one, contained in every other such set (notice that $v,s$ are \conn{2}, therefore $\lambda(u,s)\ge 2$, by the properties of 2CLGs).
If $v$ is not contained in any such set then we let $\moutM{v} = \bot$.

Let $[v]$ denote the set of ordinary vertices that have the same minimal \outcomp{2} in $G$, i.e., $[v] = \{ u : \moutM{v} = \moutM{u}\}$. 
Similarly, $[v]_R$ denotes the set of ordinary vertices that have the same minimal \outcomp{2} in $G^R$.
Note that the $[v]$-sets form a partition of the ordinary vertices of $G$. In particular, $v \in [s]$ implies $\moutM{v} = \bot$, i.e., there is no cut of size at most $2$ that separates all paths from $v$ to $s$.
In~\cite{Linear3ECC} it is proven (Proposition III.5) that the \conn{3} components are equivalent with the partitions induced by the sets $[v]\cap [v]_R$.
Then, the linear-time algorithm of \cite{Linear3ECC} computes the sets $[v]$ in 
proper order (\Cref{definition:proper}), which means that if $\moutM{u} \subset \moutM{v}$ then $[u]$ is computed before $[v]$.

Let $C_1,\ldots,C_k$ be the $3$-edge-connected components of $G$.
Our plan is to compute, for each component $C_i$, a graph $G_i$, where the vertices in $C_i$ are ordinary, that satisfies \Cref{definition:2clg}. 
Similarly to the \ecc{(k+1)} decomposition,  
we construct the \ecc{3} decomposition in two phases, where the first phase processes the sets $[v]$, corresponding to the minimal \outcomp{2}s of $G$, in proper order, while the second phase processes the sets $[v]_R$, corresponding to the minimal \outcomp{2}s of $G^R$ in proper order.

\subsubsection{First phase}

Let $G$ be the given \lightgraph{2}. We choose an arbitrary ordinary vertex $s$ as a start vertex and compute the sets $[v]$, for all ordinary vertices $v$, corresponding to the minimal \outcomp{2}s $\moutM{v}$, in proper order.

Given a proper order of the sets $[v]$, our goal is to compute, for each set $[v]$, an auxiliary graph $G_{[v]}$, satisfying the following: 
\begin{definition}[First-phase auxiliary graph]
\label{definition:auxv}
An auxiliary graph $G_{[v]}$ contains two types of vertices, ordinary and auxiliary, such that:
\begin{enumerate}
    \item The ordinary vertices are precisely the vertices of $[v]$.
    \label{definition:auxv:property1}
    \item An auxiliary vertex $x$ is not $3$-edge-connected with any other vertex $y$, ordinary or auxiliary, in $G_{[v]}$.
    \label{definition:auxv:property2}
    \item Any two ordinary vertices $u$ and $v$ of $G_{[v]}$ are $k$-edge-connected in $G_{[v]}$, for $k \ge 3$, if and only if they are $k$-edge-connected in $G$.
    \label{definition:auxv:property3}
\end{enumerate}
\end{definition}

In the first phase of our algorithm, we construct auxiliary graphs $G_{[v]}$ for all sets $[v]$, of total size $\sum_{[v]} \left ( |V(G_{[v]})| + |E(G_{[v]})| \right ) = O(|V(G)| + |E(G)|)$.
\ignore{
\begin{enumerate}
    \item $V(G_{[v]})$ contains two types of vertices, ordinary and auxiliary. The ordinary vertices are precisely the vertices of $[v]$. Moreover, an auxiliary vertex $x$ is not $3$-edge-connected with any other vertex $y$, ordinary or auxiliary, in $G[v]$.
    \item Any two ordinary vertices $u$ and $v$ of $G_{[v]}$ are $k$-edge-connected in $G_{[v]}$, for $k \ge 3$, if and only if they are $k$-edge-connected in $G$.
    \item The total size of all auxiliary graphs $G_{[v]}$ is linear, i.e., $\sum_{[v]} \left ( |V(G_{[v]})| + |E(G_{[v]})| \right ) = O(|V(G)| + |E(G)|)$.
\end{enumerate}
}
We construct the graphs $G_{[v]}$ by applying the gadget substitution of \cite{GKPP:3ECC}.
This ensures that the connectivity of ordinary vertices is preserved (Lemma~3.6 of \cite{GKPP:3ECC}) and the resulting graph is a 2CLG (Lemma~3.5 of \cite{GKPP:3ECC}).
For completeness, we describe the construction of the gadgets and compute the running time.

%
We maintain an evolving graph $G'=(V',E')$, where we apply the gadget substitution for each \outcomp{2} that we identify. Initially, $G'=G$. 
Since we process the sets $[v]$ in proper order, 
$\moutM{u} \subset \moutM{v}$ 
implies that $[u]$ is computed before $[v]$.

When we process the next set $[v]$, we wish to identify the minimal \outcomp{2} \mpout{v} in the current graph $G'$ that corresponds to \mout{v} and contains all ordinary vertices of $[v]$.
We do that, by executing the local search \Cref{algorithm:localsearch}, as in \Cref{sec:kecc-decomposition}.
\ignore{
To process the next set $[v]$, we execute the local search Algorithm~\ref{algorithm:localsearch} (presented in Section~\ref{sec:proposition:local_search}), for exponentially increasing $\Delta = 2^i$, $i \in \{0,\ldots,\lceil \log_2{m} \rceil\}$.\footnote{Even though it is not necessary, one could make the following modification in lines 2 and 9 of Algorithm~\ref{algorithm:DFS}: one can return the DFS-tree path as soon as we reach any ordinary vertex $w \not \in [v]$. This is because any such $w$ has not been contracted into a 
$2$-out auxiliary vertex.}
}
When we initiate a local search from an ordinary vertex $v$ in the evolving graph $G'$, we find a minimal \outcomp{2} $S=\mpoutM{v}$. At this point, we need to (i) create the auxiliary graph $G_{[v]}$, by performing a gadget substitution for $V(G') \setminus S$, and (ii) update the evolving graph $G'$, by performing a gadget substitution for $S$.
To enable the fast execution of tasks (i) and (ii), we use the same data structures for maintaining the lists of incoming and outgoing edges as in \Cref{sec:kecc-decomposition}.
\ignore{
we maintain, for each vertex in $V(G')$, a list of incoming edges and a list of outgoing edges, organized in circular doubly-linked lists. This allows us to delete an edge from a list and merge two lists in constant time. Also, we assume that for each instance of an edge $e=(u,v)$, we maintain a pointer to the location of $e$ in the adjacency lists of $u$ and $v$.
}

When we apply a gadget substitution for a \outcomp{2} (resp., \incomp{2}) $S$, we introduce a new auxiliary vertex $s'$ with $\mathit{out}(s')=2$ (resp., $\mathit{in}(s')=2$), that represents the vertices (ordinary and auxiliary) $v \in S$ with $\mathit{out}(v) \ge 2$ (resp., $\mathit{in}(v) \ge 2$). We refer to such a vertex $s'$ as a \emph{$2$-out auxiliary vertex} (resp., \emph{$2$-in auxiliary vertex}).

\paragraph{Construction of $G_{[v]}$.}

Let $S=\mpoutM{v}$ be the next minimal \outcomp{2} found in $G'$, and let $\bar{S}=V(G') \setminus S$. Also, let $C=\{e_1=(x_1,y_1), e_2=(x_2,y_2)\}$ be the edges of the \cut{2} that defines $S$. (Note that $x_1,x_2 \in S$ and $y_1,y_2 \in \bar{S}$).
We apply the gadget replacement operation for $\bar{S}$, which replaces $G'[\bar{S}]$ (the subgraph of $G'$ induced by $\bar{S}$) with a $2$-connectivity light gadget \ingadget{\bar{S}}, as follows: 
\begin{enumerate}
    \item \ingadget{\bar{S}} contains a \emph{central vertex} $s'$, which is a $2$-in auxiliary vertex that represents the ordinary and auxiliary vertices $v \in \bar{S}$ with $\mathit{in}(v) \ge 2$. 
    \item The gadget contains $2$ entering edges $C'=\{e'_1=(x_1,y'_1), e'_2=(x_2,y'_2)\}$ corresponding to the edges of the \cut{2} $C$. We refer to the set $\{y'_1,y'_2\}$ as the set of \emph{boundary points of \ingadget{\bar{S}}}. 
    \item If $y_i$, $i \in \{1,2\}$ is an auxiliary vertex with $\mathit{in}(y_i)=1$, then let $y'_i$ be a copy of $y_i$ and insert the edge $(y'_i,s')$ into \ingadget{\bar{S}}. Otherwise, we let $y'_i=s'$. (Note that we can have $y'_1 = y'_2$.)
    \item For each auxiliary vertex $z \in \bar{S} \setminus \{y_1,y_2\}$ with $\mathit{in}(z)=1$, such that $G'$ has an edge $(z,w)$ for some $w \in S$, \ingadget{\bar{S}} contains a corresponding auxiliary vertex $z'$ with $\mathit{in}(z')=1$, and the edge $(s',z')$. 
    \item Let \gadgetsube{\bar{S}} be the graph that results from $G'$ after the following modifications. We replace $G'[\bar{S}]$ and its entering edges $C$ with \ingadget{\bar{S}} and $C'$ respectively. Also, we replace each edge $(z,w)$ directed from $\bar{S}$ to $S$ in $G'$, with a corresponding edge $(z',w)$ in \gadgetsube{\bar{S}}.
\end{enumerate}


We also apply the following reduction rule, in order to construct \gadgetsube{\bar{S}} in $O(\mathit{vol}(S))$ time and space.

\vspace{0.5cm}
\noindent\emph{Reduction Rule}: For each vertex $u \in S$ do the following.
\begin{itemize}
\item If $u$ has at least two entering edges $(z_1,u),(z_2,u)$ from $\bar{S}$ such that $z_1 \not= z_2$ or $z_1 = z_2 = z$ and $\mathit{in}(z) \ge 2$, then we insert into \gadgetsube{\bar{S}} two copies of the edge $(s',u)$. 
\item Otherwise, $u$ has a unique entering edge $(z,u)$ such that $z \in \bar{S}$, where $\mathit{in}(z)=1$.
If $z$ is not a boundary point of $\bar{S}$, then we insert into \gadgetsube{\bar{S}} a copy of $z$, together with the edge $(s',z)$. (If $z$ is a boundary point, then it is already included in \gadgetsube{\bar{S}}.)
Finally, we add a copy of the edge $(z,u)$.
\end{itemize}

To complete the construction of $G_{[v]}$, we let $G_{[v]} = \gadgetsubeM{\bar{S}}$, for $S=\mpoutM{v}$, where we execute the reduction rule during Steps 4 and 5 of the gadget replacement procedure.


\begin{lemma}
\label{lemma:gadget-replace-time-three}
Given a \outcomp{2} $S=\mpoutM{v}$ of the evolving graph $G'$, we can construct $G_{[v]}$ in $O(\mathit{vol}(S))$ time, plus $O(m)$ time for all $G_{[v]}$ graphs constructed through the algorithm.
\end{lemma}
\begin{proof}
To construct $G_{[v]}$, we first create a copy of each vertex in $v \in S$, and also copy the list of outgoing edges of $v$ in $G'$. This takes $O(\mathit{vol}(S))$ time. Then, we introduce the central vertex $s'$ of \ingadget{\bar{S}} and, for $i \in \{1,2\}$, we add a boundary vertex $y'_i$ together with the edge $(y'_i,s')$ if $y'_i \not = s'$.

Next, we need to add the necessary auxiliary vertices of $\bar{S}$ and edges from $\bar{S}$ to $S$, according to the Reduction Rule. To that end, for each vertex $v \in S$, we scan its list of entering edges in $G'$, and we store the edges $(z,u)$ such that $z \in \bar{S}$. We stop as soon as we find two such edges, or reach the end of the list of entering edges of $v$ (note that if we access multiple copies of an edge from the same vertex $v$ with $in(v)=1$, we can merge these copies in the evolving graph (and add the multiplicities of the two copies); the time to do this contributes to the $O(m)$ overall bound, as we remove one copy of an edge from the evolving graph). 
%
%
If we find two edges entering $u$ from $\bar{S}$, then we add two copies of the edge $(s',u)$ to the list of outgoing edges of $s'$. 
Suppose, now, that we find a single edge $(z,u)$. If $z$ is neither a boundary point of $\bar{S}$ nor $s'$, then we add a copy of $z$ in \ingadget{\bar{S}}. 
Finally, we add a copy of $(z,u)$ to the list of outgoing edges of $z$. 
Since we access at most two edges entering $u$ from $\bar{S}$, the overall time to scan all lists is bounded by $O(\mathit{vol}(S))$, as claimed.
This bound also includes the time to construct the list of entering edges for each vertex of $G_{[v]}$. 
\end{proof}

\paragraph{Updating $G'$.}

After we have constructed $G_{[v]}$, we need to update the evolving graph by setting $G' = \gadgetsubeM{S}$. Recall that $S=\mpoutM{v}$ is a \outcomp{2}  of the evolving graph $G'$. Let $\{x_1,x_2\}$ be the boundary points of $S$, and let $C=\{e_1=(x_1,y_1), e_2=(x_2,y_2)\}$ be the cut edges defined by $S$.
We apply the gadget replacement operation for $S$, which replaces $G'[S]$ (the subgraph of $G'$ induced by $S$) with a $2$-connectivity light gadget \outgadget{S}. This replacement procedure is analogous to the one applied for the construction of $G_{[v]}$, but slightly different since we operate directly on $G'$.
\begin{enumerate}
    \item The gadget \outgadget{S} contains a \emph{central vertex} $s'$, which is a $2$-out auxiliary vertex that represents the ordinary and auxiliary vertices $v \in S$ with $\mathit{out}(v) \ge 2$.
    \item The gadget \outgadget{S} has two outgoing edges $C'=\{e'_1=(x'_1,y_1), e'_2=(x'_2,y_2)\}$ corresponding to the edges of the \cut{2} $C$. We refer to the set $\{x'_1,x'_2\}$ as the set of \emph{boundary points of \outgadget{S}}. 
    \item If $x_i$, $i \in \{1,2\}$ is an auxiliary vertex with $\mathit{out}(x_i)=1$, then let $x'_i=x_i$ and insert the edge $(s',x_i)$ into \outgadget{S}. Otherwise, we let $x'_i=s'$. (Note that we can have $x'_1 = x'_2$ or $x'_i=s$ for some $i \in \{1,2\}$.)
    \item For each auxiliary vertex $z \in S \setminus \{x'_1,x'_2\}$ with $\mathit{out}(z)=1$, such that $G'$ has an edge $(w,z)$ for some $w \in \bar{S}$, \outgadget{S} contains $z$ and the edge $(w,z)$. 
    \item Let \gadgetsube{S} be the graph that results from $G'$ after the following modifications. We replace $G'[S]$ and its entering edges $C$ with \outgadget{S} and $C'$ respectively.
\end{enumerate}

The important difference with respect to the replacement procedure for $G_{[v]}$, is that here we do not construct a new graph from scratch, but modify $G'$. Also, we do not need (and cannot afford) to apply the Reduction Rule.

After Step 1, we need to update the head of the edges $(w,v)$, where $w \in \bar{S}$ and $v$ is an ordinary or an auxiliary vertex with $\mathit{out}(v) \ge 2$. 
%
To do that, we apply a DSU data structure, as in \Cref{sec:kecc-decomposition}.
\ignore{
We do that indirectly with a disjoint set-union (DSU) data structure~\cite{setunion:tvl,dsu:tarjan}, which we use to contract all the vertices of $S$ with out-degree at least $2$.
%
The DSU data structure supports the operation $\mathit{unite}(u,v)$, which unites the sets of two given vertices $u$ and $v$, and makes $u$ the representative of the new set. It also supports the query $\mathit{find}(v)$, which returns the representative vertex of the set containing $v$.
To perform Step 1, we chose an ordinary vertex $r \in S$ as representative, and execute $\mathit{unite}(r,v)$ for all vertices $v \in S \setminus r$ with $\mathit{out}(v) \ge 2$. 
Hence, whenever we need to traverse an edge $(w,v)$, we traverse $(w,\mathit{find}(v))$ instead.
} 
Next, we update the list of outgoing and entering edges of $v$. To do this fast, we assume that we have stored a list $L(S)$ of the edges with both endpoints in $S$, which we can do in $O(\mathit{vol}(S))$ time.
Next, we delete each edge $(u,v) \in L(S)$, both from the list of outgoing edges of $u$ and from the list of entering edges of $v$. 
Then, 
we insert the edges $(s',x_1)$ and $(s',x_2)$, and an edge $(z,s')$ for each auxiliary vertex $z \in S \setminus \{x_1,x_2\}$ with $\mathit{out}(z)=1$ that has an entering edge from $\bar{S}$.
Finally, for every vertex $v \in S \setminus r$ with $\mathit{out}(v) \ge 2$, we merge the list of entering edges of $v$ with the list of entering edges of $r$. 

The proof that $G_{[v]}$ satisfies \Cref{definition:auxv} is deferred to a later subsection. We proceed to analyze the running time of the procedure.

\ignore{
\begin{lemma}
\label{lemma:gadget-replace-correctness}
The graphs $G_{[v]}$ constructed during the first phase of the algorithm satisfy \Cref{definition:auxv}.
\end{lemma}
\begin{proof}
We only need to prove Property~\ref{definition:auxv:property1}, as the remaining properties follow immediately from the gadget construction and the correctness of the gadget substitution (Lemma~3.6 of \cite{GKPP:3ECC}).
To prove Property~\ref{definition:auxv:property1} of \Cref{definition:auxv}, consider the local search that identified \mpout{v} in $G'$. Since $G'$ is a result of gadget substitutions, we have $\mpoutM{v} \supseteq [v]$. So, it remains to argue that \mpout{v} does not contain any ordinary vertex $w \not\in [v]$.
\end{proof}
}

\begin{lemma}
\label{lemma:gadget-replace-time-2}
Given a \outcomp{2} $S=\mpoutM{v}$ of the evolving graph $G'$, we can update $G'$ in $O(\mathit{vol}(S))$ time, plus the time required for the DSU data structure.
\end{lemma}
\begin{proof}
Our procedure receives as input a list of the vertices in $S$.
It computes a list $L(S)$ of the edges with both endpoints in $S$ in $O(\mathit{vol}(S))$ time.
Then, we can remove all edges in $L(S)$ from the adjacency lists of $G'$ in $O(|L(S)|)$ time. Also, we can find the auxiliary vertices in $z \in S \setminus \{x_1,x_2\}$ with $\mathit{out}(z)=1$ that have an entering edge from $\bar{S}$ in $O(|S|)$ time, since it suffices to check if $\mathit{in}(z) \ge 1$.
Finally, we can merge the list of entering edges of each vertex $v \in S \setminus r$ with that of $r$, for all vertices with $\mathit{out}(v) \ge 2$, in a total of $O(|S|)$ time.
Thus, the total update time is $O(|L(S)|+|S|)=O(\mathit{vol}(S))$, plus the time required by the DSU data structure.
\end{proof}

After processing all sets $[v]$, such that $s \not\in [v]$, we set $G_{[s]}=G'$. The ordinary vertices of $G_{[s]}$ are the ordinary vertices $v \in V(G)$ for which $\moutM{v} = \mpoutM{v} = \bot$.

\subsubsection{Second phase}

The second phase completes the construction of the \ecc{3} decomposition $H_1, \ldots, H_t$, where each $H_i$ corresponds to a unique $3$-edge-connected component $C_i$ of $G$. 
To that end, we process each auxiliary graph $G_{[v]}$ of the first phase as follows. Recall that if $s \not\in [v]$, $G_{[v]}$ is the graph that results from the gadget substitution of the \incomp{2} $V(G') \setminus \mpoutM{v}$. Since $s \not\in \mpoutM{v}$, the central vertex $s_v$ of \ingadget{V(G') \setminus \mpoutM{v}} represents the set of ordinary vertices and the $2$-in auxiliary vertices in $V(G') \setminus \mpoutM{v}$, which includes $s$. Hence, we set $s_v$ as the start vertex of $G_{[v]}$, and apply the same procedure as the first phase, for $G'=(G_{[v]})^R$. For $G_{[s]}$, we apply the procedure of the first phase for $G'=(G_{[s]})^R$ with start vertex $s$.

\subsubsection{Correctness proof}
In the following, we prove that in each of the graphs created, the ordinary vertices are all \conn{3}.
The rest of the properties of \Cref{theorem:3eccdecomp} follow directly from our algorithm and the fact that gadget replacements preserve the connectivity of ordinary vertices (Lemma~3.6 of~\cite{GKPP:3ECC}).

Let $C_1,\ldots,C_k$ be the $3$-edge-connected components of a \lightgraph{2} $G$, consisting of the ordinary vertices of $G$ (no component contains both auxiliary and ordinary vertices, because auxiliary vertices have either only $1$ incoming or only $1$ outgoing edge).
%
%

%

\begin{lemma} \label{lem:correctnessPhaseOne}
The graphs $G_{[v]}$ constructed during the first phase of the algorithm satisfy \Cref{definition:auxv}.
\end{lemma}
\begin{proof}
We only need to prove Property~\ref{definition:auxv:property1}, as the remaining properties follow immediately from the gadget construction and the correctness of the gadget substitution (Lemma~3.6 of \cite{GKPP:3ECC}).

To prove Property~\ref{definition:auxv:property1} of \Cref{definition:auxv}, consider the local search that identified \mpout{v} in $G'$. 
All vertices of the induced subgraph of the minimum $2$-out set of any ordinary vertex $v$ are reachable from $v$ (otherwise we could remove the non-reachable vertices and get a smaller set).
Therefore, we can apply Lemma~6.5~of~\cite{GKPP:3ECC}.
When we process $[v]$, this lemma implies that for all vertices in $[v]$, the evolving graph contains a $2$-out set $S'$ whose ordinary vertices are exactly the remaining ordinary vertices of $\mu(v)$ in $G'$.
By induction, if there exists any $u$ with $\mu(u) \subset \mu(v)$, the ordinary vertices of $\mu(u)$ have already been processed (due to the proper ordering of the sets $[v]$, provided by the algorithm in~\cite{Linear3ECC}) and thus are no longer ordinary vertices in the evolving graph. 
Furthermore, the ordinary vertices $u$ of $\mu(v)$ have $\mu(u) \subseteq \mu(v)$ (otherwise, by \Cref{lemma:submodularity3} the set $\mu(u) \cap \mu(v)$ is $2$-out and separates $u$ from $s$, contradicting the minimality of $\mu(u)$).
We conclude that the ordinary vertices of $S'$ are exactly the vertices of $[v]$.

As our local search avoids $s$ by construction, and $s\not \in [v]$, it follows that $s$ is not in the aforementioned set.
Therefore the existence of $S'$ shows that the ordinary vertices of the minimum $2$-out set we find contains a subset of the remaining ordinary vertices of $[v]$ and no other ordinary vertices.

We now argue that the minimum $2$-out set $\mu'(v)$ we find contains all vertices of $[v]$.
If this was not true, then it would contain some $v_1 \in [v]$ but not some $v_2 \in [v]$, and also not $s$ (by construction of the local search algorithm).
But our gadget replacement procedure on the evolving graph guarantees that such a set would exist in the original graph (as we prove in the following claim), which contradicts the definition of $[v]$. Therefore it suffices to prove:

\begin{claim}\label{clm:persisting2OutIn:three}
Let $u,v,s$ be two ordinary vertices in the evolving graph $G'$, and suppose there exists a $2$-out (resp. $2$-in) set $S$ containing $u$ but not $v,s$.
Then, there exists a $2$-out (resp. $2$-in) set in the original graph, containing $u$ but not $v,s$.
\end{claim}
\begin{proof}
Let $G''$ be the evolving graph at some point during the first phase when we found a $2$-out set $S''$ and replaced it with a gadget such that $G' = G''_{<S>}$.
It suffices to prove that in $G''$ there exists a $2$-out (resp. $2$-in) set containing $u$ but not $v,s$, as we can repeat the argument to reach the original graph.
Note that $S''$ does not contain $u,v,s$, as these vertices are ordinary in $G'$.

The gadget replacement in the evolving graph does not remove any edge with at most one endpoint in $S''$ (thus if $S$ does not contain any vertex from the gadget, then the result follows).
The gadget replacement only contracts vertices in $S''$ with out-degree at least $2$, and with out-degree $1$ that do not have an edge from vertices not in $S''$ (let us call $X$ the set of all these vertices that are contracted, and $s'$ the new vertex formed by the contractions), and removes edges from $s'$ to non-boundary vertices of $S''$ with out-degree $1$ that have at least one edge from vertices not in $S''$ (in case of a boundary vertex with out-degree $1$ that have at least one edge from vertices not in $S''$, it keeps exactly one edge from $s'$).

Now let us assume that $S$ does not contain $s'$.
Then, from the above discussion, $S$ exists as is in $G''$, with the same number of outgoing edges. If $S$ is $2$-out, the result follows.
However, it may have more incoming edges, because of vertices of $S''$ with out-degree $1$ having edges from $X$.
Notice that the outgoing edges of such non-boundary vertices only point to vertices in $X$; if we remove these vertices from $S$ (which does not contain any vertex in $X$) then the in-degree of the resulting set does not increase.
Similarly for boundary vertices, if we remove them the resulting set loses at least one incoming edge (the one from $s'$ to the boundary vertex removed), and may have at most one more new incoming edge (the one outgoing from the boundary vertex).
We conclude that the resulting set in $G''$ has the same number of incoming edges as $S$ in $G$. Therefore if $S$ is $2$-in, the result again follows.

Finally, if $S$ contains $s'$, then let $Y = (S\setminus \set{s'}) \cup X$.
Now notice that if $S$ is $2$-in, then $Y$ is $2$-in as the only edges ``missing'' in $G'$ compared to $G''$ are edges from vertices in $X$; but as $X\subseteq Y$, these are not counted in the incoming edges to $Y$.
We can now assume that $S$ is $2$-out.
If $Y$ is not $2$-out, it is because there exists at least one edge from a vertex in $X$ to a vertex of $S''$ with out-degree $1$ that has at least one edge from vertices not in $S''$. 
Then let $Y'$ be the extension of $Y$ that includes these vertices with out-degree $1$ that have at least one edge from vertices not in $S''$.
$Y'$ is $2$-out, because the new vertices are either non-boundary, meaning they only have edges to vertices in $S''$ (but as $G''$ is a 2CLG, these must be edges to vertices with out-degree at least $2$, which are already in $Y$ and thus in $Y'$), or boundary, meaning that they only contribute one more outgoing edge, but the edges from $X$ to them are no longer contributing (recall that $S$ in $G'$ had an outgoing edge from $s'$ to each such vertex).
As $Y'$ contains $u$ (because $u\in S$) but not $v,s$, the result follows.
\end{proof}
This concludes the proof for all sets except for $[s]$.
However, as the sets $[v]$ form a partition, and after the first phase is done we put all remaining ordinary vertices of the evolving graph in $G_{[s]}$, the result follows.
\end{proof}

\ignore{(Lemma~3.6 of~\cite{GKPP:3ECC}).
This means that in the original graph in should be that $v_1$ is also
if there is such a set in the evolving graph, then there exists such a set in the original graph as well, contradicting the definition of $[v]$.
}

Now we focus on the second phase.

\begin{lemma}[Refined Cuts]
\label{lemma:refined-cuts}
Let $S$ be a \outcomp{2} of a \lightgraph{2} $G$ that contains two distinct ordinary vertices $x$ and $y$. If there is no $(\le 2)$-out set $S' \subset S$ that separates $x$ and $y$, then any \outcomp{2} that separates $x$ and $y$ must contain all ordinary and $2$-in auxiliary vertices in $V \setminus S$.
\end{lemma}
\begin{proof}
Let $X$ be a \outcomp{2} that contains $x$ but not $y$. 
Then, $out(S \cap X) > 2$, since otherwise $S \cap X$ would be a $(\le 2)$-out set contained in $S$ that separates $x$ and $y$.
By submodularity $\mathit{out}(S)+\mathit{out}(X) \ge \mathit{out}(S \cup X) + \mathit{out}(S \cap X)$, and since $\mathit{out}(S) = \mathit{out}(X) = 2$, we have $\mathit{out}(S \cup X) < 2$. Thus, $S \cup X$ must contain all ordinary and $2$-in auxiliary vertices of $G$;
otherwise, $V\setminus (S\cup X)$ contains a vertex $w$ with $in(w)>1$, but there are no two paths from $x$ to $w$, contradicting the definition of a \lightgraph{2}.
Therefore, $X$ must contain all ordinary and $2$-in auxiliary vertices in $V \setminus S$.
\end{proof}

\begin{corollary}
\label{cor:refined-cuts}
Let $x$ and $y$ be two ordinary vertices in $[v]$ that are not $3$-edge-connected in $G$. Then, there is a \outcomp{2} $X$ in $G_s^R$ that separates $x$ and $y$ that does not contain any ordinary or $2$-in auxiliary vertex in $V \setminus \moutM{v}$. In particular, $X$ does not contain $s$.
\end{corollary}
\begin{proof}
We apply \Cref{lemma:refined-cuts} for $S=\moutM{v}$. By the definition of $[v]$, and the fact that $x,y \in [v]$, we have that there is no \outcomp{2} $S' \subset \moutM{v}$ that separates $x$ and $y$. Since $x$ and $y$ are not $3$-edge-connected, we have $x \not\in \moutRM{y}$ or $y \not\in \moutRM{x}$ (Proposition~III.5 of~\cite{Linear3ECC}).
Assume, without loss of generality, that $x \not\in \moutRM{y}$. Then, $X = V \setminus \moutRM{y}$ is a \incomp{2} in $G^R$ that contains $x$ but not $y$.
Hence, $X$ is a \outcomp{2} in $G$ that separates $x$ and $y$, so by \Cref{lemma:refined-cuts}, it contains all ordinary and $2$-in auxiliary verteices of $V \setminus \moutM{v}$. Then, \moutR{y} is a \outcomp{2} in $G_s^R$ that separates $x$ and $y$ that does not contain any ordinary or $2$-in auxiliary vertex in $V \setminus \moutM{v}$.
\end{proof}

We are now ready to show that the ordinary vertices in the graphs constructed during the second phase of the algorithm are \conn{3}.

\begin{lemma}
Let $u,v$ be ordinary vertices in a graph $H_i$ constructed during the second phase of the algorithm.
Then $u,v$ are \conn{3}.
\end{lemma}
\begin{proof}
If $u,v$ are not \conn{3} and are not splitted by the first phase of the algorithm, then $[u] = [v]$ by \Cref{lem:correctnessPhaseOne}.
\Cref{cor:refined-cuts} implies that there exists a $2$-out set $X$ in $G_s^{R}$ that w.l.o.g. contains $u$ and does not contain $v,s$.
Therefore the minimum $2$-out set of $u$ in $G_s^{R}$ that does not include $s$ also does not include $v$.
Notice that $u$ reaches everything in this minimum $2$-out set, as otherwise we could remove the unreachable vertices and create an even smaller $2$-out set separating $u$ from $s$.
Therefore we can apply Lemma~6.5 of~\cite{GKPP:3ECC} which implies that $G_{[v]}$ contains a $2$-out set $S'$ containing $u$ but not $v,s$; in turn this means that the minimum $2$-out set of $u$ in $G_{[v]}$ does not contain $v,s$.
As the second phase is the same with the first phase in the graph $G_{[v]}$ instead of $G$, by \Cref{lem:correctnessPhaseOne} we get that in the end $u$ is in a different graph $H_i$ from $v$.
\end{proof}

This completes the proof of correctness. 

\begin{corollary}\label{lem:correctnessPhaseTwo:three}
The graphs constructed during the second phase of the algorithm satisfy \Cref{theorem:3eccdecomp}.
\end{corollary}

\subsubsection{Running time}

Now we bound the running time of our construction and the total size of the constructed auxiliary graphs. We only analyze the construction of the first phase, since the second phase executes the same construction on the reverse graph of each first-phase auxiliary graph.

\ignore{
\begin{lemma}
The total time to identify all \mpout{v} sets and maintain the evolving graph $G'$ during the first phase of the algorithm is $O(m+n)$, plus the time required by the DSU data structure. 
\end{lemma}
\begin{proof}
Let $A(v)$ be the set of $2$-level in-auxiliary vertices in $M'(v)$.
From the definition of the $2$-connectivity light gadget, we can verify that during the first phase, the following invariant is maintained. 
The in-volume $\vin[,G']{M'(v)}$ of any minimal \incomp{2} $M'(v)$ that is identified in the evolving graph $G'$, is proportional to $|A(v)|+\sum_{u \in [v]}{in(u)}$, where the sum is taken for all ordinary vertices $u \in [v]$, which are the ordinary vertices that are visited for the first time while searching for $M'(v)$. 

To justify this claim, note that every auxiliary vertex $x$ with $in(x)=1$ has an entering edge $(z,x)$, where $z$ is either ordinary, in which case $z \in [v]$, or $2$-level in-auxiliary, in which case $z \in A(v)$. Moreover, for any $M'(u) \subset M'(v)$, the algorithm identified $M'(u)$ before the search for $M'(v)$ and replaced $M'(u)$ with \ingadget{M'(u)}. Each vertex in \ingadget{M'(u)} is reached by $O(1)$ edges in \ingadget{M'(u)}.

Now, we observe that $\sum_{[v]}{|A(v)|} = O(n)$. This is because all the $2$-level in-auxiliary vertices in $M'(v)$, together with the ordinary vertices in $[v]$, are contracted into a $2$-level in-auxiliary vertex $s'$ that becomes the central vertex of \ingadget{M'(v)}.

Finally, since every ordinary and $2$-level in-auxiliary vertex in $[v]$ is processed once, and we have at most $n$ such sets $[v]$, the total time to identify all $M'(v)$ sets is bounded by $O(n+\sum_{u, in(u)>1}{in(u)}) = O(n+m)$.
Recall that for the construction of \gadgetsube{M'(v)} we do not apply pruning on \ingadget{M'(v)}, since we cannot afford to spend $\vout[,G']{M'(v)}$ time for each identified \incomp{2} $M'(v)$.
\end{proof}
}

\begin{lemma}
\label{lemma:first-phase-evolving}
The total time to identify all $\mpoutM{v}$ sets and maintain the evolving graph $G'$ during the first phase of the algorithm is $O(m+n)$, plus the time to perform $O(m+n)$ operations on the DSU data structure. 
\end{lemma}
\begin{proof}
Directly by our algorithm and the linear running time of our local search, we have that the time to find an $\mpoutM{v}$ is $O(\sum_{i=0}^{\lceil \log_2{vol(\mpoutM{v})}\rceil} 2^i) = O(vol(\mpoutM{v}))$.
Furthermore, $vol(\mpoutM{v}) \le 2 \sum_{u\in \mpoutM{v}, out(u)>1} out(u)$, by Lemma~6.10 of~\cite{GKPP:3ECC}; intuitively, this means we can ignore vertices with only one outgoing edge (because we anyway account for an incoming edge to them, as there cannot be two consecutive vertices with out-degree $1$ in a 2CLG).
The gadget replacement operation removes all vertices with at least $2$ out-edges, and introduces a new vertex with $2$ outgoing edges; as we perform at most $n$ gadget replacements, the extra edges added this way are at most $2n$.

We conclude that the running time to identify all $\mpoutM{v}$ sets is at most a constant times the sum of out-degrees of all vertices of $G$, plus $O(n)$, which gives the desired $O(m+n)$.
The lemma then directly follows from \Cref{lemma:gadget-replace-time-2}.
\end{proof}

\begin{lemma}
\label{lemma:first-phase-auxiliary}
The total size of all constructed graphs during the first phase of the algorithm is $O(m+n)$.
The total time to construct all auxiliary graphs $G_{[v]} = \gadgetsubeM{V' \setminus \mpoutM{v}}$ during the first phase of the algorithm is $\Ot(m+n)$. 
\end{lemma}
\begin{proof}
By \Cref{lemma:gadget-replace-time-three}, when we identify the next minimal \outcomp{2} \mpout{v} in the evolving graph $G'$, we create $G_{[v]} = \gadgetsubeM{V' \setminus \mpoutM{v}}$ in $O(\mathit{vol}(\mpoutM{v}))$ time (plus $O(m)$ time in total, which is negligible).
So, the total size of all auxiliary graphs $G_{[v]}$ is bounded by the time to construct them, as they are all explicitly constructed.
This is $O(\sum_{[v]}{\mathit{vol}(\mpoutM{v})}) = O(m+n)$, exactly as argued in \Cref{lemma:first-phase-evolving}.

The running time is dominated by \Cref{lemma:first-phase-evolving}.
It is $\Ot(m+n)$, using e.g. the DSU data structure of~\cite{dsu:tarjan} or of~\cite{setunion:tvl}.
\end{proof}

\ignore{
If for the DSU data structure we use a structure that supports each $\mathit{find}$ operation in worst-case $O(1)$ time and
any sequence of $\mathit{unite}$ operations in total time $O(n \log n)$~\cite{setunion:tvl}, then we obtain an algorithm that constructs the \ecc{3} decomposition in $O(m+n\log{n})$ total time.
Alternatively, we implement $\mathit{unite}$ and $\mathit{find}$ using compressed trees with appropriate heuristics \cite{dsu:tarjan}, and the total time for the DSU operations is $O(m \alpha(m,n))$. 
}

%% file: localSearchDeterministic.tex

\section{Local search for $M$-sets}
\label{sec:proposition:local_search}

\subsection{Deterministic local search}

In this section, we prove \Cref{proposition:local_search}, which provides a deterministic local search algorithm.
We note that it is a straightforward modification of the local search algorithm in~\cite{chechik2017faster}.
The main difference here is that we want our local search to avoid a specific vertex $s$\footnote{One could prove the same result using a black-box reduction, by modifying the original graph so that vertex $s$ has very large out-degree. This would force the black-box local search algorithm to avoid $s$. For our purposes, it was sufficient to modify the existing algorithm from~\cite{chechik2017faster}, which also allowed us to only use the original graph.}.

\begin{lemma}
\label{lemma:reverse}
Let $S$ be a set of vertices, let $P$ be a path that starts from a vertex in $S$, and let $G'$ be the graph that is formed by reversing the direction of the edges of $P$. Then we have the following:
\begin{itemize}
\item{If the end of $P$ is in $S$, then $\mathit{vol}_{G'}(S)=\mathit{vol}_{G}(S)$ and $\mathit{out}_{G'}(S)=\mathit{out}_{G}(S)$.}
\item{If the end of $P$ is outside of $S$, then $\mathit{vol}_{G'}(S)=\mathit{vol}_{G}(S)-1$ and $\mathit{out}_{G'}(S)=\mathit{out}_{G}(S)-1$.}
\end{itemize}
\end{lemma}
\begin{proof}
We prove both statements together using induction on the length of $P$ (defined as the number of edges that is uses). For the basis of our induction, let us assume that $P$ consists of a single edge $(x,y)$. Thus, we have $x\in S$, and we distinguish two cases, depending on whether $y\in S$ and $y\notin S$. In either case, $G'$ is formed by simply replacing $(x,y)$ with $(y,x)$. Now, if $y\in S$, then it is obvious that the volume of $S$ and its number of outgoing edges have not changed in $G'$. Similarly, if $y\notin S$, then it is also obvious that the volume of $S$, and its number of outgoing edges, have been reduced by one in $G'$.

Now let us assume that both statements are true for all paths that start from $S$ and have length $l$, for some $l\geq 1$. Now consider a path $P$ with length $l+1$ that starts from $S$. Let $(x,y)$ be the last edge used by $P$, let $P'$ be the initial segment of length $l$ of $P$, let $G'$ be the graph that is formed by reversing the edges of $P'$, and let $G''$ be the graph that is formed by reversing the edges of $P$. 

First, let us assume that $y\in S$. If $x\in S$, then, by the inductive hypothesis, we have $\mathit{vol}_{G'}(S)=\mathit{vol}_{G}(S)$ and $\mathit{out}_{G'}(S)=\mathit{out}_{G}(S)$. Since $G''$ is formed by reversing $(x,y)$ in $G'$, it is then easy to see that $\mathit{vol}_{G''}(S)=\mathit{vol}_{G'}(S)$ and $\mathit{out}_{G''}(S)=\mathit{out}_{G'}(S)$. Therefore, we have $\mathit{vol}_{G''}(S)=\mathit{vol}_{G}(S)$ and $\mathit{out}_{G''}(S)=\mathit{out}_{G}(S)$. If $x\notin S$, then, by the inductive hypothesis, we have $\mathit{vol}_{G'}(S)=\mathit{vol}_{G}(S)-1$ and $\mathit{out}_{G'}(S)=\mathit{out}_{G}(S)-1$. Then, since $G''$ is formed by reversing $(x,y)$ in $G'$, it is easy to see that $\mathit{vol}_{G''}(S)=\mathit{vol}_{G'}(S)+1$ and $\mathit{out}_{G''}(S)=\mathit{out}_{G'}(S)+1$. Therefore, we have again $\mathit{vol}_{G''}(S)=\mathit{vol}_{G}(S)$ and $\mathit{out}_{G''}(S)=\mathit{out}_{G}(S)$.

The case where $y\notin S$ is treated similarly.
\end{proof}

\begin{corollary}
\label{corollary:edge-conn-after-reverse}
Let $v$ and $s$ be two distinct vertices, let $P$ be a path that starts from $v$, and let $G'$ be the graph that is formed by reversing the direction of the edges of $P$. Then, $\lambda_{G'}(v,s)\in\{\lambda_{G}(v,s),\lambda_{G}(v,s)-1\}$.
\end{corollary}
\begin{proof}
Let $k=\lambda_{G}(v,s)$, and let $S$ be a $k$-cut that separates $v$ and $s$. By \Cref{lemma:reverse}, we have $\mathit{out}_{G'}(S)\leq\mathit{out}_{G}(S)$. Thus, the edge-connectivity between $v$ and $s$ in $G'$ cannot be higher than that in $G$. Now let us suppose, for the sake of contradiction, that there is a $k'$-out set $S'$ in $G'$ that separates $v$ and $s$, with $k'<k-1$. Then, $S'$ is a cut that separates $v$ and $s$ in $G$, and thus by \Cref{lemma:reverse} we have $\mathit{out}_{G}(S')-1\leq\mathit{out}_{G'}(S')$. Since $k'=\mathit{out}_{G'}(S')$, this implies that $\mathit{out}_{G}(S')-1\leq k'$. Therefore, since $k'<k-1$, this implies that $\mathit{out}_{G}(S')<k$, which contradicts $\lambda_{G}(v,s)=k$. This shows that the edge-connectivity between $v$ and $s$ in $G'$ either has stayed the same as in $G$, or it has been reduced by one. 
\end{proof}

\begin{corollary}
\label{corollary:minimum-set-after-reverse}
Let $v$ and $s$ be two distinct vertices of $G$, let $k=\lambda_{G}(v,s)>0$, let $S$ be the (inclusion-wise) minimum $k$-out set that separates $v$ and $s$, and let $P$ be a path that starts from $v$ and ends outside of $S$. Let $G'$ be the graph that is formed from $G$ by reversing the direction of the edges of $P$. Then, $S$ is the (inclusion-wise) minimum $(k-1)$-out set that separates $v$ and $s$ in $G'$.
\end{corollary}
\begin{proof}
By \Cref{lemma:reverse}, we have $\mathit{out}_{G'}(S)=k-1$. By \Cref{corollary:edge-conn-after-reverse}, we have $\lambda_{G'}(v,s)=k-1$. Now let us suppose, for the sake of contradiction, that there is a proper subset $S'$ of $S$ that is a $(k-1)$-out set in $G'$ that separates $v$ and $s$. Then, we have that $P$ (in $G$) ends outside of $S'$. Thus, \Cref{lemma:reverse} implies that $\mathit{out}_{G'}(S')=\mathit{out}_{G}(S')-1$. But since $S'$ is a $(k-1)$-out set in $G'$, this implies that $S'$ is a $k$-out set in $G$, and thus we get a contradiction to the minimality of $S$. 
\end{proof}

\begin{lemma}
\label{lemma:reversing_paths}
Let $v$ and $s$ be two distinct vertices with $\lambda(v,s)=k$, and let $S$ be the minimum $k$-out set in $G$ that separates $v$ and $s$. Let $P_1,G_1,P_2,G_2,\dots,P_k,G_k$ be a sequence of paths and graphs with the following properties. $P_1$ is a path that starts from $v$, and $G_1$ is the graph that is formed by reversing the direction of the edges of $P_1$. Then, for every $i\in\{2,\dots,k\}$, $P_i$ is a path in $G_{i-1}$ that starts from $v$, and $G_i$ is the graph that is formed from $G_{i-1}$ by reversing the direction of the edges of $P_i$. Then we have the following:
\begin{itemize}
\item{If every path $P_i$, for $i\in\{1,\dots,k\}$, ends outside of $S$, then $S=\mathit{reach}_{G_k}(v)$.}
\item{If at least one path among $P_1,\dots,P_k$ ends in $S$, then $v$ can reach $s$ in $G_k$.}
\end{itemize}
\end{lemma}
\begin{proof}
First, let us assume that every path $P_i$, for $i\in\{1,\dots,k\}$, ends outside of $S$. Then, a repeated application of \Cref{lemma:reverse} implies that $S$ is a $0$-out set in $G_k$. Thus, a graph exploration starting at $v$ in $G_k$ will reach a set of vertices $S'\subseteq S$. Now let us suppose, for the sake of contradiction, that $S'\neq S$. Therefore, we have $S'\subset S$. Then, since $S'$ is a cut that separates $v$ and $s$, and $\lambda(v,s)=k$, and $S$ is the minimum $k$-out set in $G$ that separates $v$ and $s$, we have $\mathit{out}_{G}(S')>k$. Notice that every path $P_i$, for $i\in\{1,\dots,k\}$, ends outside of $S'$. Then, a repeated application of \Cref{lemma:reverse} implies that $S'$ is a $(\geq 1)$-out set in $G_k$. But this contradicts the fact that $S'$ is the reachability set of a vertex of $G_k$. This shows that $S=\mathit{reach}_{G_k}(v)$.

Now let us assume that at least one path among $P_1,\dots,P_k$ ends in $S$, and let $R=\mathit{reach}_{G_k}(v)$. Since $R$ is a reachability set of a vertex of $G_k$, we have $\mathit{out}_{G_k}(R)=0$. Since every path $P_i$, for $i\in\{1,\dots,k\}$, starts from $v$, a repeated application of \Cref{lemma:reverse} implies that $\mathit{out}_{G_k}(R)\geq\mathit{out}_{G}(R)-k$. Thus, since $\mathit{out}_{G_k}(R)=0$, we infer that $\mathit{out}_{G}(R)\leq k$. 

Now let us suppose, for the sake of contradiction, that $s\notin R$. This implies that $R$ is a $(v,s)$-cut. Thus, since $\lambda(v,s)=k$, we have $\mathit{out}_{G}(R)\geq k$. Therefore, $\mathit{out}_{G}(R)\leq k$ implies that $\mathit{out}_{G}(R)=k$. Then, since $S$ is the minimum $k$-cut that separates $v$ and $s$, we have $S\subseteq R$. Therefore, since at least one path among $P_1,\dots,P_k$ ends in $S$, we have that at least one path among $P_1,\dots,P_k$ ends in $R$. Then, a repeated application of \Cref{lemma:reverse} implies that $\mathit{out}_{G_k}(R)>0$, a contradiction. This shows that $v$ can reach $s$ in $G_k$.
\end{proof}

\begin{lemma}
\label{lemma:DFS}
Let $v$ and $s$ be two distinct vertices of $G$, and let $k\geq 1$ and $\Delta\geq 1$ be two integer parameters. There is an algorithm $\mathtt{FindOutPaths}(G,v,s,k,\Delta)$ (see \Cref{algorithm:DFS}) that runs in $O(k^2\Delta)$ time, and returns a collection of at most $2k$ paths that start from $v$, with the following guarantee:
\begin{itemize}
\item{If there is a $k$-out set $S$ with $\mathit{vol}(S)\leq\Delta$ that separates $v$ and $s$, then at least one of the paths returned by the algorithm ends outside of $S$.}
\end{itemize}
\end{lemma}
\begin{proof}
The algorithm works by performing a depth-first search (DFS) of the graph, starting at $v$, with budget $(2k+1)(\Delta+1)$ on the exploration of edges. That is, we will run a DFS starting at $v$, but we will stop once we have explored $(2k+1)(\Delta+1)$ entries of the adjacency lists. (However, we may stop even earlier, if we happen to meet $s$.)

The clever insight from \cite{chechik2017faster} is that we can pause the exploration every $\Delta+1$ steps, in order to pick a candidate vertex (and then continue the exploration as usual), so that, in the end, if there is a $k$-out set $S$ with $\mathit{vol}(S)\leq\Delta$ that separates $v$ and $s$, at least one of the vertices that we have gathered lies outside of $S$.

In order to describe the algorithm and prove its correctness, we have to make careful use of the concepts that are involved in a depth-first search. Let us assume that the adjacency list of every vertex $x$ consists of the vertices $y$ for which there exists an edge $(x,y)$ (where the same vertex is allowed to appear multiple times in an adjacency list, to account for the existence of parallel edges). Initially, all vertices are considered ``unvisited''. When the DFS starts, $v$ is marked as ``visited''. Then, every time we meet a vertex $x$ in an adjacency list, there are two possibilities: either the vertex is ``unvisited'', or ``visited''. If $x$ is ``unvisited'', then we mark it as ``visited'', and at this point we say that we ``discover'' $x$. However, if $x$ is already marked as ``visited'', then we simply say that we ``meet'' $x$ (i.e., when we say that we ``meet'' a vertex, we mean that we have already discovered it). We number the vertices according to the order of their discovery. Thus, $v$ is number $1$, the first vertex in the adjacency list of $v$ is number $2$, etc. Every time a new vertex $x\neq v$ is discovered, this is because we were traversing the adjacency list of a vertex $p$, and there we met $x$ for the first time (as an ``unvisited'' vertex). We call $p$ ``the parent'' of $x$, and we say that the edge $(p,x)$ is ``traversed''. Then, we maintain the edge $(p,x)$, and thus we build a directed tree $T$, which is a subgraph of $G$, and it is rooted at $v$. The tree $T$ is important, because it will be used in order to provide our desired collection of paths.

Due to the recursive nature of the DFS, there appears the concept of ``backtracking''. Specifically, each time the DFS has explored everything it could after the discovery of a vertex $x\neq v$ (and so it has reached the ``end-of-list'' entry of the adjacency list of $x$), it backtracks to its parent $p$, and the search continues from the next entry in the adjacency list of $p$ after (the first occurrence of) $x$. At this point, we say that the DFS has ``backtracked over the edge $(p,x)$''. 

Thus, at any point during the DFS, we have the following characterization of the edges of the graph: ``unvisited'', ``traversed'', ``explored'', ``backtracked over''. We say that an edge $(x,y)$ is explored, when we hit the entry $y$ in the adjacency list of $x$ that corresponds to $(x,y)$. ($y$ may be either visited or unvisited.) Thus, every exploration of an edge consumes one unit from our budget of $(2k+1)(\Delta+1)$ on the exploration of edges. Notice that every traversed edge is explored, but the converse is not necessarily true. (I.e., when we explore an edge, we do not necessarily traverse it.) Furthermore, only traversed edges may be backtracked over, and this can happen at most once for each traversed edge.

Now the description of the algorithm is very simple. (See \Cref{algorithm:DFS}.) First, we run a DFS for $\Delta+1$ steps (i.e., we explore $\Delta+1$ edges). If within those $\Delta+1$ initial steps we have discovered $s$, then we terminate the exploration, and we return the tree path from $v$ to $s$. Otherwise, let $x$ be the current vertex on which the DFS was paused. Then we repeat the following process for $2k$ times. We resume the DFS from $x$ for $\Delta+1$ steps. If we discover $s$, then we terminate the exploration, and we return the tree path from $v$ to $s$. Otherwise, we pause the exploration, and there are two possibilities: either the resumed DFS has backtracked at some point to an ancestor of $x$, or the resumed DFS did not backtrack to an ancestor of $x$. In the first case, let $z$ be the lowest ancestor of $x$ that the resumed DFS has backtracked to, and in the second case let $z=x$. Then we store in a collection of paths (that will potentially be the output) the tree path from $v$ to $z$.

Now we will prove the correctness of this procedure. So let us assume that there is a $k$-out set $S$ with $\mathit{vol}(S)\leq\Delta$ that separates $v$ and $s$. If we happen to meet $s$ during the course of the algorithm, then the exploration will terminate abruptly, and we will correctly return a path that starts from $v$ and ends in $s$. So let us assume that we did not meet $s$ during the course of the algorithm.

Since $\mathit{vol}(S)\leq\Delta$, we have that the first exploration of $\Delta+1$ edges must have left $S$ at some point (although it may have returned back thereafter). Thus, at least one outgoing edge of $S$ was traversed during the first exploration. Now, for every subsequent continuation of the DFS exploration for $\Delta+1$ edges (i.e., every time we execute a block of the \textbf{for} loop in Line~\ref{line:DFS-for}), there are three possibilities: 

\begin{enumerate}[label={(\arabic*)}]
\item{The resumed DFS has explored at least one outgoing edge of $S$.}
\item{The resumed DFS has backtracked over at least one outgoing edge of $S$.}
\item{Neither $(1)$ nor $(2)$ took place.}
\end{enumerate}

Since the first exploration has traversed an outgoing edge of $S$, and since the \textbf{for} loop in Line~\ref{line:DFS-for} is repeated $2k$ times, and since $S$ has precisely $k$ outgoing edges, and since an edge can be explored at most once and can be backtracked over at most once, we have that $(3)$ must have been the case for at least one iteration of the \textbf{for} loop in Line~\ref{line:DFS-for}.

Now let us see what happened during an iteration of the \textbf{for} loop in Line~\ref{line:DFS-for} where no outgoing edge of $S$ was explored, and no outgoing edge of $S$ was backtracked over. Let $x$ be the vertex on which the DFS was paused before this iteration began. Now there are two possibilities: either $x$ is outside of $S$, or $x$ is in $S$. First, let us assume that $x$ is outside of $S$. If the resumed DFS will not backtrack to an ancestor of $x$, then in Line~\ref{line:DFS-add-path} we collect the DFS-tree path from $v$ to $x$. So let us assume that, at some point, the resumed DFS has backtracked to an ancestor of $x$. Then, since no backtracking over an outgoing edge of $S$ took place, we have that all the ancestors of $x$ to which the resumed DFS has backtracked are also outside of $S$. Thus, in Line~\ref{line:DFS-add-path} we collect a path from $v$ to a vertex outside of $S$. Now let us assume that $x$ is in $S$. Then, since we explore $\Delta+1$ edges, and since $\mathit{vol}(S)\leq\Delta$, and since the resumed DFS did not explore an outgoing edge of $S$, the resumed DFS must have backtracked to a vertex outside of $S$ (which must be an ancestor of $x$). And then, since no backtracking over an outgoing edge of $S$ took place, we have that all further backtrackings were to vertices that were also outside of $S$. Thus, the lowest ancestor of $x$ which the resumed DFS has backtracked to is a vertex outside of $S$, and therefore in Line~\ref{line:DFS-add-path} we collect a path from $v$ to a vertex outside of $S$.

The $O(k^2\Delta)$ time bound is easy to see. Since we perform a DFS with budget $O(k\Delta)$ on the number of explored edges, every path that we return has length $O(k\Delta)$. Thus, since the output consists of at most $2k$ paths, the total size of the paths that we return is $O(k^2\Delta)$. This is clearly an upper bound on the running time of the algorithm.
\end{proof}

\begin{algorithm}[h!]
\caption{$\mathtt{FindOutPaths}(G,v,s,k,\Delta)$}
\label{algorithm:DFS}
\LinesNumbered
\DontPrintSemicolon

explore $\Delta+1$ edges, starting from $v$, in a DFS manner\;
\If{$s$ was discovered}{
  \textbf{return} the DFS-tree path from $v$ to $s$\;
}
let $\mathcal{P}\leftarrow\emptyset$ \tcp{the collection of paths to be potentially returned}
\For{$2k$ times}{
\label{line:DFS-for}
  let $x$ be the vertex on which the DFS was paused\;
  explore $\Delta+1$ extra edges\;
  \If{$s$ was discovered}{
    \textbf{return} the DFS-tree path from $v$ to $s$\;
  }
  $z\leftarrow x$\;
  \If{the resumed DFS has backtracked to an ancestor of $x$}{
    let $z$ be the lowest ancestor of $x$ to which the resumed DFS has backtracked\;
  }
  add to $\mathcal{P}$ the DFS-tree path from $v$ to $z$\;
  \label{line:DFS-add-path}
}
\textbf{return} $\mathcal{P}$\;
\end{algorithm}

Now we are ready to prove \Cref{proposition:local_search}, which we restate here for convenience.

\localsearch*
\begin{proof}
The idea is to utilize \Cref{lemma:reversing_paths}. To be specific, let us assume that $\lambda(v,s)=k+1$, and let $U=M(v)$ and $\mathit{vol}(U)\leq\Delta$. Then, we need to find a sequence $P_1,P_2,\dots,P_{k+1}$ of paths with the following properties. $P_1$ is a path of $G$ that starts from $v$ and ends outside of $U$. Let $G_1$ be the graph that is formed from $G$ by reversing the edges of $P_1$. Then, for every $i\in\{1,\dots,k\}$, supposing that we have found a path $P_i$ and we have defined a graph $G_i$, we have that $P_{i+1}$ is a path of $G_i$ that starts from $v$ and ends outside of $U$, and $G_{i+1}$ is formed from $G_i$ by reversing the edges of $P_{i+1}$. Then, a repeated application of \Cref{lemma:reverse} shows that: first, the reachability set of $v$ in $G_{k+1}$ coincides with $U$, and second, $\mathit{vol}_{G_{k+1}}(U)=\mathit{vol}_{G}(U)-(k+1)$, and thus $v$ reaches at most $\Delta-(k+1)$ edges in $G_{k+1}$.

In order to find the desired paths $P_1, P_2,\dots,P_{k+1}$, we use the procedure $\mathtt{FindOutPaths}$, whose guarantees are stated in \Cref{lemma:DFS}. Specifically, we first run $\mathtt{FindOutPaths}(G,v,s,k+1,\Delta)$. This will return at most $2(k+1)$ paths that start from $v$, with the property that at least one of them ends outside of $U$. Since we do not know for sure which of those paths will work for us, we process each of them separately. Thus, for each of them, we reverse its edges, we get a graph $G'$, and then we run $\mathtt{FindOutPaths}(G',v,s,k,\Delta)$. (Note that after processing a path, we undo the reversals that we did to the graph.) Then we repeat this process, as shown in \Cref{algorithm:localsearch}.

Notice that every time that we call $\mathtt{LocalSearchForMSet}(G,v,s,k,\Delta)$ with $k>0$, this creates a set of calls of the form $\mathtt{LocalSearchForMSet}(G',v,s,k-1,\Delta)$. If  one of those calls returns a non-empty set, then the algorithm terminates and outputs this set. Now, we have essentially shown, that, if $\lambda(v,s)=k+1$, then there is a branch of recursive calls that will return $U$. However, we must also show that there will be no incorrect outputs (i.e., non-empty outputs which are distinct from $U$). But this follows easily from \Cref{lemma:reversing_paths}: if throughout a branch of recursive calls to $\mathtt{LocalSearchForMSet}$ we process a path (in Line~\ref{line:localsearch_choosepath}) which does not end outside of $U$, then all terminal calls, which have the form $\mathtt{LocalSearchForMSet}(\cdot,v,s,0,\Delta$), will fail to return a non-empty set, because $v$ will always be able to reach $s$ (and so Line~\ref{line:localsearch_empty} will return $\emptyset$).  

In order to complete the proof of correctness, we have to consider the case where either $\lambda(v,s)>k+1$, or $\lambda(v,s)=k+1$ and $\mathit{vol}(M(v))>\Delta$. First, let us assume that $\lambda(v,s)>k+1$. Then, if we call $\mathtt{LocalSearchForMSet}(G,v,s,k+1,\Delta)$, by \Cref{corollary:edge-conn-after-reverse} we have that, in every (terminal) recursive call of $\mathtt{LocalSearchForMSet}(\cdot,v,s,0,\Delta)$, the edge-connectivity between $v$ and $s$ is larger than $0$. Therefore, in those terminal calls, either $v$ can reach at least $\Delta+1$ edges, or it can reach $s$, and so the output will always be $\emptyset$.

Now let us consider the case where $\lambda(v,s)=k+1$ and $\mathit{vol}(M(v))>\Delta$. Then, for every branch of recursive calls to $\mathtt{LocalSearchForMSet}$, there are two possibilities: either $(1)$ every call was initiated after the reversal of a path that starts from $v$ and ends outside of $M(v)$, or $(2)$ at least one such call was initiated after the reversal of a path that starts from $v$ and ends in $M(v)$. In the first case, when we reach the terminal call to $\mathtt{LocalSearchForMSet}$, by \Cref{lemma:reversing_paths} we have that the edge-connectivity of $v$ and $s$ on the input graph $G'$ has dropped to $0$, and the set of vertices reachable by $v$ on $G'$ coincides with $M(v)$ (where $M(v)$ refers to the minimal $(k+1)$-out set in the original graph $G$). Thus, if the volume of $M(v)$ in $G'$ happens to be at most $\Delta$ (which may be the case, since $\mathit{vol}_{G'}(M(v))=\mathit{vol}_{G}(M(v))-(k+1)$, as an implication of \Cref{lemma:reverse}), then $M(v)$ will be returned. Otherwise, this call will return $\emptyset$. Now let us consider case $(2)$. Then, \Cref{lemma:reversing_paths} implies that, when we reach the terminal call to $\mathtt{LocalSearchForMSet}$, the edge-connectivity of $v$ and $s$ on the input graph $G'$ is greater than $0$. Thus, in this terminal call, either $v$ can reach at least $\Delta+1$ edges, or it can reach $s$, and so the output will be $\emptyset$.

The $O(2^k(k+1)!\Delta)$ time bound is an immediate consequence of the guarantees of $\mathtt{FindOutPaths}$, provided by \Cref{lemma:DFS}, and the recursive structure of \Cref{algorithm:localsearch}, $\mathtt{LocalSearchForMSet}$. Specifically, by calling $\mathtt{LocalSearchForMSet}(G,v,s,k+1,\Delta)$, first we perform a call to $\mathtt{FindOutPaths}(G,v,s,k+1,\Delta)$. By \Cref{lemma:DFS}, this takes time $O((k+1)^2\Delta)$, and returns a collection of at most $2(k+1)$ paths. Then, for each of those paths, the recursive call to $\mathtt{LocalSearchForMSet}$ will perform a call to $\mathtt{FindOutPaths}$ with parameters $k$ and $\Delta$. Thus, there are at most $2(k+1)$ calls to $\mathtt{FindOutPaths}$ that each takes time $O(k^2\Delta)$ and returns a collection of at most $2k$ paths. Continuing in this way, there will be $2^{(k+1)}(k+1)!$ terminal calls to $\mathtt{LocalSearchForMSet}$, each of which simply performs a DFS with budget $\Delta+1$ on the number of explored edges (see Line~\ref{line:terminalcall}). Thus, the total running time can be bounded by $O(2^k(k+1)!\Delta)$.
\end{proof}

\begin{algorithm}[h!]
\caption{$\mathtt{LocalSearchForMSet}(G,v,s,k,\Delta)$}
\label{algorithm:localsearch}
\LinesNumbered
\DontPrintSemicolon
\If{$k=0$}{
  perform a DFS on $G$, starting from $v$, with budget $\Delta+1$ on the number of explored edges, and let $S$ be the set of explored vertices\;
  \label{line:terminalcall}
  \If{$s\notin S$ \textbf{and} at most $\Delta$ edges were explored}{
    \textbf{return} $S$\;
  }
  \textbf{return} $\emptyset$\;
  \label{line:localsearch_empty}
}
let $\mathcal{P}$ be the collection of paths return by $\mathtt{FindOutPaths}(v,s,k,\Delta)$ on $G$\;
\ForEach{path $P\in\mathcal{P}$}{
\label{line:localsearch_choosepath}
  let $G'$ be the graph that is formed by reversing the direction of the edges of $P$\;
  let $S$ be the set that is returned by $\mathtt{LocalSearchForMSet}(G',v,s,k-1,\Delta)$\;
  \If{$S\neq\emptyset$}{
    \textbf{return} $S$\;
  }
}
\textbf{return} $\emptyset$\;
\end{algorithm}

\input{localSearchRandomized2}

%% file: localSearchRandomized2.tex
\subsection{Randomized local search}
\label{sec:proposition:random_local_search}

Let $v$ and $s$ be two vertices such that $\lambda(v,s)=k$, and let $M(v)$ denote the inclusion-wise minimum $(v,s)$-mincut. Here we provide a randomized (Monte Carlo) algorithm for computing $M(v)$ in the case where $\mathit{vol}(M(v))\leq\Delta$, where $\Delta$ is a parameter specified by the user. The running time of this algorithm is $O(k^2\Delta)$, and thus the dependency on $k$ is significantly improved compared to \Cref{proposition:local_search}.

Our algorithm is an adaptation of the randomized local search procedure of \cite{forster2020computing} (given in Corollary A.1 in \cite{forster2020computing}). Specifically, the general idea is the same as in the proof of \Cref{proposition:local_search}: first, we need to find $k$ paths $P_1,\dots,P_k$ with the property that $P_1$ starts from $v$ and ends outside of $M(v)$, and each $P_i$, for $i\in\{2,\dots,k\}$, has the same property on the graph that is formed by reversing the direction of the edges of the paths $P_1,\dots,P_{i-1}$. And then, in order to find $M(v)$, it is sufficient to run a graph exploration, starting from $v$, on the graph that is formed by reversing the direction of the edges of the paths $P_1,\dots,P_k$, with budget $\Delta$ on the exploration of edges.

The important observation of \cite{forster2020computing} is that, if we are willing to allow for some probability of failure (so that $M(v)$ may not be discovered, and $\emptyset$ is returned instead), a desired collection of paths $P_1,\dots,P_k$ can be discovered with a random selection of edges from a simple graph exploration. Specifically, due to the bound $\Delta$ on $\mathit{vol}(M(v))$, it is sufficient to just explore enough edges, starting from $v$, using a graph exploration such as BFS or DFS, and then with uniform sampling we have a good chance to pick an edge whose tail is not in $M(v)$. With the appropriate modifications to the procedure of \cite{forster2020computing}, we get the following:

\randomizedlocalsearch*
\begin{proof}
\Cref{algorithm:randomizedlocalsearch} works by repeatedly finding paths starting from $v$ and reversing their edges. This is repeated for $k$ times, and it is followed by a final graph exploration starting from $v$. Therefore, since $\lambda(v,s)\geq k$, by \Cref{lemma:reversing_paths} it is easy to see that if the set $S$ returned by \Cref{algorithm:randomizedlocalsearch} is not empty, then it coincides with $M(v)$. Furthermore, if $\lambda(v,s)=k$ and all of the paths reversed during the execution of the \textbf{for} loop in Line~\ref{line:randomized_local_for} end outside of $M(v)$, then \Cref{lemma:reversing_paths} implies that the final exploration in Line~\ref{line:randomized_local_search} will provide $M(v)$.

Thus, we only need to show that, if $\lambda(v,s)=k$ and $\mathit{vol}(M(v))\leq\Delta$, then, with probability at least $1/2$, all of the paths reversed during the execution of the \textbf{for} loop in Line~\ref{line:randomized_local_for} have the property that they end outside of $M(v)$.

Let $P_1,\dots,P_k$ be the paths discovered by the \textbf{for} loop in Line~\ref{line:randomized_local_for}, and let $G_0,\dots,G_{k-1}$ be the graph instances on which they were found. (I.e., we have $G_0=G$, and $G_i$, for $i\in\{1,\dots,k-1\}$, is the graph that is derived from $G_{i-1}$ by reversing the direction of the edges of $P_i$.) Then, since $\mathit{vol}(M(v))\leq\Delta$, a repeated application of \Cref{lemma:reverse} shows that the volume of $M(v)$ in all of the graphs $G_1,\dots,G_{k-1}$ is at most $\Delta$. 

Now, whenever the exploration in Line~\ref{line:randomized_local_first_search} discovers $s$, it is certain that the path reversed in Line~\ref{line:randomized_local_reverse_s}  ends outside of $M(v)$. Otherwise, if the exploration in Line~\ref{line:randomized_local_first_search} explores $2k\Delta$ edges, we have that at least a fraction of $(2k\Delta-\Delta)/(2k\Delta)=1-(1/2k)$ of them have their tail outside of $M(v)$. Therefore, the probability that the sampled edge in Line~\ref{line:randomized_local_sample} has its tail in $M(v)$ is at most $1/(2k)$. Thus, by the union bound, the probability that this was the case at least once during the $k$ iterations of the \textbf{for} loop in Line~\ref{line:randomized_local_for} is at most $1/2$.

Therefore, the probability that all of the paths reversed during the execution of the \textbf{for} loop in Line~\ref{line:randomized_local_for} have their end outside of $M(v)$ is at least $1/2$.
\end{proof}

\begin{algorithm}[h!]
\caption{$\mathtt{RandomizedLocalSearchForMSet}(G,v,s,k,\Delta)$}
\label{algorithm:randomizedlocalsearch}
\LinesNumbered
\DontPrintSemicolon
\ForEach{$i\in\{1,\dots,k\}$}{
\label{line:randomized_local_for}
  perform a graph exploration, in a BFS manner, starting from $v$, but stop once $2k\Delta$ edges have been explored, or if $s$ is discovered\;
  \label{line:randomized_local_first_search}
  \If{$s$ was discovered}{
    reverse the direction of the edges of the BFS-tree path from $v$ to $s$\;
    \label{line:randomized_local_reverse_s}
  }
  \Else{
    sample an edge $(x,y)$, uniformly at random, from the set of explored edges\; 
    \label{line:randomized_local_sample}
    reverse the direction of the edges of the BFS-tree path from $v$ to $x$\;
    \label{line:randomized_local_reverse_x}
  }
}
$S\leftarrow\emptyset$\;
explore at most $\Delta+1$ edges, starting from $v$, in a BFS manner\;
\label{line:randomized_local_search}
\If{at most $\Delta$ edges were explored, and $s$ was not discovered}{
  $S\leftarrow$ the set of discovered vertices\;
}
undo all the edge reversals on $G$\;
\Return $S$\;
\end{algorithm}

The success probability of \Cref{algorithm:randomizedlocalsearch} can be amplified to at least $1-\delta$, for any $\delta\in(0,1/2)$, by repeatedly applying it for $\lceil \log_2(1/\delta)\rceil$ times.

%% file: extending_k_plus_3.tex
\section{Extending to $(k+3)$-edge-connected components} \label{sec:extensions}
As already explained, there are two main obstacles that we have to overcome following our general framework in order to compute the $k$-edge-connected components, for any fixed $k>4$, within the same running time. 
However, we would like to understand the limitations of our approach for computing the $(k+3)$-edge-connected components in a $k$-edge-connected graph. 
Towards this direction, we propose a refinement partition of a $k$-edge-connected digraph that computes the $(k+3)$-edge-connected components which are contained within a given $(k+2)$-edge-connected-component. 

Observe that the given assumptions and relaxations drop the need for a fast  decomposition into $(k+2)$-edge-connected components and concentrate on the existence of efficient algorithm for computing a good partition of a vertex $v$ with $\lambda(v,s) \leq k+2$. 
More precisely, we show that our framework can be adopted for computing a corresponding \emph{good $(k+3)$-partition} (i.e., a partition $\mathcal{P}$ for a vertex $v$ with $\lambda(v,s) \leq k+2$ according to \Cref{definition:goodpartition} with the difference that $\mathcal{P}$ maintains the $(k+3)$-edge-connected components and separates ordinary vertices from $s$ with respect to minimum $(k+2)$-out sets) 
by exploiting the existing linear-time algorithms for computing the $k'$-edge-connected components with $k' \in \{1,2,3\}$ of general digraphs \cite{georgiadis20162edge,GIP20:SICOMP,Linear3ECC}.

Let $G$ be a $k$-edge connected graph, with $n$ vertices and $m$ edges where some of its vertices are designated as ``ordinary''. 
The ordinary vertices of $G$ are $(k + 2)$-edge-connected. 
Our goal is to compute a partition $\mathcal{P}$ of $V(G)$ such that two ordinary vertices are $(k+3)$-edge-connected if they belong to the same set in $\mathcal{P}$.
Our result is stated as follows. 

\begin{theorem}\label{theo:extendingkplusthree}
Let $\delta$ be a parameter with $0<\delta<1$. 
There exists a randomized algorithm that runs in $O(k^3 m \sqrt{n} \log(n/\delta))$ time, where $n$ is the number of the ordinary vertices of $G$, and outputs a partition $\mathcal{P}$ of $V(G)$ satisfying the following guarantees:
\begin{itemize}
\item{$\mathcal{P}$ maintains the $(k+3)$-edge-connected components of the ordinary vertices of $G$. I.e., every two $(k+3)$-edge-connected ordinary vertices of $G$ are in the same set from $\mathcal{P}$.}
\item{
With probability at least $1-\delta$, every two ordinary vertices of $G$ that are not $(k+3)$-edge-connected are separated by $\mathcal{P}$.}
\end{itemize}
\end{theorem}

Let us briefly explain the idea. 
We show how to extend the algorithm given in \Cref{theorem:main} in which we avoid repeated arguments regarding correctness and running time analysis. 
Fix an ordinary vertex $s$. 
For any ordinary vertex $v \neq s$, we let $M(v)$ denote the (inclusion-wise) minimum $(k+2)$-out set such that $v \in M(v)$ and $s \notin M(v)$. 
We begin by computing the $M$-sets with volume at most $m/\sqrt{n}$. In particular, 
for every ordinary vertex $v\neq s$, we invoke
$\mathtt{RandomizedLocalSearchForMSet}(G,v,s,k+2,\Delta)$
 to compute a partition $\mathcal{P}_1$ that maintains the $(k+3)$-edge-connected components of ordinary vertices having small $M$-set in $O(m\sqrt{n})$ time.  
By \Cref{proposition:randomized_local_search} and the corresponding discussion in \Cref{theorem:main}, we conclude that $\mathcal{P}_1$ contains two ordinary vertices $v,v'$ within the same set if $v,v'$ are $(k+3)$-edge-connected and both $M(v), M(v')$ have volume at most $m/\sqrt{n}$. We note that handling small $M$-sets can be done within the same running time for any $k$-out set of a general digraph.  

Next we consider large $M$-sets. As explained in the algorithm that establishes \Cref{theorem:main}, by sampling $\tilde{O}(\sqrt{n})$ edges we get the following: 
for every $M$-set $U$ with volume more than $m/\sqrt{n}$, we have sampled at least one edge whose tail $v$ is in $U$. Now our task is to compute a good $(k+3)$-partition for $v$.  
Notice that if $\lambda(v,s) > k+2$ then there is no $(k+2)$-out set that separates $v$ and $s$. Also recall that any $(k + 3)$-edge-connected component containing $v$ is a subset of $M(v)$.
We assume henceforth that $\lambda(v,s) \leq k+2$. Since $G$ is $k$-edge connected, we have $k \leq \lambda(v,s) \leq k+2$. 

Let $S$ be the latest $(v,s)$-mincut and let $\lambda = \lambda(v,s)$ (so that $k \leq \lambda \leq k+2$). Build $G'$ by contracting $S$ into a new vertex $z$ having $\lambda$ outgoing edges $e_1, \ldots, e_{\lambda}$ and build $\lambda$ new graphs $G'_i$ that are obtained from $G'$ by contracting $z$ with an edge $e_i$ into a new vertex $z_i$. 
In order to compute a good $(k+3)$-partition for $v$, we distinguish the following cases:
\setlist[description]{font=\normalfont\space}
\begin{description}
    \item[1. $\lambda(v,s) = k+2$] \ We build the Picard-Queyranne graph (\Cref{definition:pq}) $PQ$ of $G$ and compute a corresponding partition $\mathcal{P}$ that corresponds to the strongly connected components of $PQ$. In a similar fashion as given in the proof of \Cref{proposition:13_2}, the two properties of the good $(k+3)$-partition are fulfilled. 
    \item[2. $\lambda(v,s) = k+1$] \ Observe the following: (i) for any two vertices $x,y \notin S$ such that $\lambda_{G}(x,y) \geq k+3$, we know that $x,y$ belong to the same $2$-edge-connected component of $G' \setminus \{e_1, \ldots, e_{k+1}\}$ and (ii) by \Cref{lemma:include_outgoing_edge} we have $\lambda_{G'_i}(z_i,s) \geq k+2$. Thus it suffices to compute a refinement of the following partitions (with the same argumentation given in \Cref{section:proof_of_prop_main_auxiliary}): 
    \begin{itemize}
        \item[(2.i)] the 2-edge-connected components of $G' \setminus \{e_1, \ldots, e_{k+1}\}$ and
        \item[(2.ii)] the SCCs of the $PQ$ graph of $G'_i$, whenever $\lambda_{G'_i}(z_i,s) = k+2$ (reduce to case 1). 
    \end{itemize}
    \item[3. $\lambda(v,s) = k$] \ Observe the following: (i) for any two vertices $x,y \notin S$ such that $\lambda_{G}(x,y) \geq k+3$, we know that $x,y$ belong to the same $3$-edge-connected component of $G' \setminus \{e_1, \ldots, e_{k}\}$ and (ii) by \Cref{lemma:include_outgoing_edge} we have $\lambda_{G'_i}(z_i,s) \geq k+1$. Thus it suffices to compute a refinement of the following partitions:
    \begin{itemize}
        \item[(3.i)]  the 3-edge-connected components of $G' \setminus \{e_1, \ldots, e_{k}\}$, 
        \item[(3.ii)] the SCCs of the $PQ$ graph of $G'_i$, whenever $\lambda_{G'_i}(z_i,s) = k+2$ (reduce to case 1), and 
        \item[(3.iii)] the 2-edge-connected components of the graph obtained from $G'_i$ after removing all outgoing edges of $z_i$ (2.i) and the SCC's of each $PQ$ graph obtained from $G'_i$ after contracting an outgoing edge of $z_i$ (2.ii), whenever $\lambda_{G'_i}(z_i,s) = k+1$ (reduce to case 2).
    \end{itemize}
\end{description}
It is not difficult to see that the number of all computed graphs are in total $O(k^2)$ and building a $PQ$ graph requires $O(km)$ time. Thus, by the linear-time algorithms for computing the $k'$-edge-connected components with $k' \in \{1,2,3\}$ of general digraphs \cite{georgiadis20162edge,GIP20:SICOMP,Linear3ECC}, we conclude that finding a good $(k+3)$-partition for $v$ can be carried out in $O(k^3 m)$ time. After collecting all such partitions from each sampled edge into $\mathcal{P}_2$, we compute the common refinement of $\mathcal{P}_1$ (small $M$-sets) and $\mathcal{P}_2$ (large $M$-sets). 

Although it is tempting to use \Cref{theo:extendingkplusthree} for computing all $(k+3)$-edge-connected components of $G$, we should make clear that we only claim the following:  
given a $(k+2)$-edge-connected component $S$ of a $k$-edge-connected graph $G$, the algorithm given in \Cref{theo:extendingkplusthree} (by labeling the vertices of $S$ as ordinary) can compute, within the claimed running time, a partition $\mathcal{P}$ of $V(G)$ such that (i) every two $(k+3)$-edge-connected vertices of $S$ are in the same set of $\mathcal{P}$ and (ii) with high probability every two vertices of $S$ that are not $(k+3)$-edge-connected are separated by $\mathcal{P}$.
Nevertheless, we believe that the previous analysis highlights certain obstacles for extending our framework towards higher edge-connectivity.  


\paragraph*{Acknowledgment.} We thank an anonymous reviewer for suggesting the use of the Picard-Queyranne graph in the proof of \Cref{proposition:main_ordinary}. \sidenote{Loukas: Added this acknowledgement.}